\documentclass{llncs}

\usepackage{hyperref}
\usepackage{bm}
\usepackage{amsmath}
\usepackage{url}
\usepackage{amssymb}
\usepackage{stmaryrd}
\expandafter\def\csname opt@stmaryrd.sty\endcsname
{only,shortleftarrow,shortrightarrow}

\usepackage{cmll}
\usepackage{proof}
\usepackage{ifdraft}
\usepackage{tikz-cd}
\usepackage{mathtools}
\usepackage{bussproofs}
\usepackage{esint}
\usepackage{lscape}
\usepackage{string-diagrams}

\usepackage{extpfeil}
\usepackage[inline]{enumitem}

\usepackage[bibliography=common]{apxproof}

 \newtheoremrep{theorem}[theorm]{Theorem}
   \newtheoremrep{lemma}[theorem]{Lemma}
  \newtheoremrep{definition}[theorem]{Definition}
   \newtheoremrep{proposition}[theorem]{Proposition}
   \newtheoremrep{corollary}[theorem]{Corollary}
    \newtheoremrep{example}[theorem]{Example}
    \newtheoremrep{remark}[theorem]{Remark}
   \newtheoremrep{notation}[theorem]{Notation}
\newcommand{\NN}{\mathbb{N}}
\newcommand{\RR}{\mathbb{R}}
\newcommand{\card}[1]{\text{card}(#1)}

\newcommand{\Bool}{\mathbf{Bool}}
\newcommand{\deux}{\mathbf 2}

\newcommand*{\RRpos}{\mathbb{R}_{\ge 0}}

\renewcommand*{\le}{\leqslant}
\renewcommand*{\ge}{\geqslant}
\newcommand*{\set}[1]{\left\{ #1 \right\}}

\newcommand*{\suchthat}{\;\middle|\;}
\newcommand*{\setcompr}{\;;\;}
\newcommand*{\defeq}{\coloneq}
\renewcommand*{\phi}{\varphi}
\renewcommand*{\epsilon}{\varepsilon}

\newcommand*{\dif}[1]{\mathop{d #1}}
\newcommand{\interval}[2]{[#1,#2]}
\newcommand*{\scalar}[2]{\langle #1, #2 \rangle}

\DeclareMathOperator{\ic}{ic}

\DeclareMathOperator{\Path}{Path}
\DeclareMathOperator{\FMeas}{FMeas}

\newcommand{\pcoh}{\mathbf{PCoh}}

\newcommand{\Set}{\textbf{Set}}
\newcommand{\Catone}{\mathbf{C}}
\newcommand{\Cattwo}{\mathbf{D}}
\newcommand{\Pol}{\mathbf{SBS}}

\newcommand{\Cones}{\mathbf{Cones}}
\newcommand{\Mcones}{\mathbf{MCones}}
\newcommand{\Icones}{\mathbf{ICones}}
\newcommand{\ICones}{\Icones}
\newcommand{\MCones}{\Mcones}
\newcommand{\MeasC}{\mathbf{Meas}}
\newcommand{\Stoch}{\mathbf{Stoch}}
\newcommand{\SubStoch}{\mathbf{SubStoch}}
\newcommand{\BorelStoch}{\mathbf{BorelStoch}}

\newcommand{\DistSubStoch}{\mathbf{SubStoch}^{\leq \omega}}

\newcommand{\Giry}{\mathcal G}

\newcommand{\wlim}{\lim^{w}}
\newcommand{\Meas}{\text{Meas}}
\newcommand{\norm}[1]{\lvert #1 \rvert}
\newcommand{\distrone}{\mu}
\newcommand{\distrtwo}{\nu}
\newcommand{\dirac}[1]{\{ #1\}^1}
\newcommand{\functorClin}{C}

\newcommand{\DD}{\text{DD}}
\newcommand{\dd}{\text{d}}
\newcommand{\DDn}[1]{\text{DD}_{#1}}

\newcommand{\mtest}[1]{\mathcal M_{#1}}
\newcommand{\mtestc}[2]{\mathcal M_{#2}^{#1}}

\newcommand{\Arity}{\textbf{ar}}
\newcommand{\pathes}[2]{\text{Path}(#1,#2)}
\newcommand{\unitball}[1]{\mathcal B#1}
\newcommand{\ball}{\unitball}
\newcommand{\cint}{\fint}
\newcommand{\cintc}[1]{\fint^{#1}}

\newcommand{\multkern}[1]{\text{mulnom}[#1]}

\newcommand{\pcs}[1]{\mathcal{#1}} %probabilistic coherence space
\newcommand{\web}[1]{\left\vert#1\right\vert} %web of probabilistic coherence space
\newcommand{\clique}[1]{\mathrm{P}\left(#1\right)} %set of cliques in a pcoh

\newcommand{\Rp}[1]{({\mathbb R}_{\geq 0})^{#1}}
\newcommand{\multinomial}[1]{m(#1)} %multinomial coeff of a multiset (number of its enumerations)
\newcommand{\MfSet}[1]{\mathcal{M}_{\text{f}}^{\text{Set}}\left(#1\right)} %finite multiset
\newcommand{\Mn}[2]{\mathcal{M}_{#1}\left(#2\right)} %finite multiset
\newcommand{\MnSet}[2]{\mathcal{M}^{\text{Set}}_{#1}\left(#2\right)} %finite multiset
\newcommand{\MnCat}[3]{\mathcal{M}^{#1}_{#2}\left(#3\right)}

\newcommand{\pCoh}{\mathbf{pcoh}} % linear pCoh
\newcommand{\pCoht}{\mathbf{PCoht}} % linear pCoh with totality
\newcommand{\id}[1][]{\mathrm{id}^{#1}} % identity endomorphism of #1
\newcommand{\idm}[2]{\mathrm{id}^{#1}_{#2}}
 \newcommand\PCF{\ensuremath{\mathrm{PCF}}}%PPCF for math notations
        \newcommand\ttrue{\mathrm{t}}
        \newcommand\ffalse{\mathrm{f}}
        	\newcommand{\sem}[1]{\left\llbracket #1 \right\rrbracket}

\newcommand{\iotatc}{\text{incl}_{TC}}
\newcommand{\iotares}{\iota_{TC}}

 \newcommand{\unit}{\mathbf 1}
 \newcommand{\scal}[2]{\scalar {#1}{#2} }

 \newcommand{\eq}[1]{\text{eq}_{#1}}
 \newcommand{\coeq}[1]{\text{coeq}_{#1}}
 \newcommand{\law}[1]{\text{law}(#1)}
 \newcommand{\ob}[1]{\text{Ob}(#1)}
 \newcommand{\pointed}[1]{{#1}_{\bullet}}
 \newcommand{\pointA}{\overline A}
 \newcommand{\pointAun}{\overline {A_1}}
 \newcommand{\pointAdeux}{\overline {A_2}}
 \newcommand{\tc}[1]{\widehat{#1}}
 
 \newcommand{\Ar}{\Arity}

\title{Interpreting De Finetti's Theorem in the Category of Integrable Cones (long version)}

\author{Raphaëlle {Crubillé}}
  \institute{Aix Marseille Univ, CNRS, LIS, Marseille, France}

\authorrunning{R. Crubillé} %TODO mandatory. First: Use abbreviated first/middle names. Second (only in severe cases): Use first author plus 'et al.'

\begin{document}
\maketitle
\begin{abstract}
We establish a connection between two results in the literature on probabilistic semantics: a formulation of De Finetti's theorem in the language of category theory due to Jacobs and Staton, and the generic construction of the free exponential of Linear Logic by Melliès et al, that has been instantiated in the model of probabilistic coherence spaces by Crubillé et al. The structural proximity of these two constructions is manifest, but making this connection formal requires technical developements on the relationship between the category of stochastic kernels and the category of integrable cones, two well-known categories in probabilistic semantics.  We then use this connection to give a characterisation of the total elements of the probabilistic coherence space $!\Bool$.
  \end{abstract}
%\maketitle

There have been (at least) two different approaches pursued very actively in recent years to mix probability theory and category theory. The first one, around ideas from \emph{synthetic probability theory, Markov categories...}~\cite[\dots]{Cho_Jacobs_2019,FRITZ2020107239} looks at how to abstract the laws of probability theory at the level of categories, with the goal of understanding probabilistic phenomena in a more structural way. A standard endeavour here is to start with a concrete category that arises from probability theory -- such as the category of finite sets and stochastic matrices, or the category $\Stoch$ of measurable spaces and stochastic kernels, to express a standard theorem of probability theory using categorical language: functors, limits, adjunctions\ldots, and then to try to prove it abstractly
in a category where the requirements from synthetic probability theory hold, e.g. a Markov category.  

The other approach has emerged from the field of \emph{programming languages semantics}, with the aim of finding categories that would be good models for programming languages with \emph{probabilistic primitives}, e.g.~discrete or continuous probabilistic sampling, or Bayesian reasoning. The main problem here is to make the structure needed to interpret probabilistic primitives coexist with the other structural constraints of a model for a programming language: cartesian closure for higher-order languages such as the simply typed lambda-calculus, monoidal closure plus a well-behaved comonad for models of linear logic\dots Unfortunately, the categories that appear naturally from probability theory (e.g.~$\Stoch$) are not cartesian or monoidal closed, which has driven the construction of more involved categories as models of higher-order probabilistic computation~\cite[\dots]{heunen2017convenient,paquetcontinuous,ehrhard2025integration}. 

In the present work, we are interested in how both these approaches deal with the notion of \emph{exchangeability}. On the one hand, several works in synthetic probability theory have developed categorical understandings of De Finetti's theorem,
a foundational result that characterises the \emph{exchangeable} infinite sequences of random variables on some measurable space $X$. Here exchangeable means that the distribution of the sequence is invariant under permutations of finite prefixes.
An  infinite list of independent, identically distributed  random variables (noted i.i.d.-sequence in the following) is obviously exchangeable. 
De Finetti's theorem says that, under some conditions on $X$,
\emph{every} infinite exchangeable list of $X$-valued random variables can be written as a (unique) \emph{mixture} of i.i.d.-sequences.\footnote{De Finetti did not prove the modern form of De Finetti's theorem, since he did not accept the Kolmogorov axiomatisation of probability theory--see for instance~\cite{historic,fishburn1986axioms}. But this theorem is named after him because he proposed taking exchangeable sequences of Boolean random variables--definable in discrete probability theory, where there was no controversy on the formalism--as a foundation for probability theory on the continuous interval $[0,1]$.}

On the other hand, exchangeability also appears as a crucial requirement for the exponential comonad $!$ of Linear Logic, that represents data types that can be copied as many times as necessary. These requirements imply indeed that for every $A$, $!A$ is a \emph{commutative} comonoid, and this commutativity implies that $!A$ should be understood as a \emph{exchangeable} infinite list of copies of $A$.
In a symmetric monoidal closed category $\Catone$
taking $!A$ to be the free commutative comonoid over $A$, assuming that it exists for all $A$, automatically gives a model of (multiplicative, intuitionnistic) Linear Logic: this class of models is called \emph{Lafont} models.
In~\cite{metabasson}, a generic construction of the free exponential comonad of Linear Logic on a monoidal category is developed  as a way of unifying the construction done in several existing concrete Lafont models of Linear Logic. An example of this construction, as proved in~\cite{crubille2017free}, is the model called probabilistic coherence spaces ($\pcoh$), a fully abstract~\cite{EhrhardPT18} model of discrete probabilistic $\PCF$.

In the present paper, we build a close connection between the free exponential construction of Lafont models, and the categorical formulations of De Finetti's theorem~\cite{jacobs2020finetti,moss2022probability,fritz2021finetti}.
To do this, we take as ambient category the category of integrable cones, introduced by Ehrhard and Geoffroy~\cite{ehrhard2025integration} as a model of continuous probabilistic $\PCF$, and in which both $\Stoch$ and $\pcoh$ can be embedded.  The free exponential construction and the De Finetti construction are structurally very similar, but connecting them formally requires taking a close look at the relationship between $\Stoch$ and $\Icones$: that constitutes the technical core  of our work. In the last part of our paper, we present an application of  this connection to characterise the total elements of $!\Bool$ -- where $\Bool$ is the interpretation in $\pcoh$ of the Boolean type -- using the probability measures on the $[0,1]$ interval.

\section{Exchangeability in Monoidal Categories}\label{sect:exch_in_mon_cat}
\subsection{Categorical De Finetti theorems}\label{sect:DeFinetti}
Let us note as usual $\Giry$ for the Giry monad over the category $\MeasC$ of measurable spaces, thus $\Giry X$ is the set of probability distributions over $X$ equipped with a structure of measurable space. Under suitable conditions on $X$, De Finetti's theorem says that if $s$ is an exchangeable sequence of $X$-valued random variables, then there exists a unique probability measure $\mu \in \Giry\Giry X$ such that the \emph{law} of $s$--which is a probability measure on $X^\omega,$ the infinite product of copies of  $X$--is the \emph{mixture} along $\mu(d\alpha)$ of the law of $i.i.d$ sequences, each of law $\alpha \in \Giry X$. Seen as a probabilistic process, it is the probability distribution of the sequence of elements of $X$ obtained by first sampling from $\mu$ an element $\alpha \in \Giry X$, and then sampling infinitely often an element of $X$ into the probability distribution $\alpha \in \Giry x$.
Let's look for instance at $X = \deux:= \{0,1\}$: this result means that the probability measures on $\{0,1\}^\omega$ that verify the exchangeability requirement are in bijection with the set of Bernoulli distributions, that may be identified to the $[0,1]$ interval.

Jacobs and Staton presented in~\cite{jacobs2020finetti} a categorical  De Finetti's theorem. Their setting is the category $\Stoch$, where the objects are measurable spaces, and the morphisms are the stochastic kernels; $\Stoch$ can be described as 
 the Kleisli category of the Giry monad on $\MeasC$, and 
the cartesian product on $\MeasC$ gives rise to a symmetric monoidal product in $\Stoch$. 
Their formulation is based on the construction of a chain in $\Stoch$, that they call the \emph{draw-and-delete chain}. 
For any measurable space $X$, the objects in this chain are the $\Mn n X$ -- the spaces of multisets on $X$  of size exactly $n$ -- that represent \emph{urns} containing $n$ elements. The morphisms are the  $\DDn n:\Mn n X \rightarrow \Mn {n-1}X$ obtained by drawing at random an element in the urn, and then removing  it from the urn.

From a categorical point of view, this chain can also be built from the more primitive \emph{discard} chain, whose objects are the $X^{\otimes n}$, 
and where the discard morphism $\dd_n: X^{\otimes{n+1}} \rightarrow X^{\otimes n}$ arises from the fact that $\unit$ is also the terminal object. We can then obtain the draw-and-delete chain by observing that the $\Mn n X$ are the equalisers of the symmetries on $X^{\otimes n}$, and lifting the discard chain to the level of equalisers.\footnote{In our diagrams, we will systematically omit the canonical isomorphism $X \cong X \otimes \unit$.}
\begin{equation}
\begin{tikzpicture}[scale=0.8]
\small
  \node (0) at (0,0) {$\unit$};
  \node (s0) at (0,1){$\unit$};
  \node (1) at (3.5,0) {$\Mn 1 X$};
  \node (s1) at (3.5,1){$X$};
  \node (2) at (7,0) {$\Mn 2 X$};
  \node (s2) at (7,1){$X^{\otimes 2}$};
 \node (s3) at (9.5,1){$\dots$};
  \node (3) at (9.5,0) {$\dots$};

\draw[->] (0) to node[right] {$\eq 0$} (s0);
\draw[->] (1) to node[right] {$\eq 1$} (s1);
\draw[->] (2) to node[right] {$\eq 2$} (s2);

\draw[->] (3) to node[above] {$ \DDn 2 $} (2);
\draw[->] (2) to node[above] {$ \DDn 1 $} (1);
\draw[->] (1) to node[above] {$ \DDn 0 $} (0);
\draw[loop,->,min distance=5mm,looseness=4] (s1) to node[above]{symm} (s1);
\draw[loop,->,min distance=5mm,looseness=3.5] (s2) to node[above]{symm} (s2);
\draw[->] (s3) to node[above] {$X^{\otimes 2}\times(\cdot) $} (s2);
\draw[->] (s2) to node[above] {$ X \otimes (\cdot) $} (s1);
\draw[->] (s1) to node[above] {$ (\cdot) $} (s0);
\end{tikzpicture}
\end{equation}
The law of an exchangeable infinite sequence of random variables on $X$ can then be seen as a \emph{cone} from $\unit$ to the draw-and-delete chain: the morphism $1 \rightarrow \Mn n X$ is the law of the $n$ first random variables of the sequence. The fact that the targets are the equalisers of the symmetries corresponds to the exchangeability requirement. The formulation of De Finetti's theorem proposed by~\cite{jacobs2020finetti}
is that $\Giry X$ should be the \emph{limit} of the draw-and-delete chain. In particular, the universal property applies to cones from $\unit$, thus those arising from exchangeable sequences of $X$-valued random variables, which corresponds to the usual De Finetti theorem. 

A less stratified formalisation of De Finetti's theorem~\cite{moss2022probability,fritz2021finetti} has also been proposed by Perrone et al. This version requires the existence of the infinite product $X^\omega$, axiomatised categorically as the limit of the discard chain. Then instead of looking at the equalisers of symmetries at each layer, Perrone et al.'s formulation of De Finetti's theorem is that $\Giry X$ should be the equaliser of all finitely supported symmetries on $X^\omega$:
\[\begin{tikzpicture}
    \small
\node (a) at (0,0) {$\Giry X$};
\node (b) at (3,0){$X^{\omega}$};
\draw[->,loop,min distance=5mm,looseness=4.5] (b) to node[above] {finite symm} (b);
\draw[->] (a) to node[below] {} (b);
\end{tikzpicture}\]

While $X^\omega$ can always be defined as the infinite Cartesian product in $\MeasC$, its limit characterisation in $\Stoch$ is a consequence of Kolmogorov's extension theorem that does not hold in general for measurable spaces--but holds for all \emph{standard Borel spaces}, i.e.~Polish spaces equipped with their Borel $\sigma$-algebra. This formalisation of De Finetti's theorem has  been shown in~\cite{fritz2021finetti} in the general setting of a Markov category with a number of additional requirements, that hold in particular for the category $\Pol$ of  standard Borel spaces.
\subsection{The free exponential layered construction}\label{sect:layeredconstruction}
 Starting from a symmetric monoidal category $\mathcal C$, there are several possible axiomatisations for $\mathcal C$ to be able to model the exponential modality of Linear Logic, thus to be a model of (Intuitionistic) Linear Logic -- see~\cite{mellies2009categorical} for a survey on this. 
The most demanding one, proposed by Lafont, consists in requiring the existence of the free commutative comonoid generated by $A$--noted $!A$--for every object $A$ of the category. When it is the case, $!$ can be extended into a comonad that verifies all requirements for a categorical model of Linear Logic.
A layered construction is given in~\cite{metabasson} for the free exponential modality, that applies to many--but not all-- known Lafont models of linear logic. Firstly, it asks for the existence, for every $A$, of the \emph{free copointed object} $(A_\bullet ,w:A_\bullet\rightarrow \unit) $ over $\unit$ generated by $A$, where $\unit$ is the monoidal unit.
Secondly, it requires the existence of the equalisers of the symmetries on ${A_\bullet}^{\otimes n}$ for every $n$. This leads to the \emph{approximations chain}:
\begin{equation*}
  \begin{tikzpicture}[scale=0.8]
\small
  \node (0) at (0,0) {$\unit$};
  \node (s0) at (0,1){$\unit$};
  \node (1) at (3.5,0) {$A_\bullet^{\leq 1}$};
  \node (s1) at (3.5,1){$\pointed A$};
  \node (2) at (7,0) {$A_\bullet^{\leq 2}$};
  \node (s2) at (7,1){$\pointed A^{\otimes 2}$};
 \node (s3) at (9.5,1){$\dots$};
  \node (3) at (9.5,0) {$\dots$};
\draw[->] (0) to node[right] {$\eq 0$} (s0);
\draw[->] (1) to node[right] {$\eq 1$} (s1);
\draw[->] (2) to node[right] {$\eq 2$} (s2);
\draw[->] (3) to node[above] {} (2);
\draw[->] (2) to node[above] {} (1);
\draw[->] (1) to node[above] {} (0);
\draw[loop,->,min distance=5mm,looseness=4.5] (s1) to node[above]{symm} (s1);
\draw[loop,->,min distance=5mm,looseness=4] (s2) to node[above]{symm} (s2);
\draw[->] (s3) to node[above] {$\pointed A^{\otimes 2}\times w$} (s2);
\draw[->] (s2) to node[above] {$ \pointed A \otimes w$} (s1);
\draw[->] (s1) to node[above] {$ w$} (s0);
\end{tikzpicture}
\end{equation*}
As proved in~\cite{metabasson}, the limit of this chain--when it exists, and when moreover both the equalisers of symmetries and the chain limit commute with the tensor product--is the free commutative comonoid generated by $A$, and thus gives a Lafont model of Linear Logic. It has been proved in~\cite{crubille2017free} by this method that the exponential modality of probabilistic coherence spaces is indeed the free one.

\subsection{Our contributions in the present paper} 
From a formal point of view, there is a striking similarity between the draw-and-delete chain from Section~\ref{sect:DeFinetti} and the  approximations chain of Section~\ref{sect:layeredconstruction}, that reflects the fact that the question of \emph{exchangeability} -- and how to manage symmetries between different copies of a resource -- is fundamental in Linear Logic semantics. Technically, the only difference between the two chain constructions is that $\unit$ is not required to be the terminal object in the free exponential construction from~\cite{metabasson}, and accordingly  $X^{\otimes n}$ in the draw-and-delete chain is replaced by ${X_\bullet}^{\otimes n}$ in the construction of the approximations chain. 

A natural place to explore this connection is the category of \emph{integrable cones} ($\Icones$)~\cite{ehrhard2025integration} because both the category of stochastic kernels and the category of probabilistic coherent spaces can be faithfully embedded there. Let's write $\sem{\cdot}^{\Stoch}$ and $\sem{\cdot}^{\pcoh}$ for those respective embeddings.
We can sum up the main contributions of our work as:
\begin{enumerate}
\item We show that for every standard Borel space $X$, $\Giry X$ is the \emph{free commutative comonoid} generated by $X$ in the category $\Stoch(\Pol)$ --  the full sub-category of $\Stoch$ whose objects are standard Borel spaces.
\item We show that in  $\Icones$ the draw-and-delete chain from the De Finetti formulation can be built  over the image $\sem{X}^{\Stoch}$ of  any standard Borel space $X$; that the approximations chain from the free exponential construction can be built over the image $\sem{\pcs A}^{\pcoh}$ of any probabilistic coherence space $\pcs A$; and that these chains have indeed for limit there $\sem{\Giry X}^{\Stoch}$ and $\sem{!\pcs A}^{\pcoh}$ respectively.
In order to prove that $\sem{\Giry X}^{\Stoch}$ is the limit of the draw-and-delete chain  in $\Icones$, we first show that the faithful functor from $\Stoch$ to $\Icones$ preserves both the tensor product and \emph{connected} limits, which is a relevant result in itself, and also the more technically involved part of our work.
\item A countable set $X$  can be represented both as a discrete measurable space in $\Stoch$, and as a probabilistic coherence space in $\pcoh$, thus both chains exist with respect to $X$ in the category $\Icones$. For such an $X$,  we build a chain morphism that by the universal property of the limit gives rise to a canonical monomorphism $\sem{\Giry X}^{\Stoch} \xrightarrow{\iota} \sem{!X}^{\pcoh}$, that can be understood as mapping any measure in $\Giry X$ a \emph{promotion} in $!X$.
That is the central result in our paper: it says that the (laws of) infinite sequences of exchangeable random variables on $X$ can be seen as elements of $!X$.
\begin{center}
\begin{tikzpicture}[scale=1]
    \small
  \node (0) at (-2.1,4.2) {$\unit$};
  \node (s0) at (-2.1,3.5){$\unit$};
  \node (1) at (0,4.2) {$\Mn n X$};
  \node (s1) at (0,3.5){$X^{\otimes n}$};
  \node (2) at (3,4.2) {$\Mn {n+1} X$};
  \node (s2) at (3,3.5){$X^{\otimes {n+1}}$};
 \node (s3) at (5.5,3.5){$\dots$};
  \node (3) at (5.5,4.2) {$\dots$};
 \node (as3) at (-1.7,3.5){$\dots$};
  \node (a3) at (-1.7,4.2) {$\dots$};
 
\node (limit1) at (6.3,5) {$\Giry X$};
\node (limit2) at (6.3,0.7) {$! X$};

\draw[->] (0) to node[right] {} (s0);
\draw[->] (1) to node[right] {} (s1);
\draw[->] (2) to node[right] {} (s2);
\draw[->,dashed] (limit1) to node[left] {$\exists ! \iota$} (limit2);
\draw[->, bend right=9] (limit1) to node[right] {} (2);
\draw[->,bend right=11] (limit1) to node[right] {} (1);
\draw[->] (3) to node[above] {$ \DDn n $} (2);
\draw[->] (2) to node[above] {$ \DDn {n+1} $} (1);
\draw[->] (1) to node[above] {} (a3);
\draw[->] (s3) to node[above] {$X^{\otimes {n+1}}\times(\cdot) $} (s2);
\draw[->] (s2) to node[above] {$ X^{\otimes n}\otimes (\cdot) $} (s1);
\draw[->] (s1) to node[above] {} (as3);

    %%%%%%%%%%%%%%%%%%%%%%%%%%%%%%%%%%
\node (eq) at (0,1.3) {$ X^{\leq n}$};
\node (tens1) at (0,2) {$( X_\bullet)^{\otimes n}$};
\node (tens2) at (3,2) {$(X_\bullet)^{\otimes n+1}$};
\node (t0) at (-1.7,1.3){\ldots};
\node (t1) at (5,1.3){\ldots};
\node (u0) at (-1.7,2){\ldots};
\node (u1) at (5,2){\ldots};
\node(un) at (-2.1,1.3){$\unit$};
\node(unbis) at (-2.1,2){$\unit$};
\draw[->] (un) to node{} (unbis);
\node (c) at (3,1.3) {$ X^{\leq n+1}$};
\draw[->] (eq) to node [left] {} (tens1);
\draw[->] (c) to node [left] {} (tens2);
\draw[->] (tens2) to node [above] {$(X_\bullet)^{\otimes n}\otimes{\pi_2}$} (tens1);
\draw[->,dashed] (s1) to node [above] {} (tens1);
\draw[->,dashed] (s2) to node [above] {} (tens2);
\draw[->] (c) to node [above] {} (eq);
\draw[->] (eq) to node [above] {} (t0);
\draw[->] (t1) to node [above] {} (c);
\draw[->] (tens1) to node [above] {} (u0);
\draw[->] (u1) to node [above] {} (tens2);
\draw[->, bend left=7] (limit2) to node[right] {} (eq);
\draw[->, bend left=7] (limit2) to node[right] {} (c);
\end{tikzpicture}

\end{center}

\item Finally,
we show that for any discrete countable space $X$, all \emph{total} elements -- in the sense of the theory of probabilistic coherence spaces with totality~\cite[Appendix A]{ehrhard2025variable} in $!X$ are in the image of this morphism $\Giry X \rightarrow !X$, thus that the total elements in $!X$ are \emph{exactly} the probabilistic \emph{mixtures} of promotions.
  \end{enumerate}

\section{Exchangeability in monoidal categories}\label{sect:connection_chains}

In this section,  we fix  $\Catone$  any symmetric monoidal category, and we abstract the common pattern in the two chain constructions of Section~\ref{sect:exch_in_mon_cat} by defining generic draw-and-delete chains from any $\unit$-copointed object. Then, we show at this level of generality that morphisms between $\unit$-copointed objects can be lifted to chain-morphisms between the corresponding draw-and-delete chains, and finally we use this general framework to build in $\pcoh$ a chain morphism from the De Finetti chain to the approximations chain.
\subsection{ Draw-and-Delete chains in a generic category}

 \begin{toappendix}
 \subsection{Connecting Draw-and-Delete chains in a generic category}
 \end{toappendix}

   We will consider $\unit$-copointed objects in $\Catone$, i.e.~pairs $ (A,w:A \rightarrow \unit)$. We first define what is a \emph{free copointed object generated by $A$}. 
  \begin{definition}
    Let $A$ an object in $\Catone$. A \emph{free copointed object generated by $A$} is a copointed object $(B,w)$ equipped with a morphism $\eta_B:B \rightarrow A$, such that for any copointed object $(C,u)$ and morphism $\eta_C:C\rightarrow A$, there exists a unique morphism of co-pointed object $f: (C,u) \rightarrow (B,w)$ such that $\eta_C = \eta_B \circ f$:
   $$ \begin{tikzcd}
      & B \ar[ld,"\eta_B ",swap] \ar[rd,"w"] \\
      A & C \ar[r,"u ",swap] \ar[l,"\eta_C"] \ar[dashed,u,"\exists ! f "] & \unit \\
      \end{tikzcd}$$
      
      When it exists, we note $A_\bullet$ the free copointed object generated by $A$.
    \end{definition}
\begin{remark}
Observe that if $\mathcal C$ has cartesian products, $A_\bullet$ is simply $A \with \unit$.
On the other hand, if $\unit$ is also the terminal object -- as for instance in $\Stoch$ -- the free copointed object generated by $A$ is simply $(A,A\rightarrow \unit)$, where $A\rightarrow \unit$ is the unique morphism from $A$ to $\unit$.
\end{remark}
    Given any co-pointed object $(A,w)$, we can define the \emph{draw-and-delete chain} generated by $A$ as long as for every $n \in \NN$, the equaliser of all symmetries on $A^{\otimes n}$ exists.  
\begin{definition}\label{def:dd_chain}
  Let $\pointA = (A,w:A \rightarrow \unit)$ a $\unit$-copointed object in $\Catone$.
  We say that a \emph{draw-and-delete (DD) chain can be built} over $\pointA$ when for every $n \in \NN$, the equaliser of the $n!$ symmetries over $A^{\otimes n}$ exists. When it is the case, we note $\Mn n A$ for this equaliser, and we call \emph{draw-and-delete morphism} the unique morphism  $\DDn n^{\pointA}$ given by:
    $$\begin{tikzcd}
    A^{\otimes n} & & A^{\otimes n+1} \\
    \Mn n A & & \Mn {n+1} A
    \ar[from=1-3,to=1-1,"A^{\otimes n} \otimes w"]
    \ar[from=2-3,to=2-1,"\DDn {n}^{\pointA}"]
    \ar[from=2-1,to=1-1,"\eq n"]
    \ar[from=2-3,to=1-3,"\eq {n+1}"]
  \end{tikzcd}$$
  We call \emph{draw-and-delete chain over $\pointA$} the chain:  $$\unit \xleftarrow{\DDn 0^{\pointA}} \Mn 1 A \xleftarrow{\DDn 1^{\pointA}} \Mn 2 A \xleftarrow{\DDn 2^{\pointA}}\ldots.$$

\end{definition}

\begin{example}
  The draw-and-delete chain built by Jacobs and Staton in~\cite{jacobs2020finetti} -- that we'll call in the following the De Finetti  chain, and note $\DD^{\text{DF}}$ -- is an instance of our generic construction of draw-and-delete chains, by taking $\Catone = \Stoch$, and as copointed object $(X, \_:X \rightarrow \unit)$.
  In any category $\Catone$ which is a model of intuitionistic multiplicative  Linear Logic such that moreover the exponential is built using the layered construction of~\cite{metabasson}, the approximations chain for $A$  there is an instance of our draw-and-delete chain, by taking as copointed object $\pointed A$ the free copointed object over $A$.
  \end{example}

    As mentioned earlier, the layered construction of~\cite{metabasson} asks for \emph{commutation} of the equalisers with the monoidal product. We spell out what it means below:
    \begin{definition}
      Let $A$ be an object in $\Catone$, such that for every $n$, the equaliser $(\Mn n A,\text{eq}_n:\Mn n A \rightarrow A^{\otimes n})$ of the symmetries over $A^{\otimes n}$ exists. We say that the \emph{equalisers of symmetries commute with the tensor product} when $\forall Y$ in $\Catone$, $\forall n \in \NN$:
      $$(\Mn n A \otimes Y, \text{eq}_n \otimes Y:\Mn n A \otimes Y \rightarrow A^{\otimes n}\otimes Y)$$ is the equaliser of the $\{\sigma \otimes Y \mid \sigma \text{ a symmetry on } A^{\otimes n}\}$ -- where by abuse of notation we note $Y$ for the morphism $\id_Y$.
    \end{definition}
    We now formally state the central result of~\cite{metabasson}, i.e. that -- under some conditions on commutativity of limits and monoidal products -- the limit of the draw-and-delete chain over $A_\bullet$  is the free comonoid over $A$. 
    \begin{theorem}[from~\cite{metabasson}]\label{th:metabasson}
      Let $A$ be an object in $\Catone$ such that:
      \begin{enumerate}
        \item the free copointed object $A_\bullet$ generated by $A$ exists,
        \item $\forall n \in \NN$, the equaliser of symmetries on $A^{\otimes n}$ exists and commutes with the tensor product,
        \item the limit $(!A,\rho_{\infty,n}:!A \rightarrow \Mn n A)$ of the draw-and-delete chain generated by $A_\bullet$ exists, and moreover it commutes with the tensor product -- in the sense where for any $B$ in $\Catone$, $(!A\otimes B, \rho_{\infty,n} \otimes B: ! A \otimes B \rightarrow \Mn n A \otimes B)_{n \in \NN}$ is a limit cone for the chain $\DD_n^{A_\bullet} \otimes B$,
          \end{enumerate}
then $!A$ is the free commutative comonoid generated by $A$.
      \end{theorem}

    Any morphism of copointed object $\overline A \rightarrow \overline B$ induces a chain morphism between the draw-and-delete chains generated by $\overline A$ and $\overline B$ respectively -- when these chains exist. We present the construction in Definition~\ref{def:mn_alpha} and Proposition~\ref{th:chain_morphism_from_copointed_morphisms} below.
\begin{definition}\label{def:mn_alpha}

   Let $A_1, A_2$ two objects of $\Catone$, such that their respective equalisers of symmetries exist.
   Let $\alpha:A_1 \rightarrow A_2$ a morphism.
   Using the fact that  $ \alpha^{\otimes} \circ eq_n^{A_1}:\Mn n {A_1}  \rightarrow {A_2}^{\otimes n} $ equalises all symmetries over ${A_2}^{\otimes n}$, we define  $\Mn n \alpha$ for the unique morphism such that:
   $$\begin{tikzpicture}
 \node (a) at (0,0) {$ \Mn n {A_1}$};
 \node (b) at (3,0){$A^{\otimes n}$};
 \node (c) at (6,0){${A_2}^{\otimes n}$};
   \node (d) at (6,-1){$\Mn n {A_2}$};
   \draw[->] (a) to node [above] {$\eq n ^{A_1} $} (b);
   \draw[->] (b) to node [above] {$\alpha^{\otimes n}$} (c);
   \draw[->] (d) to node [right]{$\eq n ^{A_2}$}(c);
   \draw[->] (a) to node [below] {$\Mn n \alpha$}(d);
     \end{tikzpicture}$$
  \end{definition}
 When $\alpha$ is also a morphism of copointed objects $\pointAun \rightarrow \pointAdeux$, the $(\Mn n \alpha)_{n \in \NN}$ form a chain morphism from the draw and delete chain on $\pointAun$ to the draw and delete chain on $\pointAdeux$. We prove this using the uniqueness property that comes from the definition of the $\Mn n \alpha$. The proof, which is a standard application of a universal property, can be found in Appendix.

 \begin{propositionrep}\label{th:chain_morphism_from_copointed_morphisms}
  Let $\pointAun= (A_1,w_1)$ and $\pointAdeux = (A_2,w_2)$ two copointed objects.
  Let $\alpha:\pointAun \rightarrow \pointAdeux$ a morphism of copointed objects, i.e. $\alpha:A_1 \rightarrow A_2$ is such that  $w_1 = w_2 \circ \alpha$.
    Then the $(\Mn n \alpha)_{n \in \NN}$ form a chain morphism between the corresponding draw-and-delete chains i.e.~the diagram below commutes for all $ n \in \NN$:
 
       $$\begin{tikzpicture}
 \node (a) at (0,0) {$ \Mn n {A_1}$};
 \node (b) at (3,0){$\Mn {n+1} {A_1}$};
 \node (c) at (0,-1.2){${\Mn n {A_2}}$};
 \node (d) at (3,-1.2){$ \Mn {n+1}{A_2}$};
   \draw[->] (b) to node [above] {$\DDn{n}^{\pointAun} $} (a);
   \draw[->] (a) to node [left] {$\Mn n  \alpha$} (c);
   \draw[->] (b) to node [right]{$\Mn {n+1} \alpha$}(d);
   \draw[->] (d) to node [below] {$\DDn{n}^{\pointAdeux} $}(c);
     \end{tikzpicture}$$
   \end{propositionrep}
   \begin{proof}
     The proof is a standard application of a universal property.
     Our goal is to show that $ \Mn n \alpha \circ \DDn n^{\pointAun} = \DDn n^{\pointAdeux} \circ \Mn {n+1} \alpha  $. First, observe that since $\Mn n {A_2}$ is defined as an equaliser, we can rewrite equivalently our goal as:
     \begin{equation}\label{eq:proof_chain_morph_1}
\eq n^{A_2} \circ \Mn n \alpha \circ \DDn n^{\pointAun} = \eq n^{A_2} \circ \DDn n^{\pointAdeux} \circ \Mn {n+1} \alpha \quad : \Mn {n+1} {A_1} \rightarrow \Mn n {A_2}. 
\end{equation}
Let us note $\dd_n^{\overline A_i}= A_i^{\otimes n}\otimes w_i$. From the definition of $\DDn n^{\pointAdeux}$, we know that  $\eq n^{A_2} \circ \DDn n^{\pointAdeux} = \dd_n^{\pointAdeux} \circ \eq{n+1}^{A_2}$. We can thus rewrite our goal~\eqref{eq:proof_chain_morph_1} as:
 \begin{equation}\label{eq:proof_chain_morph_2}
\eq n^{A_2} \circ \Mn n \alpha \circ \DD n^{\pointAun} = \dd n^{\pointAdeux} \circ \eq{n+1}^{A_2} \circ \Mn {n+1} \alpha \quad : \Mn {n+1} {A_1} \rightarrow \Mn n {A_2}. 
\end{equation}
We now use the fact that by definition of $\Mn {n+1} \alpha$, $\eq{n+1}^{A_2}\circ\Mn {n+1}\alpha = \alpha^{\otimes n+1}\circ \eq{n+1}{A_1}$, and similarly : $\Mn {n} \alpha$, $\eq{n}^{A_2}\circ\Mn {n}\alpha = \alpha^{\otimes n}\circ \eq{n}{A_1}$
 \begin{equation}\label{eq:proof_chain_morph_3}
 \alpha^{\otimes n}\circ \eq{n}{A_1} \circ \DDn n^{\pointAun} = \dd_n^{\pointAdeux} \circ  \alpha^{\otimes n+1}\circ \eq{n+1}{A_1} \quad : \Mn {n+1} {A_1} \rightarrow \Mn n {A_2}. 
\end{equation}
We can again apply -- this time on the left hand-side -- the fact that by definition of $\DD{}$: $\eq n^{A_1} \circ \DDn n^{\pointAun} = \dd_n^{\pointAun} \circ \eq{n+1}^{A_1}$, and we obtain as goal:
 \begin{equation}\label{eq:proof_chain_morph_4}
 \alpha^{\otimes n}\circ \dd_n^{\pointAun} \circ \eq{n+1}^{A_1} = \dd_n^{\pointAdeux} \circ  \alpha^{\otimes n+1}\circ \eq{n+1}^{A_1} \quad : \Mn {n+1} {A_1} \rightarrow \Mn n {A_2}. 
\end{equation}
So to conclude, it is enough to show that $ \alpha^{\otimes n}\circ \dd_n^{\pointAun} = \dd_n^{\pointAdeux} \circ  \alpha^{\otimes n+1}$. Recall that $\dd_n^{\overline{A_i}} = (A_i)^\otimes n\otimes w_i$: from there we see that we can conclude using functoriality of $\otimes$, and the fact that $\alpha$ is a morphism of copointed objects.
\end{proof}

As a corollary of Proposition~\ref{th:chain_morphism_from_copointed_morphisms}, when $A_\bullet$ exists, and the draw-and-delete chain over $A_\bullet$ is well defined, then this chain can be considered as free among all the draw-and-delete chains built over $A$, in the sense that there is a chain morphism from any other draw and delete chain over $A$ to the approximations one. This applies in particular to the approximations chain built for the layered construction of the free exponential, in any model of LL where this construction goes through, for instance $\pcoh$.

\subsection{\texorpdfstring{Exchangeability diagram for $X^{\omega}$}{Exchangeability diagram on the infinite tensor product}}
As said in the introduction, the categorical variant of De Finetti theorem proposed by Jacobs and Staton~\cite{jacobs2020finetti} consists in asking $\Giry X$ to be the limit of the draw-and-delete chain build over the free copointed object $X_\bullet = (X,(\_):X\rightarrow \unit)$. We now present the categorical formulation of De Finetti theorem due to~\cite{fritz2021finetti}, that uses finitely supported symmetries on the \emph{infinite monoidal product} $X^\omega$. Let us first define the infinite monoidal product over any copointed object in an arbitrary symmetric monoidal category.
\begin{definition}
Let $\overline A = (A,w)$ a copointed object. We call the \emph{delete chain over $A$} the chain $(A^{\otimes{n+1} } \xrightarrow {A^{\otimes n}\otimes w} A^{\otimes n})_{n \in \NN}$. When this chain has a limit, we say that the \emph{infinite tensor product over $\overline{A}$ exists}, and we note it  $({\overline A}^\omega, \pi_{\leq n}^{\overline A} : {\overline A}^\omega\rightarrow A^{\otimes n})$. 
\end{definition}
For every $N \in \NN$, ${\overline A}^\omega$ is also the limit of the \emph{truncated} delete chain $(A^{\otimes{n+1} } \xrightarrow {A^{\otimes n}\otimes w} A^{\otimes n})_{n \geq N}$. We use this observation to define the \emph{finitely supported symmetries.}
\begin{notation}
We note $S_n$ for the set of bijections $\{1,\ldots,n\} \rightarrow \{1,\ldots,n\}$. For $A$ an object in $\Catone$, and $\sigma \in S_n$, by abuse of notation we note also $\sigma$ for the corresponding symmetry on $A^{\otimes n}$. Moreover, for $m > n$, we note $\sigma^{\uparrow m}$ for the bijections on $\{1,\ldots,m\}$ which behaves as $\sigma$ on $\{1,\ldots,n\}$, and as the identity on $\{n+1,\ldots,m\}$.
\end{notation}
\begin{definition}
  Let $\overline A = (A,w)$ a copointed object such that the \emph{infinite tensor product over $\overline{A}$ exists}.
For $N \in \NN$, and $\sigma \in S_N$, we note $\dot \sigma$ for the unique morphism ${\overline A}^{\omega} \rightarrow {\overline A}^{\omega}$ given by the cone on the truncated delete chain $(\sigma^{\uparrow n} \circ \pi_{\leq n}^{\overline A}: {\overline A}^\omega \rightarrow{} A^{\otimes n})_{n \geq N}$. The \emph{finitely supported symmetries on ${\overline A}^\omega$} are all the $\dot \sigma$, for some finite permutation $\sigma$. 
\end{definition}
The categorical De Finetti theorem proved in~\cite{moss2022probability,fritz2021finetti} consists in the fact that for every $X$  in $\Pol$, $\Giry X$ is the equaliser of all finitely supported symmetries on $X_\bullet^\omega$ in $\Stoch(\Pol)$.
\subsection{Connecting the two categorical De Finetti theorems}

In this section, we prove that the two categorical presentations of De Finetti's theorem are actually equivalent. 
\begin{propositionrep}\label{prop:cone_from_cone_1}
  Let $Z \xrightarrow{f} X^{\omega}$ a cone on the exchangeability diagram for $X$. Then for every $n \in \NN$, $\pi_{\leq n}^X \circ f: Z \rightarrow X^{\otimes n}$ equalises all symmetries. We note   $ (\pi_{\leq n}^X \circ f)^{\dagger}:Z \rightarrow \Mn n X$ for the unique mediating morphism obtained using the universal property for $\Mn n X$:
  {$$
  \begin{tikzcd}[ampersand replacement=\&]
    Z \ar["f", rr] \ar[d,"\exists ! (\pi_{\leq n}^X \circ f)^{\dagger}",dashed] \& \& X^{\omega} \ar[right,loop,"\text{finit. supp. symm.}", out=30, in=-15, looseness=3] \ar[d,"\pi_{\leq n}^X"]\\
    \Mn n X \ar[rr,"\text{eq}_n "]\& \& X^{\otimes n} \ar[loop,"\text{symm.}", out=+15, in=-30, looseness=3]
    \end{tikzcd}
  $$}
  Moreover $(Z, Z \xrightarrow{{(\pi_{\leq n}^X \circ f)}^{\dagger}} \Mn n X)_{n \in \NN}$ is a cone on the draw-and-delete chain.
  \end{propositionrep}
    \begin{proof}
    We prove successively the two parts of the statement:
    \begin{itemize}
    \item First, let's fix $n \in \NN$, and $\sigma: \{1,\ldots,n\}\rightarrow \{1,\ldots,n\}$ a bijection. We're going to prove that $\sigma \circ \pi_{\leq n} \circ f = \pi_{\leq n} \circ f$. Let's note $\dot \sigma:\NN \rightarrow \NN$ for the finitely supported bijection on $\NN$ that behaves as $\sigma$ on the $n$'th first integers, and as the identity elsewhere. As before, we note also $\dot \sigma$ for the corresponding finitely supported symmetry $X^\omega \rightarrow X^\omega$. Observe that $\sigma \circ \pi_{\leq n} = \pi_{\leq n} \circ \dot \sigma$. From there, we can see that $\sigma \circ \pi_{\leq n} \circ f = \pi_{\leq n} \circ \sigma \circ f$. Using our hypothesis that $f$ equalises all finitely supported symmetries on $X^w$, we can conclude that $\sigma \circ \pi_{\leq n} \circ f = \pi_{\leq n} \circ f$.
    \item Let's now consider the cone $(Z,((\pi_{\leq n} \circ f)^{\dagger}:Z \rightarrow \Mn n X)_{n \in \NN})$. We want to show that it is a cone on the DD chain. So let us fix $n \in \NN$: we're going to show that $\DD_{n}\circ (\pi_{\leq n+1}\circ f)^{\dagger} = (\pi_{\leq n} \circ f)^{\dagger}$.
      First, using the unicity of $(\pi_{\leq n} \circ f)^{\dagger}$, we can reduce our goal to: $\eq n  \circ \DD_n \circ  (\pi_{\leq n+1}\circ f)^{\dagger} = \pi_{\leq n} \circ f $. We know by definition of $\DDn n$ that $ \eq n \circ \DDn n = \dd_n\circ \eq {n+1}$. Using this, we can rewrite again our goal as: $\dd_n\circ \eq {n+1} \circ  (\pi_{\leq n+1}\circ f)^{\dagger} = \pi_{\leq n} \circ f $. The next step consists in unfolding the definition of  $(\pi_{\leq n+1}\circ f)^{\dagger}$: it is the (unique) morphism $Z \rightarrow \Mn{n+1}{X}$ such that $\eq {n+1} \circ  (\pi_{\leq n+1}\circ f)^{\dagger} = \pi_{\leq n+1}\circ f $. This allows us to do one more rewriting step on our goal, which becomes:  $\dd_n\circ \pi_{\leq n+1}\circ f = \pi_{\leq n} \circ f $. We can now conclude using the fact that $\dd_n \circ \pi_{\leq n+1} = \pi_{\leq n}$, since $(X^\omega,\pi_{\leq n}:X^{\omega} \rightarrow X^n)$ is a cone on the delete chain $(d_n:X^{n+1}\rightarrow X^n )$. 
      \end{itemize}
    \end{proof}
    \begin{toappendix}
      \begin{proposition}\label{prop:cone_from_cone_1_parametrized}
        Suppose that the equalisers of symmetries of the $X^{\otimes n}$ commute with tensor product.  Let $Z \xrightarrow{f} X^{\omega}\otimes Y$ that equalises all the $(\dot\sigma \otimes Y)_{n \in \NN,\sigma \in S_n}$. Then for every $n \in \NN$, $(\pi_{\leq n}^X \otimes Y) \circ f: Z \rightarrow X^{\otimes n}$ equalises all symmetries. We note   $ ((\pi_{\leq n}^X \otimes Y)\circ f)^{\dagger}:Z \rightarrow \Mn n X$ for the unique mediating morphism obtained using the universal property for $\Mn n X \otimes Y$:
  {$$
  \begin{tikzcd}[ampersand replacement=\&]
    Z \ar["f", rr] \ar[d,"\exists ! ((\pi_{\leq n}^X\otimes Y) \circ f)^{\dagger}",dashed] \& \& X^{\omega}\otimes Y \ar[right,loop,"\text{finit. supp. symm.}", out=30, in=-15, looseness=3] \ar[d,"\pi_{\leq n}^X \otimes Y"]\\
    \Mn n X \otimes Y\ar[rr,"\text{eq}_n "]\& \& X^{\otimes n}\otimes Y \ar[loop,"\text{symm.}", out=+15, in=-30, looseness=3]
    \end{tikzcd}
  $$}
  Moreover $(Z, Z \xrightarrow{{(\pi_{\leq n}^X \otimes Y \circ f)}^{\dagger}} \Mn n X \otimes Y)_{n \in \NN}$ is a cone on the draw-and-delete chain.
\end{proposition}
\begin{proof}
    We prove successively the two parts of the statement:
    \begin{itemize}
    \item First, let's fix $n \in \NN$, and $\sigma: \{1,\ldots,n\}\rightarrow \{1,\ldots,n\}$ a bijection. We're going to prove that $(\sigma \otimes Y) \circ (\pi_{\leq n} \otimes Y) \circ f = (\pi_{\leq n} \otimes Y) \circ f$. 
      Observe that $\sigma \circ \pi_{\leq n} = \pi_{\leq n} \circ \dot \sigma$, by definition of $\dot \sigma$. As a consequence $(\sigma \otimes Y) \circ (\pi_{\leq n} \otimes Y) = (\pi_{\leq n} \otimes Y) \circ (\dot \sigma \otimes Y)$. From there, we can see that $(\sigma \otimes Y)\circ (\pi_{\leq n}\otimes Y) \circ f = (\pi_{\leq n} \otimes Y) \circ (\sigma \otimes Y)\circ f$. Using our hypothesis that $f$ equalises all the  $\dot \sigma \otimes Y$ on $X^w\otimes Y$, we can conclude that $(\sigma \otimes Y)\circ (\pi_{\leq n} \otimes Y) \circ f = (\pi_{\leq n} \otimes Y) \circ f$.
    \item Let's now consider the cone $(Z,((\pi_{\leq n} \otimes Y \circ f)^{\dagger}:Z \rightarrow \Mn n X \otimes Y)_{n \in \NN})$. We want to show that it is a cone on the $\DD \otimes Y$ chain. So let us fix $n \in \NN$: we're going to show that $(\DD_{n}\otimes Y)\circ ((\pi_{\leq n+1} \otimes Y)\circ f)^{\dagger} = ((\pi_{\leq n} \otimes Y) \circ f)^{\dagger}$.
      First, using the unicity of $((\pi_{\leq n}\otimes Y) \circ f)^{\dagger}$ -- because the equalisers of symmetries onn $X^{\otimes n}$ commutes with $\otimes$ --  we can reduce our goal to: $(\eq n \otimes Y )\circ (\DD_n\otimes Y) \circ  ((\pi_{\leq n+1}\otimes Y)\circ f)^{\dagger} = (\pi_{\leq n}\otimes Y) \circ f $. We know by definition of $\DDn n$ that $ \eq n \circ \DDn n = \dd_n\circ \eq {n+1}$. Using this, we can rewrite again our goal as: $((\dd_n\circ \eq {n+1})\otimes Y) \circ  (\pi_{\leq n+1}\otimes Y\circ f)^{\dagger} = \pi_{\leq n} \circ f $. The next step consists in unfolding the definition of  $((\pi_{\leq n+1} \otimes Y)\circ f)^{\dagger}$: it is the (unique) morphism $Z \rightarrow \Mn{n+1}{X} \otimes Y$ such that $(\eq {n+1} \otimes Y) \circ  ((\pi_{\leq n+1} \otimes Y)\circ f)^{\dagger} = (\pi_{\leq n+1}\otimes Y)\circ f $. This allows us to do one more rewriting step on our goal, which becomes:  $((\dd_n\circ \pi_{\leq n+1})\otimes Y)\circ f = (\pi_{\leq n} \otimes Y) \circ f $. We can now conclude using the fact that $\dd_n \circ \pi_{\leq n+1} = \pi_{\leq n}$, since $(X^\omega,\pi_{\leq n}:X^{\omega} \rightarrow X^n)$ is a cone on the delete chain $(\dd_n:X^{n+1}\rightarrow X^n )$.

      \end{itemize}

  \end{proof}
      \end{toappendix}
   
    Now, we show how to   transform any cone on the draw and delete chain into an equaliser of the finitely supported symmetries on $X^\omega$.
    The first step consists in observing that there is a chain morphism from the draw-and-delete chain to the delete chain.
    
      \begin{notation}
         Let $(Z \xrightarrow {f_n} \Mn n X)_{n \in \NN}$ a cone on the draw and delete chain.
      We note $f^\omega$ for the morphism $Z \rightarrow X^\omega$ obtained from this cone using the universal property of $X^\omega$ on $(Z\xrightarrow{\text{eq}_n \circ f_n} X^{\otimes n})$ -- which can be easily seen to be a cone on the delete chain:
      \vspace{-5pt}
$$\begin{tikzcd}
  \ldots  X^{\otimes n} &  X^{\otimes n+1}\ar[l,"\dd_{n}"]&\ldots   & X^\omega \ar[in = +10, out=170,ll] \ar[lll, in = +10, out=160] \\
    \Mn n X\ar[u,"\eq n"]  &  \Mn {n+1} X   \ar[l,"\DD_{n}"] \ar[u,"\eq {n+1}",swap] & \ldots & Z\ar[ll, out=200, in=-10,"f_n",swap ,near end]\ar[lll, out = 200, in=-20,"f_{n+1}",swap, near end] \ar[dashed,"\exists! f^{\omega}",u,swap]
    \end{tikzcd}$$  
  \end{notation}

  \begin{propositionrep}\label{prop:cone_from_cone_2}
    Let $(Z \xrightarrow {f_n} \Mn n X)_{n \in \NN}$ a cone on the draw and delete chain.
    Then $f^\omega:Z \rightarrow X^\omega$ equalises all finitely supported symmetries on $X^\omega$. 
\end{propositionrep}
\begin{proof}
  Let's $\sigma:\{1,\ldots,n\} \rightarrow \{1,\ldots,n\}$ a bijection, and as before we note $\dot \sigma$ for the associated isomorphism $X^\omega \rightarrow X^\omega$. Our goal is to show that $\dot \sigma \circ f^\omega = f^\omega$. Recall that $f^\omega$ has been obtained from the universal property on $X^\omega$, as the unique morphism that factorises the cone  $(Z \xrightarrow {\eq n \circ f_n} X^{\otimes n})_{n \in \NN}$ through the projections on $X^\omega$. So in order to prove that $\dot \sigma \circ f^\omega = f^\omega$, it is enough to check that $\dot \sigma \circ f^\omega$ also factorises this cone, i.e. that for every $N$, $\pi_{\leq N} \circ \dot \sigma \circ f^\omega = \eq N \circ f_N$. It is actually enough to prove it for $N \geq n$, because for any $n \in \NN$, $X^\omega$ is also the limit of $X^{\otimes n} \xleftarrow{\dd_n} X^{\otimes n+1}\ldots$. So let $N \geq n$. By definition of $\dot \sigma:X^\omega \rightarrow X^\omega$, it holds that  $\pi_{\leq N} \circ \dot \sigma = \dot \sigma_N \circ \pi_{\leq N} : X^\omega \rightarrow X^{\otimes N}$. Thus  $\pi_{\leq N} \circ \dot \sigma \circ f^\omega = \dot \sigma_N \circ \pi_{\leq N} \circ f^\omega = \dot \sigma_N \circ (\eq N \circ f_N)$ by definition of $f^\omega$. We can conclude from here using the fact that $\eq N$ is the equaliser morphism, thus $\dot \sigma_N \circ \eq N = \eq N$, which ends the proof. 
\end{proof}

\begin{toappendix}
  
  \begin{proposition}\label{prop:cone_from_cone_2_parametrized}
    Suppose the $X^\omega$ as limit of the delete chain commutes with the tensor product.
    Let $(Z \xrightarrow {f_n} \Mn n X \otimes Y)_{n \in \NN}$ a cone on the draw and delete chain.
 We note $f^\omega$ for the morphism $Z \rightarrow X^\omega\otimes Y$ obtained from this cone using the universal property of $X^\omega \otimes Y$ on $(Z\xrightarrow{(\text{eq}_n \otimes Y) \circ f_n} X^{\otimes n} \otimes Y)$ -- which can be easily seen to be a cone on the $(\text{delete chain} \otimes Y)$:
      
 $$\begin{tikzcd}
   \ldots  {X^{\otimes n}\otimes Y} &  {X^{\otimes n+1} \otimes Y}\ar[l,"\dd_{n}"]&\ldots   & {X^\omega \otimes Y}\ar[in = +10, out=170,ll] \ar[lll, in = +10, out=160] \\
     {\Mn n X \otimes Y}\ar[u,"\eq n \otimes Y"]  &  {\Mn {n+1} X\otimes Y }  \ar[l,"\DD_{n} \otimes Y"] \ar[u,"\eq {n+1}\otimes Y",swap] & \ldots & Z\ar[ll, out=200, in=-10,"f_{n}",swap ,near end]\ar[lll, out = 200, in=-20,"f_{n+1}",swap, near end] \ar[dashed,"\exists! f^{\omega}",u,swap]
     \end{tikzcd}$$  
    Then $f^\omega:Z \rightarrow (X^\omega \otimes Y)$ equalises the $(\dot \sigma \otimes Y)_{n \in \NN, \sigma \in S_n}$ on $X^\omega \otimes Y$. 
\end{proposition}
\begin{proof}
  Let's $\sigma:\{1,\ldots,n\} \rightarrow \{1,\ldots,n\}$ a bijection, and as before we note $\dot \sigma$ for the associated isomorphism $X^\omega \rightarrow X^\omega$. Our goal is to show that $(\dot \sigma \otimes Y) \circ f^\omega = f^\omega$. Recall that $f^\omega$ has been obtained from the universal property on $(X^\omega \otimes Y)$, as the unique morphism that factorises the cone  $(Z \xrightarrow {(\eq n \otimes Y) \circ f_n} X^{\otimes n} \otimes Y)_{n \in \NN}$ through the $(\text{projections on }X^\omega \otimes Y$. So in order to prove that $(\dot \sigma \otimes Y) \circ f^\omega = f^\omega$, it is enough to check that $(\dot \sigma \otimes Y) \circ f^\omega$ also factorises this cone, i.e. that for every $N\geq n$, $(\pi_{\leq N} \otimes Y) \circ (\dot \sigma \otimes Y) \circ f^\omega =( \eq N \otimes Y) \circ f_N$.
  So let $N \geq n$. By definition of $\dot \sigma:X^\omega \rightarrow X^\omega$, it holds that  $\pi_{\leq N} \circ \dot \sigma = \dot \sigma_N \circ \pi_{\leq N} : X^\omega \rightarrow X^{\otimes N}$. Thus  $((\pi_{\leq N} \circ \dot \sigma)\otimes Y) \circ f^\omega = ((\dot \sigma_N \circ \pi_{\leq N})\otimes Y) \circ f^\omega = (\dot \sigma_N \otimes Y)\circ ((\eq N\otimes Y) \circ f_N)$ by definition of $f^\omega$. We can conclude from here using the fact that $\eq N$ is the equaliser morphism, thus $\dot \sigma_N \circ \eq N = \eq N$, which ends the proof. 
\end{proof}
  \end{toappendix}

As an immediate consequence of Propositions~\ref{prop:cone_from_cone_1} and~\ref{prop:cone_from_cone_2}, we obtain the equivalence of the two formulations of De Finetti theorem. By proving \emph{parametrized variants} of Propositions~\ref{prop:cone_from_cone_1} and~\ref{prop:cone_from_cone_2} -- Propositions~\ref{prop:cone_from_cone_1_parametrized} and~\ref{prop:cone_from_cone_2_parametrized} in appendix -- we can also prove that commutation with tensor product for one of the exchangeability diagrams also induces commutation with the tensor product for the other. 

\begin{theorem}\label{theorem:equivalence_of_two_df_formulations}
In any symmetric monoidal category $(\Catone,\otimes,\unit)$, for any $\unit$-copointed object $\overline X = (X,w:X\rightarrow \unit)$ such that:
\begin{enumerate}
\item the draw-and-delete chain can be built over $\overline X$; 
\item the delete chain $\unit \leftarrow X \leftarrow X^{\otimes 2}\ldots$ built over $\overline X$ has a limit. 
  \end{enumerate}

  The draw-and-delete chain has a limit $(L,\rho_n:L \rightarrow \Mn n X)$ in $\Catone$ if and only if the equaliser $L'$ of all finitely supported symmetries on $X^\omega$ exists, and when it is the case $L \cong L'$.
  If moreover both the equalisers of symmetries on $X^{\otimes n}$ and the limit $\overline{X}^{\omega}$ of the delete chain commute with the tensor product, then it holds that whenever the limit of one of the two diagrams commutes with the tensor product, the other commutes too.
\end{theorem}

\subsection{\texorpdfstring{Application to $\Stoch$}{Application to the category of stochastic kernels}}

\begin{toappendix}
  \label{appendix:categorical_df}

  \subsection{The category Stoch}
 We first give a more formal presentation of the structure of the symmetric monoidal category $\Stoch$.
Let us note as usual $\MeasC$ for the category of measurable spaces, and measurable maps between them. Probability effects can be modelled in $\MeasC$ by way of the $\Giry$ probability monad:
\begin{definition}[The Giry Monad]
  Let $(X,\Sigma_X)$ be a measurable space. We note $\Giry (X)$ for the measurable space obtained by equipping the set of probability measures over $X$ with the smallest sigma-algebra that makes the evaluation map $(\mu \in \Giry X \mapsto \mu(A) \in \RR)$ measurable for any $A \in \Sigma_X$. 
\end{definition}
The construction above is well-known to give a monad on $\Meas$, and moreover the morphisms $X \rightarrow Y$ in the Kleisli category of $\Giry$ are exactly the stochastic kernels (also called Markov kernels) from $X$ to $Y$. We note $\Stoch$ for the Kleisli category of the Giry Monad on $\Meas$.

For any $\Catone$ a full sub-category of $\Meas$ stable by $\Giry$, we will also note $\Stoch(\Catone)$ for the Kleisli category of $\Giry$ on $\Catone$. In this work, we are especially interested in full sub-category of standard Borel spaces: they are the measurable spaces whose $\sigma$-algebra is generated as the Borel $\sigma$ algebra of a \emph{polish} topological space, i.e. a separable and completely metrizable topological space.
  \begin{definition}
  A standard Borel space is a measurable space $(X,\Sigma_X$), where the underlying space $X$ is a polish topological space, and $\Sigma_X$ is its Borel $\sigma$-algebra.   We note $\Pol$ for the full sub-category of $\Meas$ consisting of standard Borel spaces and measurable maps between them.
  \end{definition}
This can in particular be applied to $\Pol$--see e.g.~~\cite{kechris,DBLP:journals/entcs/DahlqvistSDG18}, as we recall below:
\begin{proposition}
$\Pol$ is stable by $\Giry$, and for any standard Borel space $X$, the underlying topology on $\Giry X$ is the \emph{weak topology},i.e.~the smallest topology that make the maps  $(\mu \in \Meas (X) \mapsto \int f(x) \cdot \mu(dx) \in \RR)$ continuous, for $f:X \rightarrow \RR$ bounded continuous. 
  \end{proposition}

  Since $\Giry$ is a commutative strong monad on the cartesian monoidal structure of $\MeasC$ the category $\Stoch$ is symmetric monoidal. The monoidal product in $\Stoch$ of two measurable spaces is simply their cartesian product in $\Meas$.
  The monoidal unit $\unit_{\Stoch}$ is the singleton space, the exchange morphism  $\text{switch}:A \times_{\MeasC} B\rightarrow B \times_{\MeasC} A $ is defined as $\text{switch}(a,b) = \dirac{(b,a)}$, and the associator is $\text{assoc}:(A\times_{\MeasC} B)\times_{\MeasC} C \rightarrow A \times_{\MeasC} (B\times_{\MeasC} C)$ is defined as $\text{assoc}(a,(b,c)) = \dirac{((a,b),c)}$.

  If $X \xrightarrow f X'$ and $Y \xrightarrow g Y'$ are two stochastic kernels, $f \otimes g$ is the kernel which for any $(x,y) \in X \times Y$ returns the product measure $f(x) \times g(y)$. 

  By a similar construction, starting from the \emph{Panangaden Monad} of \emph{sub-probability} measures, we can build $\SubStoch$, the category of sub-stochastic kernels as its Kleisli category. There is obviously a faithful embedding $\Stoch \rightarrow \SubStoch$. Moreover, $\SubStoch$ also inherits a symmetric monoidal product of the $\Meas$ cartesian product, and the embedding $\Stoch \rightarrow \SubStoch$ preserves the monoidal product. However $\unit$ isn't terminal in $\SubStoch$.

      We start with the construction of the infinite product over a measurable space.
\begin{definition}
  For any measurable space $X$,  $X^{\omega}$ is the set of infinite lists of elements in $X$, equipped with the $\infty$-product $ \sigma$-algebra, i.e. where the $\sigma$-algebra is generated by the $X_1\ldots X_n \uparrow:= \{x_1,\ldots,x_n,y \mid x_1 \in X_1,\ldots, x_n \in X_n, y \in X^{\omega}\}$ for $X_1,\ldots,X_n$ elements of the $\sigma$-algebra on $X$. If moreover $X$ is the Borel $\sigma$-algebra generated by some topology, $X^\omega$ coincides with the Borel $\sigma$-algebra generated by the $\infty$-product topology.
 
\end{definition}

  Starting from a probability measure $\distrone$ on $X$ we can build an infinite-product probability measure $\distrone^\omega$ on $X^\omega$, by $\distrone^\omega(A_1\times \ldots \times A_n \uparrow) = \prod_{1 \leq i \leq n} \distrone(A_i)$. Moreover, by construction of $X^\omega$, the function $T^{\omega}:\distrone \in \Giry X \mapsto \distrone^\omega \in \Giry{(X^\omega)}$ is measurable.
While the construction above can be done in \emph{any} measurable spaces, for a standard Borel space $X$, $X^\omega$ can be expressed as the categorical limit of its finite approximations: it is a consequence of Kolmogorov's extension theorem (e.g.~~\cite{kechris}).
The Proposition below was proved by Fritz and Rischel~\cite{fritz2020infinite}.  
\begin{propositionrep}\label{prop:infinite_product_borel_standard}
  Let $X$ be a standard Borel space.
$(X^\omega,\pi_{\leq n}^X)_{n \in \NN}$ is the limit in $\Stoch(\Pol)$ of the diagram $(X^{n} \xleftarrow{\dd_n} X^{n+1})_{n \in \NN}$, where $\dd_n$ is the deterministic stochastic kernel that to any $(x_1,\ldots,x_{n+1}) \in X^{n+1} \mapsto (x_1,\ldots,x_n) \in X^{n}$, and  $\pi_{\leq n}^X$ is the deterministic stochastic kernel that projects any infinite word on $X^\omega$ to its prefix of length $n$. Moreover this limit is preserved by the functor $\_ \otimes Y$, for any $Y$ a standard Borel space.
\end{propositionrep}
\begin{appendixproof}
It has been shown in~\cite{fritz2020infinite} that $X^\omega$ is the limit in $\Stoch$  of the diagram $(\Pi_{n \in F} X \xleftarrow{d_n} \Pi_{n \in G} X)_{F\subseteq G, \, F,G \text{ finite subsets of }\NN}$-- in~\cite{fritz2020infinite}, this result is actually shown in $\BorelStoch$, but the proof uses only the fact that $X$ is Borel standard, and not the fact that measurable spaces in the ambient category are also standard Borel spaces. The proposition above can be derived from there, by observing that there is a correspondence between the cones on the diagram above, and the cones on $(X^{n} \xleftarrow{\dd_n} X^{n+1})_{n \in \NN}$.   
\end{appendixproof}

  \subsection{Finitely supported symmetries on the infinite product}\label{sect:df_infinite_product}

We define below \emph{finitely supported symmetries} on $X^\omega$.
\begin{notation}
Recall that in $\Stoch$ the space  $X^n$ is  the monoidal product of $n$ copies of $X$, thus  any bijection $\sigma$
% : \{1, \ldots,n\} \rightarrow \{1,\ldots,n\}$
  on $\{1,\ldots,n\}$ for some $n \in \NN$ induces a symmetry $\dot \sigma:X^n \rightarrow X^n$. For $m \geq n$, we can extend this by defining $\dot \sigma_m: X^m \rightarrow X^m$ as the symmetry on $X^m$ associated to the bijection on $\{1,\ldots,m\}$ that permutes the first $n$-th elements along $\sigma$, and fix the others.
\end{notation}

The next step is to associate to $\sigma$ a morphism $X^\omega \rightarrow X^\omega$, that we'll call the \emph{finitely supported permutation on $X^\omega$ associated to $\sigma$}.

For this, we use the fact that $\forall N \in \NN$, $X^\omega$ is also the limit of the chain $X^N \xleftarrow {\dd_N} X^{N+1} \xleftarrow{\dd_{N+1}}\ldots$, and that moreover  $\dd_{m} \circ \dot \sigma_{m+1} = \dot \sigma_m$. Observe that this construction is not specific to $\Stoch$, and will hold in any monoidal category where the limit $X^\omega$ exists.
\begin{definition}
Let $n \in \NN$, and $\sigma: \{1, \ldots,n\} \rightarrow \{1,\ldots,n\}$ a bijection. We note $\dot \sigma: X^\omega \rightarrow X^\omega$ for the $\Stoch$-morphism obtained by the universal property of $X^\omega$ from the cone $(X^\omega \xrightarrow{\dot \sigma_m \circ \pi_{\leq m}} X^m)_{m \geq n}$. 
  \end{definition}

  \end{toappendix}

  Jacobs and Staton in~\cite{jacobs2020finetti} showed the existence in $\Stoch$ of the draw-and-delete chain for $X=\deux$.
  \footnote{In a later work, Staton and Summers~\cite{staton2023quantum} extended the draw-and-delete formulation of De Finetti's theorem to the category of \emph{compact convex Hausdorff spaces}, equipped with the Radon monad. Any second-countable compact convex Hausdorff space is Polish, but the two classes are incomparable}
 As we state in Proposition~\ref{prop:existence_of_equalisers_and_coequalisers_in_stoch} below, the construction can actually be done similarly for any standard Borel space $X$, by taking some additional care for measurability checks; the proof can be found in Appendix, Proposition~\ref{prop:mn_equaliser_and_coequaliser}. 
 \begin{toappendix}
\subsection{Proof of Proposition~\ref{prop:existence_of_equalisers_and_coequalisers_in_stoch}}
   \end{toappendix}
 \begin{propositionrep}\label{prop:existence_of_equalisers_and_coequalisers_in_stoch}
    For a standard Borel space $X$,
the draw-and-delete chain over the copointed object $(X,\_:X \rightarrow \unit)$ -- that we call the De Finetti chain -- can be built in $\Stoch$. Moreover, the equaliser of the symmetries on $A^{\otimes}$ commutes with the tensor product.
\end{propositionrep}
\begin{proof}
The construction follows the same path as for in Jacobs and Staton work for the Boolean case;  we additionnally need to check measurability of all the morphisms we need. First, we equipp the set of $n$-sized multisets with the structure of a measurability space. Our construction is the standard one, see for instance~\cite{BlanchiP24}.

\begin{definition}
  We define $\Mn n X$ as the measurable space with as underlying space the set of multiset of length $n$ of elements in $X$, and as $\sigma$-algebra the smallest one that makes measurable the map
  $\coeq n := (a_1,\ldots,a_n) \in X^n \mapsto [a_1,\ldots,a_n] \in \Mn n X$.
  
 \end{definition}
 It is known -- see for instance~\cite{dash2021monads,macchi1975coincidence} that if $X$ is a standard Borel space, then it is also the case of $\Mn n X$. Moreover, when $X$ is a standard Borel space, there is a measurable total order  $ < \,\subseteq X \times X$, because any standard Borel space is isomorphic to the reals, or discrete and countable. From there, it is possible -- as done in~\cite{dash2021monads} -- to define uniquely an \emph{ordering} $\text{ord}(\mu) \in X^n$ for any multiset $\mu \in \MnCat \Set n X$, such that moreover $\mu \in \Mn n X \mapsto \text{ord}(\mu) \in X^n$ is measurable. 
 \begin{definition}
   For $X$ a standard Borel space,
We define a measurable map $\eq n:\Mn n X \rightarrow X^n$ by: $\eq n = \sum_{\sigma \in S_n} \frac 1 {\card{S_n}}\cdot  \dot{\sigma} \circ \text{ord}$, where $S_n$ is the set of all permutations on $\{1,\ldots,n\}$.
\end{definition}

\begin{proposition}\label{prop:mn_equaliser_and_coequaliser}
For $X,Y$ a standard Borel space, $(\Mn n X \otimes Y,\eq n \otimes Y)$ and $(\Mn n X \otimes Y,\coeq n \otimes Y )$ are in $\Stoch$ respectively the equalizer and coequalizer of all the $(\dot \sigma \otimes Y)_{\sigma \in S_n}$. %for $\sigma:A^{\otimes n} \rightarrow A^{\otimes n}$ symmetries. 
\end{proposition}

\begin{proof}
  It is clear that $\eq n \otimes Y:\Mn n X \otimes Y  \rightarrow X^n \otimes Y$ equalizes all the $\sigma \otimes Y$, and similarly that $\coeq n: X^n \rightarrow \Mn n X$ co-equalizes  all the $\sigma \otimes Y$. We now need to show that they are the  equalizer and  co-equalizer respectively.
    \begin{itemize}
    \item Let $(A,f:A \rightarrow{(X^{\otimes n}\otimes Y)}$ in $\Stoch$ that equalizes the $\dot \sigma \otimes Y$, for all the symmetries $\dot \sigma$ on $X^{\otimes n}$. Thus we can define $f^{\dagger}:a \in A \mapsto ((\coeq n \otimes Y) \circ f)(a) \in \Mn n X \otimes Y$, and it is immediate that it is a morphism in $\Stoch$, because it is build there by composition and $\otimes$ on morphisms. From the fact that for every $\sigma$, $(\dot \sigma \otimes Y )\circ f = f$, we can check by computing concretely the composition that $ (\eq n \otimes Y) \circ (\coeq n \otimes Y) \circ f = f$.  %It isn't obvious, but it can be obtained from the fact that  if we start from a multiset $\mu$, the number of permutations that fix an enumeration $(a_1,\ldots,a_n)$ of $\mu$ doesn't depends on the enumeration we chose.
        Indeed, observe that for $(a_1,\ldots,a_n)\in X^{\otimes n}$:   
      $$  \eq n \circ \coeq n(a_1,\ldots,a_n) =  \frac 1 {\card{S_n}}\cdot \sum_{\sigma}\dot \sigma \circ \text{ord}([a_1,\ldots,a_n]) 
        =   \frac 1 {\card{S_n}}\cdot \sum_{\sigma}\dot \sigma (a_1,\ldots,a_n),$$ 
      thus $\eq n \circ \coeq n = \frac 1 {\card{S_n}} \cdot \sum_{\sigma}\dot \sigma$, and since $f$ equalises the $\dot \sigma \otimes Y$,  $(\eq n \otimes Y) \circ (\coeq n \otimes Y) \circ f =  \frac 1 {\card{S_n}} \sum_{\sigma } (\sigma \otimes Y) \circ f = f.$

        It means that $(\eq n \otimes Y) \circ f^{\dagger} = f$, so $f^\dagger$ is indeed a factorisation of $f$ by $\eq n \otimes Y : \Mn n X \rightarrow X^n$. Now let us check the unicity of this factorization. Let be any $f^\star$ such that $(\eq n \otimes Y) \circ f^{\star} = f$. We can rewrite this equality as $(\eq n \otimes Y) \circ f^{\star} = (\eq n \otimes Y) \circ (\coeq n \otimes Y) \circ f$.  Observe now that $\eq n \otimes Y: (m,y) \in \Mn n X \times Y \rightarrow (\eq n (m) \otimes \dirac y) \in X^{\otimes n} \times y$ is a monomorphism, it means that it follows from the previous inequality that  $ f^{\star} =  (\coeq n \otimes Y) \circ f = f^\dagger$.
      \item Let $(A,g:X^{\otimes n} \otimes Y \rightarrow A)$ in $\Stoch$,  that co-equalises all the $\sigma \otimes Y$. We reason symmetrically to the previous case. Our hypothesis on $g$ is that for every $\sigma$, $g \circ (\sigma \otimes Y) = g$. From there, we can check--again, because  $\eq n \circ \coeq n = \frac 1 {\card{S_n}} \cdot \sum_{\sigma}\dot \sigma$ --that $g \circ (\eq n \otimes Y) \circ (\coeq n \otimes Y) = g$. We pose $g^\dagger =  g \circ (\eq n \otimes Y)$. We can see immediately that $g^\dagger \circ (\coeq n \otimes Y) =g$, and we can check the uniqueness of the factorization using the fact that $\coeq n \otimes Y$ is an epimorphism.
        \end{itemize}
      \end{proof}

  \end{proof}

  The existence of the infinite tensorial product $X^\omega$ in $\Stoch(\Pol)$ was shown in~\cite{fritz2020infinite} to be a consequence of Kolmogorov's extension theorem. As said in the introduction, Perrone et al showed {in~\cite{moss2022probability,fritz2021finetti}}\footnote{The \emph{existence} of a decomposition through $\Giry X \rightarrow X^\omega$ for any morphism $Z \rightarrow X^{\omega}$ that equalises the symmetries is shown in~\cite{fritz2021finetti} in a general framework of Markov categories-- with some additional requirements, while the \emph{unicity} of the decomposition is shown specifically for $\Stoch(\Pol)$ in~\cite{moss2022probability}. } that $\Giry X$ is the limit of the symmetries on $X^{\omega}$ when taking as ambient category $\Stoch(\Pol)$, where the objects are Standard Borel spaces, and the morphisms are stochastic kernels. 
  Our Theorem~\ref{theorem:equivalence_of_two_df_formulations} thus shows that we can transfer this De Finetti result in $\Stoch(\Pol)$ to extend the stratified formulation of Jacobs and Staton to every standard Borel space $X$-- since we've shown in Proposition~\ref{prop:existence_of_equalisers_and_coequalisers_in_stoch} that the draw-and-delete chain over $(X,\_:X\rightarrow \unit)$ could indeed always be built. As a consequence, $\Giry X$ is the limit of the draw-and-delete chain over $X$ in $\Stoch(\Pol)$.

  We are interested in a slightly stronger \emph{parametrized version}, namely that $\Giry X\otimes Y$ is the limit of the chain $\DD^{\text{DF}} \otimes Y$: our goal is to be able to apply Theorem~\ref{th:metabasson} to prove that $\Giry X$ is moreover the \emph{free commutative comonoid} over $X$. It was proved in~\cite{fritz2020infinite} that the limit of the delete chain commutes with the tensor product, and we proved in Proposition~\ref{prop:existence_of_equalisers_and_coequalisers_in_stoch} that also the equalisers in $\Stoch(\Pol)$ commutes with the tensor product. From there, we only have to show a stronger, \emph{parametrized} variant of~\cite{moss2022probability,fritz2021finetti}\footnote{We've been made aware after publication that this parametrized De Finetti theorem was already proved in~\cite{Fritz_2026}, Corollary 4.5, in a general categorical setting. Our proof seems to be a concrete instantiation of theirs.} result that $\Giry X$ is the equalisers of the finitely supported, i.e. that this equaliser also commutes with the tensor product:

  \begin{toappendix}
\subsection{Proof of Proposition~\ref{th:de-finetti-perrone}}
    \end{toappendix}
\begin{propositionrep}\label{th:de-finetti-perrone}
  Let $X, Y$ be standard Borel spaces. Then  $\Giry X \otimes Y$ 
  is the equaliser in $\Stoch(\Pol)$ of  the $X^{\omega} \otimes Y\xrightarrow{\sigma \otimes Y} X^{\omega} \otimes Y$, for all   finitely supported symmetries $\sigma$.
\end{propositionrep}
\begin{proofsketch}
The proof given in~\cite{fritz2021finetti} of their categorical De Finetti theorem -- $\Giry X$ as the equaliser of all finitely supported symmetries on $X^\omega$ --  is in the language of general Markov categories, and uses heavily the idea of conditional expectations. We didn't try to adapt their proof to show commutation of equaliser with monoidal product, and instead find it simpler here to do a more direct proof in $\Stoch$ of our parametrized version of De Finetti theorem. For this, we use crucially the notion of \emph{weak convergence of measures} in $\Pol$, and a caracterisation, presented in~\cite{klenke2013probability} -- of the representation measure in $\Giry(\Giry X)$ from De Finetti theorem as \emph{weak limit} of a sequence of simpler \emph{empirical measures}.
  \end{proofsketch}
\begin{proof}
    The goal of this appendix is to show that the equalizer of the symmetries on $X^\omega$-- which we know to be $\Giry X$ by the categorical De Finetti theorem of~\cite{fritz2021finetti,moss2022probability} -- commutes with the monoidal product in $\Stoch(\Pol)$, as we stated in Theorem~\ref{th:de-finetti-perrone}. More precisely, we are going to show a \emph{parametrized variant} of the categorical De Finetti theorem from~\cite{fritz2021finetti,moss2022probability}, i.e. that for any standard Borel spaces $X,Y$, $\Giry X\otimes Y$ is in $\Stoch(\Pol)$ the  equalisers of the $\{\dot \sigma \otimes Y \mid \sigma \text{ a symmetry}\}$ on $X^\omega \otimes Y$.
The proof given in~\cite{fritz2021finetti} of their categorical De Finetti theorem -- $\Giry X$ as the equaliser of all finitely supported symmetries on $X^\omega$ --  is in the language of general Markov category, and uses heavily the idea of conditional expectations. While it is likely it could be adapted to show commutation of equaliser with monoidal product, we do a more direct proof here of our parametrized version of De Finetti theorem, that uses crucially the notion of \emph{weak convergence of measures} in $\Pol$.

Before starting our proof of this parametrized categorical De Finetti Theorem, let us look at where the difficulty lies already in the non-parametrized case: let $f:Y\rightarrow X^{\omega}$ that equalises all symmetries. Let us note $(\Giry X,\text{eq}:\Giry X \rightarrow X^\omega)$ the equalizer. By the usual statement in measure theory of De Finetti Theorem for standard Borel spaces, it is not complicated to obtain the existence, for each $y \in Y$, of a representing measure $\distrone_Y \in \Giry X$ such that $\text{eq}(\distrone_y) = f(y)$. The main problem consists in proving that defining $y \in Y \mapsto \mu_y \in \Giry Y$ produces indeed a $\Stoch$-morphism.  Our way of solving this problem, which is obviously still there in the parametized case, is to use a stronger variant of De Finetti theorem for $\Pol$, that gives a caracterisation of the representing measure as the \emph{weak limit} of a sequence of \emph{empirical measures}. This caracterisation is of interest for us, because the pointwise limit of a sequence of measurable functions $R\rightarrow Q$ with $Q$ metrizable is again measurable--see e.g.~\cite{kechris} 11.2. So we can prove measurability of $(y \in Y \mapsto \mu_y \in \Giry Y)$ by proving measurability of every $(y \in Y \mapsto \mu_n^y)$, where $\mu_n^y$ is the $n$-th empirical measure approximating $\distrone_y$, which will be much easier.

\begin{definition}[Weak topology on $\Giry(X)$]
  Let $X$ be a Polish space.
  Let $(\mu_n)_{n }\in \NN$ be a sequence in $\Giry(X)$ and $\mu \in \Giry(X)$. The sequence $(\mu_n)_{n \in \NN}$ \emph{converges weakly} towards $\mu$ when: for all $\psi: X \rightarrow \RR$ bounded and continuous, $\lim_{n \rightarrow + \infty} \int \psi(r)\cdot \mu_n(dr) = \int \psi(r)\cdot \mu(dr)$. 
 \end{definition}
\begin{lemma}[Linearity of Weak convergence]
Let $(\mu_n)_{n \in \NN},{\nu_n}_{n \in \NN}$ be two sequences in $\Giry X$ such that $\wlim_{n \in \NN} \mu_n = \mu$, and $\wlim_{n \in \NN} \nu_n = \nu$. Let $\lambda \in [0,1]$. Then it holds that $\wlim_{n \in \NN}(\lambda \cdot \mu_n + \nu_n) = \lambda \mu + \nu$. 
  \end{lemma}
  It is known that if $X$ is Polish, $\Giry X$ equipped with the weak topology is again Polish.
  By definition, the weak topology is the smallest topology that makes the maps $(\mu \in \Giry (X) \mapsto \int f(x) \cdot \mu(dx) \in \RR)$ continuous, for $f:X \rightarrow \RR$ bounded continuous.
\begin{proposition}[From Kechris, th 17.24]\label{proposition:weak_maps_generate_giry}
  Let $X$ be a Polish space. Then $\Sigma_B$--the Borel $\sigma$-algebra of $\Giry(X)$ equipped with the weak topology is generated by:
  \begin{itemize}
  \item the maps $\mu \mapsto \mu(A)$, with $ A \in \Sigma_B$;
    \item and also by the maps $\mu \mapsto \int f(x)\cdot \mu(dx)$, for $f:X \rightarrow \RR$ bounded continuous.
    \end{itemize}
  \end{proposition}

  As said above, one way of proving (non-categorical) De Finetti type results for Polish spaces is to show that for any infinite exchangeable sequence $X_1,\ldots,X_n,\ldots$ of $X$-valued  random variables, the sequence of the empiric mean $Z_n = \sum \frac 1 n \dirac{X_i}$--which are random variables with values in $\Giry X$--converges \emph{in law} to some probability measure $\distrone$ in $\Giry \Giry X$ when $n \rightarrow +\infty$ -- meaning that the sequence $\law{Z_n} \in \Giry \Giry X$ converges for the weak convergence towards $\mu$. Then, it is possible to conclude by proving that distinct exchangeable sequences of $X$-random variables produce this way necessary distinct measures in $\Giry\Giry (X)$. This approach is for instance developed in~\cite{klenke2013probability}, Chapter 13.4.
 
 \begin{proposition}[from~\cite{klenke2013probability}]\label{th:de_finetti_via_weak_convergence}
  Let $E$ be a Polish space and $(X1,X2,…)$ an infinite exchangeable sequence of $E$-valued random variables. Then there exists a (necessarily unique) probability measure $D \in \Giry{\Giry E}$ such that, for every $n\in \NN$, and $A_1,\ldots,A_n\subseteq E$ measurable,
  $\text{Prob}(X_1\in A_1,…,X_n\in A_n)=\int \distrone(A_1)\ldots \distrone(A_n) D(d\distrone).$
  Furthermore, if we define the empirical measures
  $Z_n:=\frac 1 n \cdot \sum_{i=1}^n \dirac{X_i}: \Omega \rightarrow \Giry E$, then the family $(\text{law}(Z_n))_{n \in \NN}$ is tight, and moreover $\text{law}(Z_n)$ converge weakly towards $D$ when $n \rightarrow + \infty$.

\end{proposition}
  \begin{proof}
    The proof can be found in~\cite{klenke2013probability},
    % and also in~\cite{post_on_maths_stack_exchange}
    but we give here the main lines. Let us define $c_n = \text{law}(Z_n) \in \Giry \Giry E)$.
    \begin{itemize}
    \item First, it is shown that the sequence $c$ is \emph{tight} in $\Giry\Giry E$.  Let us recall that \emph{tight} means that for every $\epsilon >0$, there exists a compact set $K_\epsilon \subseteq \Giry E$ such that for all $n \in \NN$, $ c_n(K_\epsilon) > 1-\epsilon$. ~\cite{klenke2013probability} (Exercise 13.41) provides the following characterisation of tight families in $\Giry \Giry E$: they are those families $(\distrone_i)_{i \in I}$ such that for any $\epsilon$, there exists a compact $K\subseteq E$ such that:

      \begin{equation}\label{eq:tight_families}
        \forall i \in I,\,f_i(\{\alpha \in \Giry E \mid \alpha(K) \geq 1-\epsilon\}) \geq 1- \epsilon.
        \end{equation}
        So let $\epsilon >0$. First (since $E$ is polish), a singleton $\{\alpha\}$ with $\alpha$ in $\Giry E$ is always tight (see Lemma 13.5 in~\cite{klenke2013probability}), i.e we can find a compact $K_1 \in E$ such that $\text{Prob}(X_1 \in K_1 ) > 1 - \epsilon^{2}$. But recall that by exchangeability, all the $X_i$ are identically distributed, thus $\forall i$, $\text{Prob}(X_i \in K_1 ) > 1 - \epsilon^{2}$. As a consequence, for $n\in \NN$:
        $c_n(\{\alpha \in \Giry E \mid \alpha(K) \geq 1-\epsilon\} ) =  \frac 1 n \cdot \card{\{i \mid law(X_i(K_1) \}} $
\item 
    From there, we obtain that $(c_n)$ is relatively compact using Prokhorov's theorem for Polish spaces. From relative compactness, we can deduce that the sequence $c$ has a converging sub-sequence. Moreover, it can be proved that there is only one possible limit for such a sub-sequence, which proves that $c$ itself converges.
\end{itemize}
\end{proof}
We are now going to prove a kind of \emph{parametrized} variant of Proposition~\ref{th:de_finetti_via_weak_convergence}.The proof follows closely the proof of Proposition~\ref{th:de_finetti_via_weak_convergence} in Klenke book~\cite{klenke2013probability}. 
\begin{definition}
  Let $E$ a topological space. We say that $\mathcal F \subseteq \MeasC(E,\RR)$ is \emph{separating} when:
  $$\forall \mu,\nu \in \Giry E, \quad \quad ( \forall f \in \mathcal F , \int f\cdot d\mu = \int f \cdot d\nu) \quad \Leftrightarrow \quad \mu = \nu.  $$
  If $E$ is a standard Borel space, $\mathcal M(E)^c:= \{f \in \MeasC(E,\RR) \mid f \text{ is bounded continuous}\}$ is a separating family (e.g.~\cite{klenke2013probability}, Theorem 13.11).
\end{definition}
For $E = \Giry X$ for some standard Borel space $X$, the separating families on $\Giry E$ have been studied in~\cite{blount2010convergence}. In particular, the following is proved in Theorem 11 there:
\begin{lemma}
  Let $E$ a standard Borel space.
  Let $\mathcal F \subseteq \mathcal M(E)^c$ be such that:
  \begin{itemize}
  \item is stable by multiplication;
  \item separates points, i.e. $\forall e,e' \in E, (\forall f, f(e) = f(e') \Leftrightarrow e=e')$.
  \end{itemize}
  Then $\mathcal F$ is a separating family on $\Giry E$.
\end{lemma}
\begin{notation}
  Let $E,F$ standard Borel space. Given $n \in \NN, f_1,\ldots,f_n \in \mathcal M(E)^c$, $g\in \mathcal M(F)^c$, we note $\pi[f_1,\ldots,f_n,g]:= (\mu,z) \in \Giry E \times F \mapsto (f_1(z)\cdot \prod_{1\leq i \leq n}\int f_i.d\mu) \in \RR).$ It is easy to check that $\pi[f_1,\ldots,f_n,g] \in \mathcal M(\Giry E \times F)^c$.
\end{notation}
\begin{corollary}
  The family:
  $\{ \pi[f_1,\ldots,f_n,g] \mid  n \in \NN, \, f_1,\ldots,f_n \in \mathcal M(E)^c, g\in \mathcal M(F)^c\}$ is a separating family on $\Giry E \times F$.
  \end{corollary}

  \begin{proposition}\label{th:de_finetti_via_weak_convergence_parametrized}
    Let $E,F$ be a standard Borel spaces,  $(X1,X2,…)$ an infinite sequence of $E$-valued random variables, and $Y$ an $F$-valued random variable, such that $\forall \sigma$ permutation in $S_n$, $\law{(X_1,\ldots,X_n,Y)} = \law{(X_{\sigma(1)},\ldots,X_{\sigma(n)}}$. Then there exists a (necessarily unique) probability measure $D \in \Giry{(\Giry E) \times Y}$ such that
    $\forall n\in \NN, \forall f_1,\ldots,f_n : E\rightarrow \RR$ continuous bounded, $\forall g:F\rightarrow \RR$ continuous bounded,
$$ \int_{(\mu,z)\in \Giry E \times F} \pi[f_1,\ldots,f_n,g](\mu,z). dD = \mathbb{E}\left[g(Y)\cdot \prod_{i=1}^n f_i(X_i) \right] $$

  Furthermore, if we define the empirical measures
  $W_n=(\frac 1 n \cdot \sum_{i=1}^n \dirac{X_i},Y): \Omega \rightarrow (\Giry E) \times F$, the sequence of $\text{law}(W_n) \in \Giry (\Giry E \times F)$ is tight and thus -- by Prohorov's theorem -- sequentially compact, and as a consequence there exists a subsequence of $\text{law}(W_n) \in \Giry (\Giry E \times F)$ converge weakly towards $D$ when $n \rightarrow + \infty$. Since $D$ is unique, it also means that the sequence $\law{W_n}$ converge weakly toward $D$.
 
\end{proposition}
\begin{proof}
  \begin{itemize}
  \item We are first going to show that the family $(W_n)_{n \in \NN}$ is tight. Let $\epsilon > 0$
    \begin{itemize}
    \item The first step consists in applyin Proposition~\ref{th:de_finetti_via_weak_convergence}, and the well known result that any singleton is tight in $\Giry F$ -- see e.g. Lemma 13.5 in~\cite{klenke2013probability} -- to obtain two compacts: $C_E \subseteq \Giry E$ and $C_F \subseteq F$ such that $\law(Y)(C_F)\geq 1- \frac{\epsilon}{4}$ and  $\forall n, \law(\frac 1 n \cdot \sum \dirac{X_i})(C_F)\geq 1- \frac{\epsilon}{4}$.
      Moreover, by Tychonoff theorem, $C_E \times C_F$ is again a compact in $(\Giry E) \times  F$.
    \item The first step consists in observing that the family of $E\times F$-valued random variables $U_i = X_i \times Y$ is again exchangeable. It means that we can apply Theorem~\ref{th:de_finetti_via_weak_convergence} above, and we obtain that $Z_n:=\frac 1 n \cdot \sum_{i=1}^n \dirac{X_i \times Y}$ is tight. So let $K \subseteq \Giry(E\times F)$ a compact such that $\forall n$, $\text{law}(Z_n)(K) \geq 1-\frac {\epsilon}{4}$.
    \item Let us note $\lambda: (\mu,f) \in (\Giry E) \times F \mapsto \mu \times \dirac f \in \Giry(E\times F)$. $\lambda$ is a continuous function.  We define $K':=\lambda^{-1}(K) \cap (C_E\times C_F)$. We can prove that $K'$ is a compact, because of the following fact (that holds because in any compact metric space, the compacts subset are exactly the closed ones, and moreover the inverse image of a closed set by a continuous function is closed) :
\begin{align*}
  & A,B \text{ metric spaces}, f:A\rightarrow B \text{ continuous}, A \text{ compact }\\ &
                                                                                           \quad \Rightarrow \quad (\forall I \subseteq B, \text{ s.t.} I \text{ is compact}, f^{-1}(I) \text{ is compact}).
  \end{align*}
      
    \item So to conclude, it is enough to show that $\forall n, \law{W_n}(K') \geq 1-\epsilon. $ Observe that:
      $$\law{W_n}(K') \geq \law{W_n}(\lambda^{-1}(K)) - \law{W_n}{(\Giry E\times F \setminus (C_E\times C_F))} $$
      We are going to compute the two parts separately. On the one hand: 
      $\law{W_n}(\lambda^{-1}(K)) = \law{Z_n}(K) \geq 1-\frac \epsilon K $ by hypothesis on $K$.
      On the other hand:
     $ \law{W_n}{(\Giry E\times F \setminus (C_E\times C_F))} = \text{Prob}(\sum_{i} \frac 1 n X_i \not \in C_E \vee Y \not \in C_F) \leq  \text{Prob}(\sum_{i} \frac 1 n X_i \not \in C_E) + \text{Prob} ( Y \not \in C_F) \leq \frac {\epsilon}{4}$ by hypothesis. Combining the two, we obtain $\law{W_n}(\lambda^{-1}(K)) \geq 1 - \frac {3\epsilon} 4$, which allows to conclude.
   \end{itemize}
 \item Let's now take $f_1,\ldots f_n \in \mathcal M^c(E), g\in \mathcal M^C(F)$, and prove that: $\pi[f_1,\ldots,f_n,g](\mu,z). dD = \mathbb{E}\left[g(Y)\cdot \prod_{i=1}^n f_i(X_i) \right]$. Since the $\pi[f_1,\ldots,f_n,g]$ are bounded continuous, by definition of the weak limit, it is enough to show:
   $$\lim_{m\rightarrow +\infty}\pi[f_1,\ldots,f_n,g](\mu,z). d\law{Z_m} = \mathbb{E}\left[g(Y)\cdot \prod_{i=1}^n f_i(X_i) \right]$$.
   We can see that:
   \begin{align*}
      \pi[f_1,\ldots,f_n,g](\mu,z). d\law{Z_m} = & \frac 1 {m^k} \cdot \mathbb{E}[g(Y)\cdot \prod_{1\leq i \leq n} \sum_{1\leq j \leq m} f_i(X_j)]  \\
     = & \frac 1 {m^k} \cdot \mathbb{E}[\sum_{(m_1,\ldots,m_n)\in \{1,\ldots,m\}^n}g(Y)\cdot \prod_{1\leq i \leq n}  f_i(X_{m_i})]\\
         = & \frac 1 {m^k} \cdot \sum_{(m_1,\ldots,m_n)\in \{1,\ldots,m\}^n}  \mathbb{E}[g(Y)\cdot f_1(X_{m_1}) \ldots f_n(X_{m_n})]
     \end{align*}
   Recall there the \emph{parametrized exchangeability} on the $((X_1, X_1,\ldots),Y)$: it says exactly that $\forall (m_1,\ldots,m_n)\in \{1,\ldots,m\}^n $,  $ \mathbb{E}[g(Y)\cdot f_1(X_{m_1}) \ldots f_n(X_{m_n})]=  \mathbb{E}[g(Y)\cdot f_1(X_{1}) \ldots f_n(X_{n})]$. We can conclude from there. 
    \end{itemize}
  \end{proof}
  \begin{proposition}
    Let $E,F$ standard Borel spaces.
    For $n \in \NN$, we note $\text{eq}_n: \distrone \in \Giry(\Giry E \otimes F)\mapsto \int_{(\distrtwo,z)} (\mu \times \ldots \times \mu \times \dirac z) d\distrone \in \Giry (E^n \times F)$.
    Let $(X1,X2,…)$ an infinite sequence of $E$-valued random variables, and $Y$ an $F$-valued random variable, such that $\forall \sigma$ permutation in $S_n$, $\law{X_1,\ldots,X_n,Y} = \law{X_{\sigma(1)},\ldots,X_{\sigma(n)}}$, and $D$ the measure in $\Giry(\Giry E^\omega \times F)$ obtained from there along Proposition~\ref{th:de_finetti_via_weak_convergence_parametrized} above. Then:
$$ \text{eq}_n(D) = \law{X_1,\ldots,X_n,Y}$$
\end{proposition}
\begin{proof}
We prove that using Proposition~\ref{th:de_finetti_via_weak_convergence_parametrized}, combined with separating families, this time on $\Giry (E^n \times F)$.
  \end{proof}

\begin{lemma}\label{lemma:de_finetti_aux_1_parametrized}
Let $(Y \xrightarrow{f}{X^\omega \otimes Z})$ that equalises the $\{\dot \sigma \otimes Z \mid \sigma \in S_n\}$. For any $y \in Y$, let us call $X_1^y,\ldots,X_n^y,\ldots$ the infinite exchangeable sequence of $X$-valued random variables defined by taking $(X^\omega\times Z,f(y))$ as probability space, and $X_i^y:((x_1,x_2,\ldots),z) \mapsto x_i$, and similarly $Z^y$ the $Z$-valued random variable $Z^y: :((x_1,x_2,\ldots),z) \mapsto z$ . Let $D^y \in \Giry{(\Giry X \times Z)}$ the probability distribution given by Theorem~\ref{th:de_finetti_via_weak_convergence_parametrized}. Then $\alpha: y \in Y \mapsto D^y \in \Giry (\Giry X\times Z)$ is a morphism in $\Stoch(Y,\Giry X\otimes Z)$. 
\end{lemma}
\begin{proof}
  First, by theorem~\ref{th:de_finetti_via_weak_convergence_parametrized}, we know that for every $y$, $\alpha$ is the pointwise weak limit of the functions $\beta_n:= y \in Y \mapsto \law{ \frac 1 n \cdot \sum_{i=1}^n \dirac{X_i^y},Z^y} \in \Giry (\Giry X \times Z)$. It is a known fact in measure theory that the pointwise limit of a sequence of measurable functions $R\rightarrow Q$ with $Q$ metrizable is again measurable--see e.g.~\cite{kechris} 11.2. Thus to conclude that $\alpha$ is indeed in $\MeasC(Y,\Giry (\Giry X \times Z))$, it is enough to show that each of the function $\beta_n$ is measurable, or in other terms that $\beta_n$ is a stochastic kernel in $\Stoch(Y,\Giry X \otimes Z)$.
 Observe that $\beta_n(y) = \sum_{i} \frac 1 n \cdot \law{\dirac {X_i^y},Z^y}$, thus it is enough to show that for any $i$, the function $(y \in Y \mapsto \law{\dirac {X_i^y},Z^y} \in \Giry(\Giry X\times Z)$ is measurable. We can see that for $y \in Y$, $\law{\dirac {X_i^y},Z^y} = b_{\star}(f(y))$, where $b_\star$ is the pushforward of the measurable function $b:((x_1,x_2,\ldots), z) \in X^\omega \times Z \mapsto (\dirac x_i,z) \in \Giry X\times Z$. In other terms, $\beta_n = b \circ f$, where $b$ is seen as a element in $\Stoch(X^\omega \times Z,\Giry X\times Z)$ via the embedding $\MeasC \rightarrow \Stoch$ that see any measurable function as a (deterministic) kernel. That concludes the proof.

  \end{proof}

  \end{proof}

We have now everything we need to apply the layered construction of the free commutative comonoid from~\cite{metabasson} to the category $\Stoch(\Pol)$, and we can thus conclude that $\Giry X$ is the free commutative commonoid generated by $X$.

\begin{theorem}\label{prop:dd_limit_in_stoch}
For any standard Borel space $X$, the De Finetti DD-chain over $X$ has $\Giry X$ as limit, and this limit commutes with the tensor product. Moreover, $\Giry X$ is the free commutative commonoid over $X$.
\end{theorem}
\begin{proofsketch}
  Though~\cite{metabasson} is primarily interested by models of (intuitionnistic) linear logic, thus symmetric monoidal \emph{closed} categories, the layered construction for the free commutative comonoid is proved there for any symmetric monoidal category, and as a consequence can be applied to $\Stoch(\Pol)$. Since $\unit$ is a terminal object in $\Stoch$, the unique morphism $X \rightarrow \unit$ is the free copointed object over $X$. 
  We proved in Proposition~\ref{prop:existence_of_equalisers_and_coequalisers_in_stoch} that the equalisers of symmetrie ox $X^{\otimes n}$ exists and commute with the tensor product, and we obtain by combining Theorem~\ref{th:de-finetti-perrone} and  Proposition~\ref{theorem:equivalence_of_two_df_formulations} that the limit of the draw-and-delete chain generated by $X\rightarrow \unit$ exists and commute with the tensor product, thus all conditions are met.
  \end{proofsketch}

\section{Draw and Delete Chains in Probabilistic Coherence spaces}\label{new_pcoh}

 Probabilistic coherence spaces (PCS), originally introduced by Girard~\cite{girard200410}, have been developed as a model of linear logic by Danos and Ehrhard~\cite{danos2011probabilistic}. As mentioned earlier, it was shown in~\cite{crubille2017free}
 that the layered construction of the free exponential works in $\pcoh$, equipping it with the structure of a Lafont model of Linear Logic\footnote{this free $!$ moreover  coincides with the exponential defined earlier in~\cite{danos2011probabilistic}}.

 The category $\pcoh$ of probabilistic coherence spaces contains as a full subcategory $\DistSubStoch$, the category of \emph{discrete countable} measurable spaces, and sub-stochastic kernels between them.
 There are thus two DD chains of interest in the category $\pcoh$. The first one is obtained by applying the embedding $\Stoch^{\leq \omega} \hookrightarrow \pcoh$ to the De Finetti chain over a discrete countable space $X$, while the second one is the DD-chain built over the free copointed object $A \with \unit$ for any PCS $A$, which as seen before is the one of the layered construction for the exponential $!A$.

  \subsection{Probabilistic Coherence Spaces}
 We give a very brief overview of the category $\pcoh$, only designed to give the reader some intuitions.
    In our proofs in Appendix, we need more formal definitions, that can be found there.

\begin{notation}\label{notation:weak_orthogonality}
Let us fix $W$ a countable set. If $u,v \in \Rp{W}$, we note $\scal u v := \sum_{a \in W}u_a\cdot v_a \in \Rp{}\cup \{+\infty\}$. For $A \subseteq \Rp{W}$, we note $A^{\perp}:= \{u \in \Rp{W} \mid \forall v \in A, \scal v u \leq 1\}.$
\end{notation}
A probabilistic coherence space consists in a subset of the non-negative cone of a countable dimensional vector space $\RR^{\web{X}}$. It must verify a number of requirements, the most important being the closure under \emph{bi-orthogonality} -- where two vectors with non-negative coefficients $x$ and $y$ are considered \emph{orthogonal} when $\scal x y \leq 1$.
  \begin{definition}\label{def:pcoh}
    A probabilistic coherence space (PCS) is a pair $X = (\web X, \clique X)$, where $\clique X \subseteq \Rp{\web X}$ such that $\clique{X}^{\perp\perp} = \clique X$, and moreover $\forall a \in \web X$, $c_X(a) := \sup_{u \in \clique X} u_a$ is such that $0 < c_X(a) < +\infty$. We call \emph{elements of $X$} the $x\in\clique X$, and \emph{web of $X$} the set $\web{X}$. A PCS morphism between $X$ and $Y$ is a matrix $f \in \Rp{\web X \times \web Y}$ such that for all $x \in \clique X$, $f\cdot x \in \clique{Y}$, where $\cdot$ is the usual product between the matrix $f$ and the column vector $x$.   We note $\pcoh$ for the category that has PCSs as objects, PCS morphisms as morphisms, and matricial product as composition.
  \end{definition}

   $\pcoh$ has the structure of a model of linear logic; we give below some of the constructions,that correspond to the cartesian and symmetric monoidal products. More details can be found in~\cite{danos2011probabilistic}.
  \begin{definition}[Model of Multiplicative Additive Linear Logic ]
    Let $X,Y$ be two PCS. We define $X\with Y$, $X \otimes Y$ respectively as:
    \begin{align*}
      & \web{X \with Y} := \web X \sqcup \web Y, \qquad \web{X \otimes Y} := \web X \times \web Y, \\ & \clique{X \with Y} := \{x \in \Rp{\web X \sqcup \web Y} \mid x_{\mid \web X} \in \clique{X}, \, x_{\mid \web Y} \in \clique Y \},\\
         & \clique{X \otimes Y} := \{x \otimes y\mid x \in \clique{X},y \in \clique{Y}\}^{\perp\perp},
    \end{align*}
    where $\web X \sqcup \web Y$ is the disjoint union of $\web X$ and $\web Y$, and for $x\in\Rp{\web X}$ and $y\in\Rp{\web Y}$ respectively, $x \otimes y \in \Rp{\web X \times \web Y}$ is defined as $(x \otimes y)_{(a,b)} := x_a \cdot y_b$, for $(a,b) \in \web X \times \web Y$.
    \end{definition}
    For $X_1,X_2$ two PCSs, $X_1 \with X_2$ is the cartesian product of $X_1$ and $X_2$. The projections morphisms are the $\pi_i:X_1\with X_2 \rightarrow X_i$ for $i \in \{1,2\}$, defined as $(\pi_i)_{a,b} := \delta_{a,b}$ for $a \in \web{X_1}\sqcup \web{X_2}$, $b \in \web{X_i}$ -- here $\delta$ is the Kronecker symbol, i.e.~$\delta_{a,b}=1$ if $a=b$, and $0$ otherwise.

  \begin{example}\label{ex:pcs}

 The monoidal unit of the symmetric monoidal structure defined on $\pcoh$ in~\cite{danos2011probabilistic} is the PCS $\unit$, defined by $ \web{\unit} = \{\star\}$, and $\clique{\unit}:=\{x \in \Rp{\{\star\}} \mid x_\star \in [0,1]\} \sim [0,1]$. The monoidal unit of $\pcoh$ \emph{is not} a terminal object. $\pcoh$ has also a terminal object, $0$, whose web is empty, and whose set of elements is $0$, the unique vector in $\RR_{\geq 0}^{\emptyset}$ .
 \end{example}

   \begin{definition}\label{def:wstoch_to_pcoh}
     For $X$ a discrete countable space, we define the PCS $ X^{\pcoh}:= (X,\{p \in \RR_{\geq 0}^X \mid \sum_{x \in X} p_x \leq 1\})$. For a morphism $k \in \SubStoch^{\leq \omega}(X,Y)$, we define $k^{\pcoh} \in \pcoh(X^{\pcoh},Y^{\pcoh} )$ as the transition matrix associated with the sub-stochastic kernel $k$. In particular, we note $\Bool:= \deux^{\pcoh}$ -- for $\deux = \{0,1\}$, to highlight that it is the natural semantics of the Boolean datatype.
   \end{definition}
   Observe that the elements of $X^{\pcoh}$ are exactly the sub-probability distributions over $X$:  this construction was already used to interpret the $\PCF$-data type Nat in~\cite{EhrhardPT18}.
Definition~\ref{def:wstoch_to_pcoh} defines a full and faithful functor $\DistSubStoch \hookrightarrow \pcoh$. It is moreover easy to check that it is strong monoidal; in particular, observe that the the monoidal unit of $\Stoch$ (the singleton space) is sent to the monoidal unit of $\pcoh$ from Example~\ref{ex:pcs}. 
 This functor also preserves small limits -- the proof can be found in Lemma~\ref{lemma:limits_preservation_substoch_pcoh} in appendix --  and the combined functor $\Stoch^{\leq \omega} \rightarrow \SubStoch^{\leq \omega} \rightarrow \pcoh$ preserves small \emph{connected} limits.

 \subsection{The Layered Construction in Probabilistic Coherence spaces}
 We present here briefly -- following~\cite{crubille2017free} -- how the layered construction for the free exponential plays out in $\pcoh$. The monoidal $\unit$ is not terminal in $\pcoh$, thus it is not possible anymore  to take as in $\Stoch$ the free copointed object as $(X,X\rightarrow \unit)$ with a unique $X\rightarrow \unit$. Rather, since $\pcoh$ has cartesian products, the free copointed object generated by any PCS $A$ is $A_\bullet = (A \with \unit, \pi_2:A \with \unit \rightarrow \unit)$.

       \begin{notation}
Consistently with Section~\ref{sect:connection_chains}, we note $(\Mn n {\pcs A},\eq n :{\Mn n {\pcs A}} \rightarrow {\pcs A}^{\otimes n} )$ for the equaliser of the symmetries over $\pcs A^{\otimes n}$ in $\pcoh$. When talking about the set of all multisets over a set $X$, we write $\MnSet n X$.
\end{notation}
An equalizer of the symmetries on $(A\with \unit)^{\otimes n}$ have been computed in~\cite{crubille2017free}: we recall it below:
\begin{proposition}[from~\cite{crubille2017free}]\label{def:mn_pcoh} Let  $\pcs A$ be a PCS.
  Let us define the PCS $\Mn n {\pcs A}$ as
 
  $\web{\Mn n {\pcs A}} := \MnSet n {\web{\pcs A}}$, and $ \clique{\Mn n {\pcs A}} := \{x \in \RR_{\geq 0}^{\MnSet n {\web{\pcs A}}}\mid \eq n.x \in \pcs A^{\otimes n}\}$   where $ (\eq n)_{\mu,(a_1,\ldots,a_n}) := \delta_{\mu, [a_1,\ldots,a_n]}$.
   Then $\eq n$ is a PCS morphism from $\Mn n A$ to $A^{\otimes n}$, and  $(\Mn n A,\eq n)$ is an equalizer of the symmetries on $A^{\otimes n}$.
 \end{proposition}
Applying Proposition~\ref{def:mn_pcoh} to PCSs of the form $A \with \unit$, we see that $\web{\Mn n A \with \unit} = \{\hat \nu \mid \nu \in \MnSet {\leq n} {\web{A}}\}$, where $\hat \nu$ is a notation to mean the multiset of size $n$ over $\web{A \with \unit} = \web{A}\cup\{\star\}$ obtained from a multiset $\nu$of size $\leq n$ by adding as many $\star$ as necessary. By abuse of notation, we will identify in the following $\web{\Mn n A \with \unit}$ and $\MnSet {\leq n} {\web{A}}$.
 
\begin{lemma}[From~\cite{crubille2017free}]\label{lemma:dd_morphisms_in_pcoh}
  The draw-and-delete morphisms over the copointed object $(\pcs A\with \unit, \pi_2:\pcs A \with \unit \rightarrow \pcs A)$ are the  $\DDn n: \Mn {n+1} { {\pcs A} \with \unit} \xrightarrow{} \Mn n {{\pcs A} \with \unit}$ given by
   $(\DDn n)_{\mu,\nu} := \delta_{\mu,\nu} $, for $\mu \in \MnSet{\leq n+1}{\web{\pcs A}}$ and $\nu \in \MnSet{\leq n}{\web{\pcs A}}$. 
\end{lemma}

From Proposition~\ref{def:mn_pcoh} and Lemma~\ref{lemma:dd_morphisms_in_pcoh}, a concrete expression for 
$!\pcs A$ -- the free commutative comonoid over $\pcs A$ -- can be computed as the limit of the approximations chain.

 \begin{proposition}[From~\cite{crubille2017free}]\label{prop:limit_morphisms}
   For a PCS $\pcs A$, the PCS $!\pcs A$ is specified by: $\web {!\pcs A}$ is the set of all finite multisets over $\web A$, and $\clique{! \pcs A} = \{x^! \mid x \in \clique {\pcs A}\}^{\perp \perp}$, where $x^!$ is the promotion of $x \in \clique A$, defined as $x^!_{[a_1,\ldots,a_n]} = \prod_{1 \leq i \leq n} x_{a_i}$. The limit morphisms $\rho_{\infty,n}:!\pcs A \rightarrow (\pcs A \with \unit)^{\otimes n}$ are defined as $(\rho_{\infty,n})_{\mu, \nu} := \delta_{\mu,\nu}$.
 \end{proposition}
We can observe that $!\pcs A$  is built as the bi-orthogonality closure of all the \emph{promotions} of elements in $\pcs A$. For instance, for an element $x$ in ${\deux}^{\pcoh}$ seen as a (sub)-probability distribution over $\{\ttrue,\ffalse\}$, its promotion $x^!$ should be understood as a program of type Bool that can be called several times, and at each call returns a sample drawn independently in the probability distribution $x$. 
   Promotions are enough to interpret probabilistic $\PCF$ into $\pcoh$, but the biorthogonality closure enforces the existence of more elements in $!\Bool$ than just promotions: we'll investigate the nature of those elements that \emph{are not} promotions in the last section.

   \subsection{\texorpdfstring{Connecting the De Finetti construction and the free exponential construction in $\pcoh$.}{Connecting the De Finetti construction and the free exponential construction in probabilistic coherence spaces.}}\label{subsect:connecting_in_pcoh}
   For a countable discrete space $X$ in $\Stoch^{\leq \omega}$, the image of $X_{\bullet}$ by the embedding $\Stoch^{\leq \omega} \rightarrow \pcoh$ \emph{is not isomorphic as copointed object} to the free copointed object  generated by $X^{\pcoh}$ in $\pcoh$, which is $(X^{\pcoh} \with \unit, \pi_2)$. Observe that this does not contradict our result that $\Stoch^{\leq \omega} \rightarrow \pcoh$ preserves \emph{connected} limits, because the free copointed object is a non-connected limit.

   It means that two different draw-and-delete chains can be built over $X^{\pcoh}$: the one generated by $(X_{\bullet})^{\pcoh}$, and the one generated by  $(X^\pcoh)_{\bullet} = (X^\pcoh \with \unit,\pi_2)$. Since $\Stoch^{\leq \omega} \rightarrow \pcoh$ preserves connected limits and is strong monoidal, the draw-and-delete chain over ${(X_\bullet)}^\pcoh$ is simply the image of $\DD^{DF}$ in $\pcoh$, that we will call the \emph{$\pcoh$-De Finetti chain}. 
   By our Proposition~\ref{th:chain_morphism_from_copointed_morphisms}, there exists a chain morphism from this chain to the draw-and-delete chain generated by the free copointed object $X^{\pcoh} \with \unit$.
We compute a concrete expression of this chain morphism in Proposition~\ref{prop:high-level_chain_morph_in_pcoh_bis} below.
      \begin{notation}
  For $X$ any set, and $\mu$ a finite multiset of elements of $X$, we note $\multinomial \mu$ for the multinomial coefficient of $\mu$, defined as the number of enumeration of $\mu$, i.e.~the number of different $n$-uples $(a_1,\ldots,a_n)$ such that $\mu = [a_1,\ldots,a_n]$.
 
  \end{notation}

\begin{proposition}\label{prop:high-level_chain_morph_in_pcoh_bis}
      For $X$ a discrete countable measurable space,
      the unique morphism of copointed object 
      $\alpha:= (X_{\bullet}^{\pcoh}) \rightarrow X^{\pcoh} \with \unit$
     induces a canonical chain morphism: $$(\Mn n \alpha : {\MnCat {\pcoh} n {X^{\pcoh}}}) \rightarrow \MnCat  {\pcoh} n {X^{\pcoh}\with \unit})_{n \in \NN}$$ from the $\pcoh$-De Finetti chain to the approximations $\DD^{!}$ chain over $X^{\pcoh}$, that can be concretely expressed as:
$$(\Mn n \alpha)_{\mu,\nu} = \multinomial{\mu- \nu} \text{ if }\nu \subseteq \mu, \, 0 \text{ otherwise.} .$$
\end{proposition}
        \begin{proofsketch}
          First observe that 
      $\alpha = \langle X^{\pcoh},(\_)^{\pcoh}\rangle:X^{\pcoh} \rightarrow  X^{\pcoh} \with \unit$, where $(\_)^{\pcoh}$ stands for the image under $\Stoch^{\leq \omega} \rightarrow \pcoh $ of the unique morphim $X \rightarrow \unit$ in $\Stoch$.      
      The $(\Mn n \alpha)_{n \in \NN}$ can be computed by combining: their definitions as the unique morphisms that make some diagram commutes in Definition~\ref{def:mn_alpha}, and the concrete description of the equaliser of $n!$ symmetries in $\pcoh$ in Proposition~\ref{def:mn_pcoh}.
    
          \end{proofsketch}

        Summing up, we have shown in Proposition~\ref{prop:high-level_chain_morph_in_pcoh_bis} that we can formally connect in $\pcoh$ the $\pcoh$-De Finetti DD-chain for $X$ and the layered chain for building the free exponential $!X^{\pcoh}$. However, it is not known whether the $\pcoh$-De Finetti chain has a limit even for $X=\deux$, so it seems complicated to have a version of De Finetti's theorem in $\pcoh$. 
This is why in the following we will work with $\Icones$~\cite{ehrhard2025integration}, the conservative extension of $\pcoh$ designed to deal with continuous data types. The intuition is that since the limit of the De Finetti chain in $\Stoch$ is the non-discrete space $[0,1]$, the limit of the $\ICones$-De Finetti chain for $X=\deux$ should be the interpretation of the continuous data type $[0,1]$.

\section{Integrable cones} \label{sec:icones}

In this Section, we present the category of integrable cones $\ICones$ as defined in~\cite{ehrhard2025integration}. It is a model of intuitionistic linear logic where it is possible to interpret a \emph{higher-order} language with continuous probabilistic primitives, and which is a conservative extension of $\pcoh$. The long-term hope of designing an extension of $\pcoh$ to deal with continuous probabilities was to obtain continuous counterparts to the full-abstraction results for discrete probabilistic higher-order languages~\cite{EhrhardPT18,ehrhard2019probabilistic} that hold in $\pcoh$, but this is however still an open problem.

All the definitions and results presented in this Section comes from~\cite{ehrhard2025integration}.

\subsection{The three categorical layers}
This definition is done via three successive layers, each a refinement of the previous one: $\Cones$, $\MCones$ and finally $\Icones$. $\Cones$ manages the \emph{algebraic} and \emph{order} structure; it can be thought as an extension of $\pcoh$ that allows for web-less, uncountable spaces. $\Mcones$ ensures that measurability at ground types is preserved, and $\Icones$ adds a notion of \emph{integration} on all objects required to interpret languages with call-by-value evaluation strategy. There are canonical fully faithful functors $\pcoh \rightarrow \Cones,\MCones,\ICones$. There are also canonical faithful functors $\SubStoch \rightarrow \Cones,\MCones$, but those are not full. The additional constraints on $\ICones$ equip it with a \emph{fully faithful} functor $\SubStoch(\Pol) \rightarrow \ICones$.
\paragraph{\texorpdfstring{The Category $\Cones$}{The category of cones and bounded Scott-continuous functions}}
\begin{definition}[\cite{ehrhard2025integration}, Cone]
    A \emph{cone} $C$ is an $\RRpos$-semimodule such that:
    \begin{enumerate}
        \item for all $x_1, x_2, y \in C$, $x_1 + y = x_2 + y \implies x_1 = x_2$, and 
        \item for all $x_1, x_2 \in C$, $x_1 + x_2 = 0 \implies x_1 = x_2 = 0$;
    \end{enumerate}
    together with  a function $\norm{\cdot}_C : C \to \RRpos$ called its \emph{norm} satisfying:

    \begin{itemize}
         \item \emph{Homogeneity}: for all $x \in C$ and $\lambda \in \RRpos$, $\norm{\lambda x}_C = \lambda \norm{x}_C$
        \item \emph{Separation}: for all $x \in C$, $\norm{x}_C = 0 \implies x = 0$
         \item \emph{Triangle inequality}: for all $x, y \in C$, $\norm{x + y} \le \norm{x}_C + \norm{y}$
        \item \emph{Positivity}: for all $x, y \in C$, $\norm{x} \le \norm{x + y}$
        \item \emph{Continuity}: $\norm{\cdot}_C$ is $\omega$-continuous, with respect to the partial order $\le_C$ defined by
        $x \le_C y$ when $\exists x' \in C,\; x + x' = y$.
        \end{itemize}
\end{definition}

We denote by $\ball C \defeq \set{ x \in C \suchthat \norm{x}_C \le 1}$ the unit ball of $C$.  The cone axioms make $\ball C$ convex, down-closed and  $\omega$-closed (the order $\le_C$ allows to use the notion of $\omega$-continuity coming from domain theory).
Cones morphisms are linear and $\omega$-continuous functions between cones, that moreover preserve the unit ball:
\begin{definition}[\cite{ehrhard2025integration}]
  If $P$ and $Q$ are cones,  a \emph{cone morphism} $f:P \rightarrow Q$ is a linear and $\omega$-continuous function $P \to Q$, which moreover is \emph{1-bounded}, i.e.~sends the unit ball of $P$ into the unit ball of $Q$.

 The category $\Cones$ has cones as objects and cone morphisms as morphisms.
\end{definition}

\paragraph{\texorpdfstring{The category $\MCones$}{The category of measurable cones}}

The category $\MCones$ of \emph{measurable} cones is a refinement of $\Cones$, designed to enforce that morphisms behave well with respect to measurability at ground types.

$\MCones$ is parameterised by the choice of $\Ar$, an \emph{essentially small}\footnote{The construction in~\cite{ehrhard2025integration} is done with a small category $\Arity$. We show in Proposition~\ref{prop:essentially_small_is_ok} in appendix how we can relax this hypothesis to essentially small.} full subcategory of $\Meas$ that we'll call \emph{arity category}, closed under cartesian products and containing the one-point measurable space, that we note $0$. 
In the following, we will always take $\Ar$ to be the category $\Pol$ of standard Borel spaces. $\Pol$ is indeed essentially small, since between any two uncountable standard Borel spaces there is a measurable isomorphism (see e.g.~\cite{kechris}, Theorem 15.6). 

\begin{definition}\label{def:measurability_structure}
    A \emph{measurable cone} is a cone $C$ equipped for each object $R$ of $\Ar$ with a set of $R$-tests $\mtest R^C \subseteq \text{Hom}_{\Set}(R, (\text{Hom}_{\Cones}(C, \RR_{\geq 0}))$, satisfying:
     \begin{itemize}
        \item \emph{Measurability}: for all $R \in \Arity$, $m \in \mtest R$ and $x \in \ball C$, $m(\cdot, x) \in \Meas(R, \interval{0}{1})$
        \item \emph{Composability}: for all $R,S \in \Arity$, $m \in \mtest R$ and $\phi \in \Ar(S, R)$, $m(\phi(\cdot), \cdot) \in \mtest S$
        \item \emph{Separation}: for all $x, y \in C$, $\forall m \in \mtest 0,\; m(x) = m(y) \implies x = y$
        \item \emph{Respect of the norm}: for all $x \in C$,
            $\norm{x} = \sup_{m \in \mtest 0 \setminus 0} \frac{m(x)}{\norm{m}}.$
    \end{itemize}
A \emph{measurable $R$-path on $C$} is a function $\beta : R \to C$ such that \begin{enumerate}\item $\beta(R)$ is bounded in $C$, and moreover\item  for every object $S$ of $\Ar$, and all tests $m \in \mathcal{M}^C_S$, the function
        $(s,r \in S \times R) \mapsto m(s, \beta(r)) \in \MeasC(S \times R, \RRpos).$\end{enumerate} We note $\Path(R,C)$ for the set of all measurable $R$-paths on $C$.
  
\end{definition}

The category $\MCones$  has measurable cones as objects and as morphisms the  $f \in \Cones(P, Q)$ that \emph{preserve the measurability of paths} i.e.~  $\forall \beta \in \Path(R, P)$, $f \circ \beta \in \Path(R,Q)$. 

There is a canonical way to make the space of finite measure over any measurable space into a measurable cone.
\begin{definition}\label{def:meas_paths}
  Let $(X,\Sigma_X)$ a measurable space. We note $X^{\Cones}$ for the cone whose underlying space is $\FMeas{(X)}$, the algebraic operations $+,\cdot$  are defined pointwise, the norm as $\norm \distrone = \distrone(X)$. We note $X^{\MCones}$ for the measurable cone with $X^{\Cones}$ as underlying cone, and whose  measurability tests are:
 
  $$\mtest d = \{f_U:(a,\distrone) \in d \times \FMeas (X) \mapsto \distrone(U) \mid U \in  \Sigma_X \}. $$
  We can moreover extend $( \cdot)^{\MCones}$ in a faithful functor  $\SubStoch \rightarrow \Mcones$ by defining  $f^{\MCones} := (\distrone \in \FMeas(X)  \mapsto f_{\star}(\mu) \in \FMeas(Y))$ for $f$ a sub-stochastic kernel $X \rightarrow Y$. 
\end{definition}
For $X,Y \in \Arity$, the morphisms $f$ in $\MCones(X^{\MCones},Y^{\MCones})$ coincide with the measurable functions $\Giry X \rightarrow  \Giry^{\leq 1} Y$ that moreover are \emph{linear} -- where $\Giry^{\leq 1}$ is the Panangaden monad~\cite{panangaden1998probabilistic}, the analogue of the Giry monad $\Giry$ for the \emph{sub-probability distributions}. If $X$ is a \emph{countable discrete} measurable space, the linearity constraint implies that $f$ is entirely described by its value on Dirac distributions: in other terms $f$ is the push-forward of a sub-stochastic kernel $X \rightarrow Y$. It is however not always the case when considering an arbitrary $X$, thus the embedding $\SubStoch(\Arity) \rightarrow \Mcones$ is not full.

\paragraph{\texorpdfstring{The category $\ICones$}{The category of integrable cones}}
An integrable cone is a measurable cone $C$ such that every measurable path $R \rightarrow C$ can be \emph{integrated} with respect to any finite measure on $R$, producing an element in $C$.

\begin{definition}
    Let $C$ be a measurable cone, $R \in \Ar$, $\beta : R \to C$ a measurable path and $\mu \in \FMeas(R)$, the set of finite measures over $R$. An \emph{integral of $\beta$ over $\mu$} is an element $x \in C$ such that
        $\forall m \in \mathcal{M}^C_0,\; m(x) = \int m(\beta(r)) \mu(\dif r)$.
    By the axioms on the measurability tests of $C$, such an integral is unique: we note it $\cint^C \beta \dif\mu$, or sometimes $\cint^C_{r \in R}\beta(r)\cdot \mu(dr)$.
    A measurable cone $C$ is \emph{integrable} if any measurable $R$-path on $C$ is integrable over any finite measure on $R$.
  \end{definition}
 
    The category $\ICones$ has integrable cones as objects and as morphisms the $f \in \MCones(P, Q)$ that \emph{preserve integrals} i.e.~for every measurable path $\beta : R \to P$ and $\mu \in \FMeas(P)$, $f\left(\cint^P \beta \dif\mu\right) = \cint^Q (f \circ \beta) \dif\mu$.

\begin{example}
  A simple example of integrable cone is the cone $\unit$ of non-negative reals, that will be the unit of the monoidal product.  Addition and scalar multiplication are as usual, and $\norm{\lambda} = \lambda$.
For $R \in \Arity$, there is only one $R$-test: the function $((r,x) \in R\times \RR_{\geq 0} \mapsto x \in \RR_{\geq 0})$. There, integration in the sense of $\Icones$ coincides with integration in the sense of measurable spaces:
  let $e \in \Arity$, $\mu \in \FMeas(e)$, and $\gamma \in \pathes e {\mathbf 1}$.
  Then $\gamma: e \rightarrow \RR_{\geq 0}$ is a measurable (and integrable) function, and $\cint^{\unit} \gamma(t)\cdot \mu(dt) = \int \gamma(t) \cdot \mu(dt)$. 
    \end{example}
    For any measurable space $X$, $X^{\MCones}$ is an integrable cone -- that we will thus note also $X^{\ICones}$. To understand the notion of integration on this cone, observe first that the measurable paths  $R \rightarrow X^{\MCones}$
     are exactly the bounded kernels from $R$ to $X$.

    Then we can see that integrating a path consists in pushing-forward the corresponding kernel.
\begin{lemma}\label{lemma:meas_paths_are_kernels2}
  Let $R\in \Arity$, $X$ any measurable space.

  The measurable \emph{1-bounded} paths $\gamma:R\rightarrow X^\MCones$
  are exactly the morphisms in $\SubStoch(R,X)$, and for $\mu$ a sub-probability measure on $R$, $\cint_{R}^{\FMeas(X)} \gamma\cdot d\mu = \gamma_{\star}(\mu)$, where $\gamma_{\star}$ is the usual push-forward lifting of the kernel $\gamma$.
\end{lemma}
As a corollary of Lemma~\ref{lemma:meas_paths_are_kernels2} -- by considering the measurable path $\gamma:=(x\in X\rightarrow \dirac x \in X^{\MCones})$ --  we see that for $X\in \Arity$, the morphisms $X^{\MCones} \rightarrow X^{\MCones}$ that preserves integrals are exactly the push-forward of sub-stochastic kernels in $X \rightarrow Y$. It means that the embedding $\SubStoch(\Arity) \rightarrow \MCones$ can be decomposed as $(\ICones \rightarrow \MCones) \circ (\SubStoch(\Arity) \rightarrow \ICones)$, where $(\ICones \rightarrow \MCones)$ is the obvious forgetful functor, and moreover  $(\SubStoch(\Arity) \rightarrow \ICones)$ is \emph{fully faithful}.

    \subsection{Monoidal product of integrable cones}
     The monoidal structure on $\Icones$ is defined abstractly in~\cite{ehrhard2025integration}. The tensor product bifunctor is built by first defining concretely an internal Hom  $A\multimap B$, and then using some adjoint functor theorem to show the existence of a left-adjoint for $\multimap$. This construction is the one presented in Eilenberg-Kelly~\cite{EilenbergKelly}, to obtain a (symmetric) tensor product starting from a (symmetric) \emph{closed category} $(\Catone,\multimap)$: it consists in proving that for every object $A$, the functor $(A \multimap \_)$ is a right adjoint, and that moreover it is an \emph{internal adjunction}, in the sense that it is witnessed by an isomorphism between internal Hom in the category $\ICones$. 
     
     \begin{toappendix}
     \begin{definition}[Linear Arrow]
  Let $P,Q \in \ICones$. The measurable cone $P \multimap Q$ is   defined as:
  \begin{itemize}
  \item its underlying cone is the sub-cone of the cone $P \multimap Q$ consisting of those functions that 
   preserve measurable paths, and commute with integrals;

  \item the  measurable tests on $P \multimap Q$ are defined from the tests on $Q$ and the paths on $P$ as:
    $$\mtestc R {P \multimap Q} = \{t \triangleleft \gamma \mid t \in \mtestc R Q, \gamma \in \pathes R {P} \}, $$
    where $t \triangleleft \gamma$ is the function $R \times {(P\multimap Q)} \rightarrow [0,1]$ defined as $t \triangleleft \gamma(r,f) := t(r,f(\gamma(r)))$.
    \end{itemize}
  \end{definition}

  \begin{proposition}[From~\cite{ehrhard2025integration}]\label{prop:closure_of_icones_v}
    If $P,Q$ are integrable cones, then, $P \multimap Q$ is also an integrable cone, and moreover
     this construction can be extended into a bifunctor $\multimap:\ICones^{\text{op}} \times \ICones \rightarrow \ICones$.
\end{proposition}

  We give here a more precise version of Proposition~\ref{prop:closure_of_icones_v}, that we will need for our proof of Theorem~\ref{prop:monoidal_structure_on_icones} in appendix.
    \begin{proposition}\label{prop:closure_of_icones}[From~\cite{ehrhard2025integration}]
    For  $P,Q$ in $\ob\Icones$, $P \multimap Q$ is also in $\ob\Icones$, and 
    $\cint^{P \multimap Q}_R \gamma d\mu = (x \in P \mapsto \cint^{Q}_R \gamma(x) d\mu),$ $\forall R \in \Arity,\,\mu \in \FMeas(R), \gamma:R \rightarrow(P \multimap Q)$  measurable path.
 Moreover, $\multimap$ can be extended into a bifunctor $\multimap:\ICones^{\text{op}} \times \ICones \rightarrow \ICones$ with   $f\multimap g:=(\lambda \in P_1 \multimap Q_1)\mapsto (g\circ \lambda\circ f) \in P_2\multimap Q_2$ for $f:P_2\rightarrow P_1$, and $g:Q_1 \rightarrow Q_2$.
\end{proposition}
 
    \begin{theorem}\label{prop:monoidal_structure_on_icones}[From~\cite{ehrhard2025integration}]

      There exists an $\Icones$ monoidal product $\otimes:\Icones \times \Icones \rightarrow \Icones$, such that $(\Icones,\otimes,\multimap)$ is a symmetric monoidal closed category.

\end{theorem}

In~\cite{ehrhard2025integration}, the category $\Arity$ is required to be \emph{small} -- instead of \emph{essentially small}, as we require in the present paper. As highlighted in~\cite{ehrhard2025integration}, the only place where the size of the $\Arity$ category plays a role is in applying the adjoint functor theorem to $(P \multimap \_)$, and more precisely in the proof that $\ICones$ is well-powered. Technically, what is used is that since $\Arity$ is a set, for every integrable cone $C \in \ob{\ICones}$ the class of all integrable $D$ such that  $\underline D \subseteq \underline C$ is a set -- where $\underline D$ and $\underline C$ are the underlying sets of $D$ and $C$ respectively.

    \begin{lemma}\label{prop:essentially_small_is_ok}
Let $\Arity$ an essentially small sub-category of $\MeasC$,  closed under cartesian product and that contains the one-point measurable space. Let $C$ an object in $\ICones$. Then the class of all integrable cones $D$ such that  $\underline D \subseteq \underline C$ is a set. 
\end{lemma}

\begin{proof}
  For any measurable cone $C$, an morphism $f:R_2\rightarrow R_1$ m in $\Arity$ induces a function $\mtest{R_1}^C \rightarrow \mtest{R_2}^C$, defined by: $m\in \mtest{R_1}^C \mapsto m \circ (f\times P) \in  \mtest{R_2}^C$ -- for this, we use the axiom \textbf{(Mscomp)} from~\cite{ehrhard2025integration}.  In other terms, $\mtest{}^C:\Arity^{\text{op}} \rightarrow \Set$ is a pre-sheaf on $\Arity$. Moreover, it is a sub-presheaf of the presheaf
  \begin{align*}
    M^C:=& \Arity^{\text{op}} \rightarrow \Set \\
         & R \mapsto \{f:R\times C \rightarrow[0,1]\}
  \end{align*}
  Since $\Arity$ is essentially small, the presheaf category $[\Arity^{\text{op}},\Set]$ is well-powered (since $\Arity$ is essentially small, $[\Arity^{\text{op}},\Set]$ is both a topos and locally small)
  We can deduce from there that if we fix an object $C$ in $\Cones$, the class of all the measurability structures on the cone $C$, that are in particular sub-objects of $M^C$, form a set. Then the class of (all integrable cones $D$ such that  $\underline D \subseteq \underline C$) can be expressed as: $(\cup_{\underline D \subseteq \underline C} \cup_{D = (+,.,\norm{})\text{ a cone-structure on }\underline D} \{\text{measurability structures on }D\})$, and thus it is a set.
  \end{proof}

 \end{toappendix}

 \subsection{Embedding probabilistic coherent spaces into integrable cones}
 As shown in \cite{ehrhard2025integration} there exists also a full and faithful  functor, that moreover is strong monoidal $\ic : \pcoh \to \ICones$.
 $\ic$ is compatible with the other inclusions functors we presented, in the sense that $\SubStoch^{\leq \omega} \hookrightarrow \SubStoch \hookrightarrow \Icones = \SubStoch^{\leq \omega}\hookrightarrow{} \pcoh \xhookrightarrow{ic} \Icones$. Besides, the functor $c:=\pcoh  \xhookrightarrow{ic} \ICones \to \Cones$ is again fully faithful.
 
 Intuitively, we build a cone from a PCS $X$ by \emph{extending} the set of $X$-elements by scalar multiplication, and choosing the norm in such a way that $\clique X$ becomes exactly the unit ball of the cone. The measurability structure is obtained by using the elements of $\clique X^{\bot}$ as \emph{measurability tests} for paths.
\begin{definition}\label{def:ic}
  For $X$ a PCS, we define $X^{\Cones}$ as the cone:
  % with underlying set
  \[ X^{\Cones}: = \set{ \lambda \cdot x \setcompr \lambda \in \RRpos\mid x \in \clique X }; \qquad
  % equipped with the norm
    \norm{x}_{X^{\Cones}} \defeq \inf \{\lambda > 0 \mid \frac{1}{\lambda}\cdot x \in \clique X\}.\] We note $X^{\MCones}$ for the object in $\MCones$ that has $X^{\Cones}$ as underlying cone, and as measurability tests: 
    $\mathcal{M}^{X^{\MCones}} := \{ \scalar{\cdot}{x'} \setcompr x' \in \clique X^\bot\}.$ 
  \end{definition}
  For any PCS $X$, $X^{\MCones}$ is an integrable cone, that we will thus  note also $X^{\ICones}$. Moreover
  we can extend Definition~\ref{def:ic} into functors $c:\pcoh \rightarrow \Cones$ and $\ic:\pcoh \rightarrow \ICones$ by defining for $f \in \pcoh(X,Y)$, $c(f) = \ic(f) = (x \in X^{\Cones} \mapsto f.x \in Y^{\MCones})$.
 
\begin{theorem}[\cite{ehrhard2025integration}] \label{thm:ic-full-faithful}
  Both $c$ and $\ic$ are full and faithful functors, and $\ic$ is strong monoidal. 
  \end{theorem}

\section{De Finetti theorems in Integrable Cones}\label{sec:df_in_icones}
Recall that in Section~\ref{new_pcoh}, we proved that from a countable discrete space $X$, we were able to generate two draw-and-delete chains from $X^{\pcoh}$: the $\pcoh$ De Finetti chain, and the approximations chain for $!X^{\pcoh}$. The fact that no limit is known for the $\pcoh$-De Finetti chain stopped us going further.  
Our goal now is to overcome this roadblock by going from  $\pcoh$ to  $\ICones$, where all limits exist.

To use the structural proximity of the draw-and-delete chains to deduce a formal connection between $\Giry X$ and $!X^{\pcoh}$ this way, however, we need to prove that \begin{enumerate} \item\label{item:goalun} the $\Icones$-De Finetti draw-and-delete chain is indeed  the  image by $\Stoch(\Arity) \rightarrow \ICones$  of the $\Stoch$-De Finetti chain, and that its limit is indeed $\Giry X^{\ICones}$, and that\item\label{item:goaldeux} the $\ICones$ approximation chain over $X^{\ICones}$ is indeed the image by the embedding $\pcoh \rightarrow \ICones$ of the approximation chain for $!X^{\pcoh}$, and that $(!X^{\pcoh})^{\ICones}$ is indeed its limit. \end{enumerate} 
In the present Section, we tackle~\eqref{item:goalun}, by proving that $\Stoch(\Arity) \rightarrow \ICones$ is strong monoidal and preserves connected limits. We will tackle~\eqref{item:goaldeux} in the next Section~\ref{sect:total_elements}. 
\subsection{\texorpdfstring{The embedding $\SubStoch(ar) \hookrightarrow \ICones$ is strong monoidal.}{The embedding from sub-stochastic kernels to integrable cones is strong monoidal.}}
Recall that a functor $F:\Catone \rightarrow \Cattwo$ between two symmetric monoidal categories is strong symmetric monoidal when for every pair of objects $A,B$ in $\Catone$, there exists an isomorphisms $F A\otimes FB \sim F(A\otimes B)$ natural in $A$ and $B$, and that moreover respects the coherences and symmetries. 

\begin{theorem}\label{th:substoch_to_icones_is_strong_monoidal}
$C:\SubStoch(\Arity) \hookrightarrow \Icones$ is a strong monoidal functor.
\end{theorem}
\begin{proofsketch}
It has been shown in~\cite{ehrhard2025integration} (Theorem 7.9) that for any $X,Y \in \Arity$, there is an $\ICones$-isomorphism  $\FMeas(X\times_{\Arity} Y) \sim \FMeas X \otimes \FMeas Y$, which is natural in $X$ in $Y$ (in $\Arity)$. This isomorphism is defined by composing several isomorphisms built in the course of the proof of the monoidality of the category $\ICones$. It seems likely that their proof could be adapted to show that this isomorphism  is still  natural (in $\SubStoch(\Arity))$  when seen as $\FMeas(X\otimes_{\SubStoch(\Arity)} Y) \sim \FMeas X \otimes \FMeas Y$.

In Appendix~\ref{section_appendix:strong_monoidal}, we follow another approach to prove that $\SubStoch(\Arity) \rightarrow \ICones$ is symmetric strong monoidal:
we highlight how the adjunction $(\_\otimes A) \dashv (A \multimap \_)$ lead to a universal property for $\otimes$, and prove that the universal property for $\FMeas X \otimes \FMeas Y$ holds for  $\FMeas(X\times_{\Arity} Y)$. That implies the existence of a natural isomorphism $\FMeas(X\otimes_{\SubStoch(\Arity)} Y) \sim  \FMeas X \otimes \FMeas Y$, and it can be also easily deduced from the universal property that it preserves coherences and symmetries. 
\end{proofsketch}

  \begin{remark}
    It would not be possible to have the analogue of Theorem~\ref{th:substoch_to_icones_is_strong_monoidal} for the larger category $\Cones$ instead of $\ICones$. There are two reasons for this:
    \begin{itemize}
    \item first, to the best of our knoweldge, the existence of a monoidal product $\otimes$ in $\Cones$ -- that would be  left adjoint to the internal Hom functor -- has not been investigated in the literature. It is not immediate how the arguments in~\cite{ehrhard2025integration} could be adapted from $\ICones$ to $\Cones$: for instance,~\cite{ehrhard2025integration} uses the $\ICones$ measurability tests to show that the hypothesis of the Special Adjoint Functor Theorem (SAFT) are met.
    \item Second, even assuming the existence of such a monoidal product on $\Cones$,  the functor $\Stoch(\Arity) \rightarrow \Cones$ could not be strong monoidal. Indeed, as said above, the strong monoidality of this functor boils down to 
      the universal property of the tensor product  $\FMeas X \otimes \FMeas Y$ holding for the cone $\FMeas(X \times_{\Arity} Y)$. This is not true in $\Cones$; to highlight what happens there, let us spell out this universal property: 
      for any cone $C$, for any $g:\FMeas X,\FMeas Y \rightarrow C$ bilinear and Scott-continuous (but with no constraints on measurability or integrability), it would require the existence of a \emph{unique} $\Cones$ morphism $h:\FMeas(X \times_{\Arity} Y) \rightarrow C$ such that for any finite measures $a, b$ on $X$ and $Y$ respectively, $g(a,b) = h(a \times b)$. What breaks in $\Cones$ is the \emph{unicity} part, because unlike in $\ICones$ it isn't the case that any $\Cones$ morphism $h:\FMeas(X \times_{\Arity} Y) \rightarrow C$ is uniquely specified by its values on the product measures of shape $a \times b$. A counter-example can be found in Lemma~\ref{lemma:cones_vs_icones} in Appendix.

      \end{itemize}
    \end{remark}
    \begin{toappendix}
    \begin{lemma}\label{lemma:cones_vs_icones}
There exist $X,Y$ standard Borel spaces, and  two distinct $\Cones$ morphisms $h_1,h_2:\FMeas(X \times_{\Arity} Y) \rightarrow C$ such that for any finite measures $a, b$ on $X$ and $Y$ respectively, $h_1(a \times b) = h_2(a \times b)$.
      \end{lemma}
      \begin{proof}
        First, we fix $X = Y = [0,1]$, and $C = \RR_{\geq 0}$. The first mophism $h_1$ we consider is simply the $0$ morphism. 
      Let us build now a non-zero $h_2: \FMeas(X \times_{\Arity} Y) \rightarrow C$ such that $\forall a,b$,  $h_2(a \times b) = 0$.
      We're going to use the Lebesgues-Radon-Nikodym decomposition theorem, that says that for any measurable space $Z$, and any fixed finite measure $\mu_0$ on $Z$, then any other finite measure $\nu$ on $Z$ can be uniquely written as $\nu = \nu_0 + \nu_1$, where $\nu_0$ is absolutely continuous with respect to the reference measure $\mu_0$, while $\nu_1$ is singular with respect to $\mu_0$. For clarity, we recall the definition of absolute continuity and singularity for measures on some measurable space $X$:
      \begin{itemize}
      \item $\kappa$ is absolutely continuous with respect to $\mu$, noted $\kappa \ll \mu$ whenever for every measurable $A$, $\mu(A) = 0$ implies $\kappa(A)=0$.
        \item $\kappa$ is singular with respect to $\mu$, noted $\kappa \perp \mu$ whenever there exists a measurable $A$, with $\kappa(A) = 0$ and $\mu(X \setminus A) = 0$.
        \end{itemize}

      Our first step is to build a reference measure on $X \times Y = [0,1]\times [0,1]$: for this, we take $\mu_0:=(r \in [0,1]\mapsto (r,r) \in [O,1]^{2})_{\star}(\lambda)$, where $\lambda$ is the uniform measure on $[0,1]$ -- $\lambda([a,b]) = \lvert a - b\rvert$. In other words, $\mu_0$ is the measure concentrated on the diagonal $\Delta = \{(r,r) \mid r \in [0,1]\}$, and which is the uniform measure on $\Delta$.

      From there, we build $g_2$ as:
      $$g_2(\nu) := \nu_0(X\times Y ) \text{ with }\nu = \nu_0 + \nu_1 \text{ given by the decomposition theorem.}$$
      First, we can check that $g$ is indeed a $\Cones$ morphism. The main tool here is the \emph{unicity} of the decomposition.

      We are now going to show that for any $\kappa_1,\kappa_2$ measures on $[0,1]$,  $g(\kappa_1 \times \kappa_2)=0$. So let $\kappa_1 \times \kappa_2 = \nu + \eta$, where $(\nu,\eta)$ is given by the decomposition theorem with respect to the reference measure $\mu_0$ on $[0,1] \times [0,1]$. The problem now consists in computing $\nu(I^{2})$.
      First, since the measure $\mu_0$ is concentrated on the diagonal, $\nu(I^{2} \setminus \Delta) = 0$, it holds that $\nu(I^{2}) = \nu(\Delta)$.
      \begin{itemize}
\item The first step consists in proving that $\mu_1 \times \mu_2(\Delta) = \sum_{r \text{ s.t. }\mu_1(\{r\}) >0 \wedge \mu_2(\{r\})>0}\mu_1(\{r\}) \cdot \mu_2(\{r\})$ -- and the index set is countable, because any finite measure has at most a countable number of non-zero weighted singletons. We prove this now:
  observe that by Fubini theorem, $\mu_1 \times \mu_2(\Delta) = \int_{(r_1,r_2) \in I^{2}} \delta_{r_1,r_2}\cdot d(\mu_1 \times \mu_2) = \int_{r_1 \in I} (\int_{r_2 \in I} \delta_{r_1,r_2}\cdot d\mu_1)\cdot d\mu_2 = \int_{r_1 \in I} (\mu_1(\{r_1\}))\cdot d\mu_2 $ -- $\delta_{(r_1,r_2)} = 1 $ whenever $r_1=r_2$, and $0$ otherwise. Since $\mu_1$ is finite, there is at most a countable numbers of singleton with non-zero $\mu_1$ measure. Let us write $X_n$ for the countable set of those $r$ such that $\mu_1(\{r\}) >0$;  we can rewrite $\mu_1 \times \mu_2(\Delta) = \sum_{r_n \in X_n} \mu_1(\{r_n\}) \cdot \mu_2(\{r_n\})$. From there, we can rewrite: $\mu_1 \times \mu_1(\Delta) = \sum_{r \text{ s.t. }\mu_1(\{r\}) >0 \wedge \mu_2(\{r\})>0}\mu_1(\{r\}) \cdot \mu_2(r)$.
\item Let us note now $\alpha$ for the measure on $I^{2}$, $\alpha = \sum_{r \text{ s.t. }\mu_1(\{r\}) >0 \wedge \mu_2(\{r\})>0}\mu_1(\{r\}) \cdot \mu_2(r)\cdot \dirac{(r_1,r_1)}$. It is immediate to see that $\mu_0 \perp \alpha$. From there (and using the unicity of the decomposition), we can prove that $\alpha \leq \eta$ : indeed we can decompose $\mu_1 \times \mu_2 - \alpha = \beta + \beta'$, where $\beta \ll \mu_0$ and $\beta'\perp \mu_0$, and then $(\beta, \beta' + \alpha)$ is a valid decomposition)
  \item Since $\alpha$ is contained in the singular part, and $\alpha(\Delta) = \mu_1 \times \mu_2(\Delta)$, we obtain that $\nu ( \Delta) = 0 $, which allows us to conclude.
\end{itemize}
      
      \end{proof}
\end{toappendix} 
\begin{toappendix}
  \subsection{\texorpdfstring{The embedding $\SubStoch(ar) \hookrightarrow \ICones$ is strong monoidal.}{The embedding from sub-stochastic kernels to integrable cones is strong monoidal.}}\label{section_appendix:strong_monoidal}

  \begin{notation}\label{def:functor_clin}
    We note $\functorClin : \SubStoch(\Arity) \rightarrow \Icones$ as follows for the fully faithful embedding.
  \end{notation}

\subsubsection{Universal property of the tensor product.}
Recall that the monoidal product on $\Icones$ is built abstractly as the adjoint of the internal Hom. As well-known in the literature, we can express equivalently this adjunction by a universal property on the monoidal functor. Let us first recall the definition of adjunctions using universal morphisms -- see e.g.~\cite{asperti:hal-03316030}, Theorem 5.2.1.
\begin{proposition}\label{prop:adjunction_as_univ_morphism}
  Let $F:\Catone \rightarrow \Cattwo$ a functor. Then $F$ is a right adjoint if and only if for any object $X$ in $\Cattwo$, there exists a \emph{universal morphism} from $X$ to $F$, i.e. a pair $(GX,\eta_X)$ where $A \in \ob \Catone$, and $\eta_X:X \rightarrow F(GX)$ with the following universal property: $\forall Y \in \ob{\Catone},\forall g:X \rightarrow F(Y)$, $\exists !h:GX \rightarrow Y  $ such that:
  $$
    \begin{tikzcd}
      F(GX) \arrow[d,dashed,"\exists !F(h)"] & & X \arrow[ll,"\eta_X"]  \arrow[lld,"g"]\\
      F(Y) 
\end{tikzcd}
$$
Then a left adjoint of $F$ is obtained as taking $(X\mapsto GX)$ as its actions on objects, and for $f:X\rightarrow Y$, $Gf$ is the unique morphism given by applying the universal property of $(GX,\eta_X)$:
  $$
    \begin{tikzcd}
      F(GX)  \arrow[d,dashed,"\exists !"]& & X \arrow[ll,"\eta_X"] \arrow[d,"f"]\\
      F(GY) & & Y  \arrow[ll,"\eta_Y"]
\end{tikzcd}
$$
  \end{proposition}

Let's now apply this view of adjunction via universal morphism to a monoidal closed category $(\Catone,\otimes,\multimap)$.  For any object $A$, the functor $F :=(A \multimap \_ ): \Catone \rightarrow \Catone$ is right adjoint to the functor $(\_ \otimes A)$. Proposition~\ref{prop:adjunction_as_univ_morphism} applied to this adjunction gives us:
\begin{corollary}\label{prop:univ_product_tensor}
  In a monoidal closed category $(\Catone,\otimes,\multimap)$, for any objects $A,B$ of $\Catone$, $A \otimes B$ is caracterised (up to isomorphism) via the following universal property: $\forall$ object $C$, $\forall g:A\rightarrow(B\multimap C)$:
 
   $$
    \begin{tikzcd}
      B \multimap (A \otimes B) \arrow[d,dashed,"\exists !( B\multimap h)"] & & A \arrow[ll,"\eta_{A,B}"] \arrow[lld,"g"]  \\
      B \multimap C 
\end{tikzcd}
$$
and moreover for $f:(A_1\rightarrow A_2)$, $(f\otimes B)\circ \eta_{A_1,B} = \eta_{A_2,B}\circ f$.

\end{corollary}
In vector spaces, the universal property caracterising the tensor product is exprimed with \emph{bilinear maps}. A similar role as the one of bilinear maps is played in a closed category by
the morphisms $X \rightarrow (Y \multimap Z)$.
Conveniently, a description of those morphisms in $\ICones$ was proved in~\cite{ehrhard2025integration}, that we recall below:
\begin{proposition}[from~\cite{ehrhard2025integration}]\label{prop:n_linear_ci_maps}
  Let $A_1,A_2,B$ three $\Icones$ objects.
  Then the morphisms in $\Icones(A_1,  (A_2 \multimap B))$ can be described as the $\Set$ functions: $A_1 \times_{\Set}  A_2 \rightarrow B$ such that  :
       
      \begin{itemize}
      \item $f$ is separately linear and $\omega$-continuous in each of its arguments;
        \item composing $f:A_1 \times_{\Set}A_2 \rightarrow B$ with  measurable paths $R_i\rightarrow A_i$ for $i \in \{1,2\}$ always gives a measurable path $R_1\times_{\Arity} R_2\rightarrow B$;
      \item $f$ commutes with integrals separately in each of its arguments\footnote{This requirement makes sense because of Fubini's theorem for integrable cones~\cite{ehrhard2025integration}, that says that multiple integrals can be permuted.};
      \end{itemize}
In the following, we'll call such $\Set$-maps \emph{bilinear ci-maps}.      
       
\end{proposition}
We can now reformulate the tensor universal property from Corollary~\ref{prop:univ_product_tensor}, using the notion of bilinear ci-map.
\begin{proposition}\label{prop:univ_product_tensor_bilinear}
  Let $A,B,C$ be three $\ICones$ objects, and let $\eta_{A,B}:A,B\rightarrow A\otimes B$ the canonical bilinear ci-map -- i.e. the canonical morphism $A \rightarrow (B\multimap A\otimes B)$ from Corollary~\ref{prop:univ_product_tensor}). Then:
\begin{align*} & \forall  D,\,\forall \text{ bilinear ci-map } g:A,B\rightarrow D, \,\\ &\exists !  h: A\otimes B \rightarrow D, \text{ s.t. } g(x,y) =  h(\eta_{A,B}(x,y)).
\end{align*}
For $x\in A$, and $b \in B$, we will sometimes note $x\otimes y:=\eta_{A,B}(x,y) \in A\otimes B$. Observe that Proposition~\ref{prop:adjunction_as_univ_morphism} also tells us that for $f:A_1\rightarrow A_2$, , $\forall a_1\in A_1,a_2\in A_2, \, (f\otimes A_2)(a_1\otimes a_2) = f(a_1) \otimes a_2$.
  \end{proposition}
\subsubsection{\texorpdfstring{$\FMeas(X\times_{\MeasC}Y)$ verifies the universal property of the $\Icones$ tensor product.}{Universal property for the embedding of product of measurable spaces.}}
To show that the embedding $C:\SubStoch(\Arity) \rightarrow \Icones$ is strong monoidal and symmetric, we will prove that $ X,W$ objects in $\Arity$  that $C(X \times_{\MeasC} W)$ verifies the universal property  for the tensor product of $CX$ and $CW$ in $\Icones$. To do that, we need to define a morphism $\kappa_{X,W}:CX \rightarrow (CW\multimap CX)$ in $\Icones$, designed with the aim of having $(C(X \times_{\MeasC} W),\kappa_{X,W})$ a universal morphism from $CX$ to the functor $(CW \multimap \_)$. 

\begin{lemmarep}\label{prop:construction_kappa}
  Let $X_1,X_2$ be measurable spaces in $\Arity$.
  The following map is a bilinear ci-map
\begin{align*}
  \kappa_{X_1,X_2}: &\FMeas {(X_1)} \times  \FMeas{(X_2)}  \rightarrow  \FMeas(X_1 \times  X_2) \\
  &(x_1,x_2) \mapsto (x_1 \times  x_2),
\end{align*}
where $x_1\times x_2$ is the \emph{product measure} of $x_1$ and $x_2$, in the sense of measure theory.
\end{lemmarep}
\begin{proof}
  First, it is immediate from the definition of product measure on a product measurable space that $\kappa_{X_1,,X_2}$ is bilinear continuous map -- ~i.e.~separately linear, and separatey Scott-continuous. 
  \begin{itemize}
\item Let's now look at measurability: recall that  for any $d,D \in \Arity$,  $\{\gamma \in \pathes d {\functorClin (D)} \text{ s.t. } \norm{\gamma}\leq 1 \} = \SubStoch(d,D)$. What we need to check--as highlighted in~\cite{ehrhard2025integration} in the section on the cone $(X,Y \multimap Z)$--is that for any $d_1,d_2 \in \Arity$, for any $\gamma_1 \in \pathes{d_1}{\functorClin X}$, $\gamma_2 \in \pathes{d_2}{\functorClin y}$, the map $(r_1,r_2 \in d_1 \times d_2 \mapsto \kappa(\gamma_1(r_1),\gamma_2(r_2)))$ is in $\pathes{d_1\times d_2}{\functorClin(X \otimes Y)}$. Using linearity, we can suppose without loss of generality that the norm of both $\gamma_1$ and $\gamma_2$ is smaller than $1$. It means that we can rewrite the requirement above in $\SubStoch$: we need to check that for every $\gamma_1 \in \SubStoch(d_1,X)$, $\gamma_2 \in \SubStoch(d_2,Y)$, it holds that  $(r_1,r_2 \in d_1 \times d_2 \mapsto \gamma_1(r_1)\times \gamma_2(r_2)) \in \SubStoch(d_1 \otimes d_2, X \otimes Y)$. It is immediate, because this map is simply the tensor product of $\gamma_1$ and $\gamma_2$ in $\SubStoch$. 
\item Let's look at integrability. We start from  $d_1,d_2 \in \Arity$, $\mu_1,\mu_2 \in \FMeas(d_1),\FMeas(d_2)$, and $\gamma_1 \in \pathes{d_1}{\functorClin X}$, $\gamma_2 \in \pathes{d_2}{\functorClin Y}$. We need to show that:
  \begin{equation}\label{goal:kappa_1} \kappa_{X,Y}(\cintc{\functorClin X} \gamma_1\cdot d\mu_1,\cintc{\functorClin Y} \gamma_2\cdot d\mu_2) =\cint^{\functorClin{(X \otimes Y)}}_{r_2 \in d_2} (\cint^{\functorClin{(X \otimes Y)}}_{r_1 \in d_1} \kappa_{X,Y}(\gamma_1(r_1),\gamma_2(r_2)) \mu_1 (dr_1))\mu_2(dr_2),
\end{equation}
     where the order of the two integrals in the right part of the equality doesn't matter, because of the $\Icones$-Fubini theorem shown in~\cite{ehrhard2025integration}. Once again, by linearity of all the constructions implied, it is enough to prove it when $\gamma_1,\gamma_2$ are $1$-bounded. As before,  whenever $d,D \in \Arity$, a $1$-bounded $\gamma \in \pathes d D$ is simply an element of $\SubStoch(d,D)$, and moreover for any $\mu \in \FMeas(D)$ it is the case that $\cint^D\gamma d\mu = \gamma \circ \mu$, where the composition is in $\SubStoch$, and a sub-probability measure on $d$ is identified with a $\SubStoch$ morphism from $\unit = \{\star\}$ to $d$. Rewriting our goal in $\SubStoch$--and using the fact that $\kappa_{X,Y}$ there is simply the tensorial product on $\Stoch$ morphisms--we see that our goal~\eqref{goal:kappa_1} becomes:
  $$(\gamma_1 \circ \mu_1) \otimes (\gamma_2 \circ \mu_2) = (\gamma_1 \otimes \gamma_2)\circ(\id_{d_1} \otimes \mu_2)\circ (\mu_1 \otimes \id_{d_1}),$$ which obviously holds by functoriality of the tensor product in $\SubStoch$.
\end{itemize}
\end{proof}
In the following, we will show that $(C(X\times_{\MeasC}Y),\kappa_{X_1,X_2})$ verifies the universal property presented in Proposition~\ref{prop:univ_product_tensor_bilinear}. For this, we need the auxilliary Lemma~\ref{prop:construction_ci_map_by_integration} below -- proved in~\cite{ehrhard2025integration} as part of Theorem 7.1 there, and for which we give a self-contained proof in the appendix -- that gives us a way of building an $\Icones$ morphism from a measurable path.
 \begin{lemmarep}\label{prop:construction_ci_map_by_integration}
    Let $X \in \Arity$, $C$ an arbitrary object in $\Icones$, and  $\gamma \in \pathes X C$.
    The map
  $(z \in \functorClin{X} \mapsto \cint_{X}^{C} \gamma \cdot dz \in C)$ is  an $\Icones$-morphism.
  \end{lemmarep}
  \begin{proof}
    Let's note $\alpha:= (z \in \functorClin{X} \mapsto \cint_{X}^{C} \gamma \cdot dz \in C)$.
  First,  Lemma 4.6 from~\cite{ehrhard2025integration} tells us that $\alpha$ is 
  linear, continuous and measurable. Let's now look at integrability.
 
  Let $d \in \Arity$, $\delta \in \pathes d {\functorClin{X}}$, and $\nu$ a bounded measure on $d$. Our goal is to show:
$ \alpha(\cint_{d}^{\functorClin{X}} \delta. d\nu) = \cint_d^{C} \alpha\circ \delta. d\nu$. We can rewrite this goal as:
\begin{equation}
\cint^{C}_{r \in X} \gamma(r) \cdot \left( \cint_{s \in d}^{\functorClin{(X)}} \delta(s). \nu(ds))\right)(dr)  = \cint_{s \in d}^{ C}\left( \cint^{C}_{r \in X} \gamma(r).\delta(s)(dr)\right). \nu(ds)
  \end{equation}

  Recall that $ \cint_{s \in d}^{\functorClin{X}} \delta(s). \nu(ds)$ is simply $\delta_\star(\nu)$, i.e.~the push-forward of the kernel $\delta$ applied to $\nu$, thus we can again rewrite our goal as:
\begin{equation}
\cint^{C}_{r \in X} \gamma(r) \cdot \delta_\star(\nu)(dr)  = \cint_{s \in d}^{ C}\left( \cint^{C}_{r \in X} \gamma(r).\delta(s)(dr)\right). \nu(ds)
  \end{equation}
  
  Let's recall the \emph{change of variables} lemma (5.7) from~\cite{ehrhard2025integration}: if $e,e' \in \Arity$, $e' \xrightarrow f e$ a measurable function, $C$ an object in $\Icones$, and $\lambda \in \pathes{e}{D}$, then for every $\mu \in \Meas(e')$: $\cint^{C}_{r' \in e'}\lambda\circ f(r') .\mu(dr') = \cint^C_{r \in e} \lambda(r). (f_{\star}(\mu))(dr)$. To conclude our proof, what we need is to check the validity of a stronger version of this change-of-variable idea, that could be applied when $e' \xrightarrow f e$ is in $\SubStoch$ instead of being merely in $\Meas$. More precisely, we need to show that if $e,e' \in \Arity$, $e' \xrightarrow f e$ a kernel in $\SubStoch$, $C$ an object in $\Icones$, and $\lambda \in \pathes{e}{C}$, then for every $\mu \in \Meas{(e')}$: $\cint^{C}_{r'\in e'}(\cint^{C}_{r \in e}\lambda(r).f(r')(dr)) .\mu(dr') = \cint^C_{r \in e} \lambda(r). (f_{\star}(\mu))(dr)$.
 This stronger change-of-variable lemma can be proved with the same tools as the one from~\cite{ehrhard2025integration}, plus  associativity of the composition in $\SubStoch$.
\end{proof}

\begin{proposition}\label{proposition:universal_morphism}
For any objects $X,W$ in $\SubStoch(\Arity)$, $(C(X\times_{\MeasC} W),\eta_{X,W})$ is an universal morphism from $C X$ to the functor $(C W \multimap \_)$. 
\end{proposition}
\begin{proof}
We need to show that the universal property stated in Proposition~\ref{prop:univ_product_tensor} holds for $(C(X \times_{\MeasC} W),\eta_{X,W})$.
Let us reformulate it -- using again the notion of bilinear maps:
\begin{align*} & \forall \Icones\text{-object} D,\,\forall \text{ bilinear ci-map } g:[CX,CW]\rightarrow D, \,\\ &\exists ! \, \ICones\text{-morphism } h: C(X\times_{\MeasC} W) \rightarrow D, \text{ s.t. } g(x,y) =  h(x\times y).
\end{align*}
In particular, it would mean that for any \emph{Dirac measures} $\dirac a \in \FMeas{(X)},\,\dirac b \in \FMeas{(W)}$, $g(\dirac a,\dirac b) = h(\dirac a \times \dirac b) = h(\dirac{(a,b)})$.
\begin{itemize}
\item The unicity of such a $h$  is immediate, because any $\ICones$-morphism $\FMeas(Z) \rightarrow D$  is entirely characterised by its values on the Dirac measure $\delta_z$, for $z \in Z$--since $(z \in Z \mapsto \dirac z \in \FMeas(Z))$ is always a measurable path.
\item Let us now show the existence of $h$. Lemma~\ref{prop:construction_ci_map_by_integration} allows us to define an $\Icones$ morphism $h: C(X\times_{\MeasC} W) \rightarrow D$ as:
  $h := (z \in C(X\times_{\Meas} W) \mapsto \cint_{X}^{C} \gamma \cdot dz \in Y)$, for $\gamma:(x,w) \mapsto g(\dirac x,\dirac w) \in \pathes{(X \times_{\MeasC} W,Y)}$. Observe that this $\gamma$ is indeed a valid path, because it is the composition of the bilinear ci-map $g$ with the two measurable path $(x \in X\mapsto \dirac x \in C X)$ and $(w \in W \mapsto \dirac w \in C W)$. We can check that indeed $\forall x \in \FMeas X,y \in \FMeas Y, \, g(x,y) = h(x\times y)$, using~\cite{ehrhard2025integration} Fubini's theorem for integrable cones.
  \end{itemize}
\end{proof}
\subsubsection{\texorpdfstring{Monoidality of $\SubStoch(\Arity) \hookrightarrow \Icones$.}{Monoidality of the embeding from sub-stochastich kernels to integrable cones.}}
We proved in Proposition~\ref{proposition:universal_morphism} that $C(X\times_{\MeasC} Y)$ verifies the universal property for the tensor product. As an immediate consequence it is isomorphic to the chosen tensor product in $\Icones$: $CX \otimes CY$, and this isomorphism is natural. To be able to conclude, however, that $\SubStoch(\Arity) \hookrightarrow \Icones$ is a strong symmetric monoidal functor, we need to show that it also preserves the coherences and symmetries.
\begin{theorem}\label{th:substoch_to_icones_is_strong_monoidal2}
$C:\SubStoch(\Arity) \hookrightarrow \Icones$ is a strong monoidal functor.
\end{theorem}
\begin{proof}
  Combining Proposition~\ref{proposition:universal_morphism} and Proposition~\ref{prop:univ_product_tensor}, we obtain immediately  for any $X,Y \in \Arity$ an isomorphism  $$i_{X,Y}: \FMeas X\otimes_{\Icones}\FMeas Y \rightarrow \FMeas(X\times_{\MeasC} Y)$$
  such that moreover
\begin{equation}\label{eq_preserv_symm_1}
  \forall x \in CX,\forall y \in CY,\,i_{X,Y}(x\otimes y)=\kappa_{X,Y}(x,y)=x\times y
  \end{equation}
  \begin{itemize}
  \item Preservation of symmetries: let $s_{CX,CY}:CX \otimes CY \rightarrow CY \otimes CX$ the canonical symmetry morphism in $\ICones$, and $\sigma_{X,Y}$ the canonical symmetry morphism in $\SubStoch(\Arity)$. Recall that $\sigma_{X,Y}$ is given by the cartesian structure on $\MeasC$, and the strong structure of the Giry monad: concretely  $\forall x\in X,y\in Y$, $\sigma_{X,Y}(x,y) = \dirac{(y,x)}$. By definition of the symmetric monoidal structure on $\ICones$ -- using the universal property of the tensor -- it holds that $s_{X,Y}$ is the unique morphism  such that:
    \begin{equation}\label{eq_preserv_symm_2}
      \forall x \in CX,y\in CY, \, s_{X,Y}(x\otimes y) = y\otimes x \end{equation}
      We want to show that
    $ s_{CX,CY} = i_{X,Y}^{-1}\circ C\sigma_{X,Y} \circ i_{X,Y} : CX\otimes CY \rightarrow CY\otimes CX.$
    Using the universal property of the tensor product $CY\otimes CX$, it is equivalent to show that $\forall x\in CX,y\in CY$:
    $s_{CX,CY}(x\otimes y) = i_{X,Y}^{-1}\circ C\sigma_{X,Y} \circ i_{X,Y}(x\otimes y) \, \in  CY\otimes CX. $
    Applying~\eqref{eq_preserv_symm_2}, we can rewrite this goal as:
    $ \forall x,y,\, y\otimes x = i_{X,Y}^{-1}\circ C\sigma_{X,Y} (\kappa_{X,Y}(x,y) \in  CY\otimes CX.$
    We can easily check that $C\sigma_{X,Y}(\kappa_{X,Y}(x,y)) = y\times x = \kappa_{Y,X}(y,x)$. Thus we can rewrite our goal as:
    $$\forall x,y,\,y \otimes x = i_{X,Y}^{-1}(y\times x),$$
    and we can see immediately from~\eqref{eq_preserv_symm_1} that this equation holds.
  \item Preservation of unit: it has been shown in~\cite{ehrhard2025integration} that a unit of $\otimes$ in $\ICones$ is the cone $\unit = \FMeas(\{\star\})$, where $\{\star\}$ is a singleton set. $\{\star\}$ is also a unit of the monoidal product in $\SubStoch(\Arity)$ . By the universal property of $\otimes_{\ICones}$, for any $A \in \ob{\ICones}$, $\rho_A:A\otimes 1 \rightarrow X$ is the unique morphisms such that for every $a \in A$, $\rho_X(a \otimes \dirac{\star}) = a$. Let's note now for $X\in \Arity$, $\rho^{\dagger}:X\times_{\MeasC}\{\star\} \rightarrow X$ for the morphism given by the monoidal structure on $\SubStoch(\Arity)$: concretely $\forall x\in X,\rho^{\dagger}(x,\star) = \dirac x$. We need to show that $\rho_{CX} = C(\rho^{\dagger}_{X}) \circ i_{CX,1}$. Let's check this: for any $x\in \FMeas (X)$, we have that $C(\rho^{\dagger}_X)\circ i_{CX,1}(x\otimes \delta_{\star}) = C(\rho^{\dagger}_X)(x\times \dirac{\star}) = x$, which concludes the proof. 
    \item Preservation of associator. Using the universal property of $\otimes_{\ICones}$ given by the adjunction $(\_ \otimes B_1)\otimes B_2) \dashv B_1 \multimap B_2 \multimap \_)$, we see that for any $\ICones$ objects $A,B,D$, the associator $\alpha_{A,B,D}: (A\otimes B)\otimes D \rightarrow A \otimes (B\otimes D)$ is the unique morphism such that for every $a \in A,b \in B,d\in D$, $\alpha_{A,B,D}((a\otimes b)\otimes d) = (a\otimes (b\otimes d))$. For $X,Y,Z\in \Arity$, let's write $\alpha^{\dagger}_{X,Y,Z}$ for the associator $(X \otimes Y)\otimes Z \rightarrow X \otimes (Y\otimes Z)$ in $\SubStoch(\Arity)$: concretely $\forall x\in X,y\in Y,z\in Z$, $\alpha^{\dagger}((x,y),z) = \dirac{(x,(y,z))}$. We need to check that $\alpha_{CX,CY,CZ} =( CX\otimes i_{Y,Z}) \circ i_{X,Y\times Z}\circ C(\alpha^{\dagger}_{X,Y,Z}) \circ i_{X\times Y,Z}\circ (i_{X,Y}\otimes CZ)$. So let's take $x\in \FMeas(X), y \in \FMeas(Y), z\in \FMeas(Z)$. Using Proposition~\ref{prop:univ_product_tensor_bilinear}, we see that: $ i_{X\times Y,Z}\circ (i_{X,Y}\otimes CZ) ((x\otimes y)\otimes z) =  i_{X\times Y,Z}((i_{X,Y}(x\otimes y))\otimes z)  = \circ i_{X\times Y,Z}((i_{X,Y}(x\times y))\otimes z) =  \circ i_{X\times Y,Z}(x\times y)\otimes z) = (x\times y) \times z.$ Since$ \alpha^{\dagger}_{X,Y,Z}((x \times y)\times z) = x \times (y \times z)$, we can conclude from here.
    \end{itemize}
  
  \end{proof}
\end{toappendix}

\subsection{\texorpdfstring{Preservation of limits from $\Stoch(\Arity)$ to $\Icones$}{The embedding from stochastic kernels to integrable cones preserves connected limits.}}

It is rather easy to prove that the forgetful functor $\Stoch(\Arity) \hookrightarrow \SubStoch(\Arity)$ is a parametric right adjoint, thus preserves all small connected limits -- we prove this in Proposition~\ref{prop:preserv_limits_stoch_substoch} in appendix.
\begin{toappendix}
\begin{propositionrep}\label{prop:preserv_limits_stoch_substoch}
$\Stoch(\Arity) \hookrightarrow \SubStoch(\Arity)$ is a parametric right adjoint, thus preserves all small connected limits.
 \end{propositionrep}

 \begin{proof}
  Let us first recall a caracterisation of parametric right adjoint, due to~\cite{weber2004generic}, that uses generic factorizations.
     \begin{definition}
       Let $T:\Catone \rightarrow \Cattwo$ a functor. A morphism $f:B \rightarrow TA$ is \emph{$T$-generic} when for any commutative square
       $$\begin{tikzcd}[ampersand replacement=\&]
         B \arrow[r," \alpha"]\arrow[d, "f "] \& TX \arrow[d,"T\gamma " ]\\
         TA \arrow[r, "T\beta "] \& TZ
         \end{tikzcd}$$
       there exists a unique lift of the form $T\delta:TA \rightarrow TX$ such that $\gamma \circ \delta = \beta$. 
       \end{definition}
     \begin{definition}
       Let $T:\Catone \rightarrow \Cattwo$ a functor, such that $\Catone$ has a terminal object $\unit$. Let  $A \xrightarrow f {TB}$ be a morphism in $\Cattwo$. A  \emph{generic factorization} of $f$ is a factorisation:
$B \xrightarrow g TD \xrightarrow {Th}{TA}, $
       such that $g$ is $T$-generic.
     \end{definition}
     \begin{proposition}[From~\cite{weber2004generic}]
       Let $\Catone$ a category with a terminal object.
A functor $T:\Catone \rightarrow \Cattwo$ is a parametric right adjoint if and only if any map $A \xrightarrow f TB$ admits a $T$-generic factorization.
\end{proposition}
\begin{lemma}
  Let $T$ be the forgetful functor $\Stoch(\Arity) \rightarrow \SubStoch(\Arity)$.
Any map $A \xrightarrow f {TB}$ in $\SubStoch(\Arity)$ admits a $T$-generic factorization.
\end{lemma}
\begin{proof}
  Let $X,Y$ two objects in $\Arity$, and $X \xrightarrow f TY$ a sub-stochastic kernel (recall that $T$ is the identity on objects). We need to show that $f$ admits a $T$-generic factorization. It is immediate that $f= X \xrightarrow f TY \xrightarrow {TY} TY$, thus it is enough to show that the identity on $TY$ is $T$-generic. Let us do this: we consider any commutative square:
  $$\begin{tikzcd}[ampersand replacement=\&]
         TY \arrow[r," \alpha"]\arrow[d, "TY"] \& TX \arrow[d,"T\gamma " ]\\
         TY \arrow[r, "T\beta "] \& TZ
       \end{tikzcd}$$
       And we need to find a lift $T\delta:TY \rightarrow TX$, such that moreover $\gamma \circ \delta = \beta$. The key point here is that we can conclude from: $T\gamma \circ \alpha = TY \circ T\beta = T\beta$, that $\alpha$ is itself of the form $T\alpha'$ -- indeed the equality above tells us that when we post-compose the sub-stochastic kernel $\alpha$ with a stochastic kernel, we obtain again a stochastic kernel, and the only way that can happen is when $\alpha$ itself is a stochastic kernel. From there, we can rewrite the commutative diagram as:
         $$\begin{tikzcd}[ampersand replacement=\&]
         TY \arrow[r," T\alpha'"]\arrow[d, "TY"] \& TX \arrow[d,"T\gamma " ]\\
         TY \arrow[r, "T\beta "] \& TZ
       \end{tikzcd}$$
       We see from there that $\alpha=T\alpha'$ lifts the square, and moreover since $T$ is faithful it holds that $\gamma \circ \alpha' = \beta$.
  \end{proof}
   
         \end{proof}
         
\end{toappendix}
       Next, we show a convenient feature of the category $\Cones$ with respect to limits preservation: 

         \begin{lemmarep}\label{lemma:generic_on_icones_limits}
          
           Let $D:I\rightarrow \Cones$ a diagram. 
           $(L\xrightarrow {g_i} D_i)$ is a limit cone of  $D$ if and only if any cone on $D$ \emph{with source $\unit$} can be factorised uniquely through $L$.
           As a consequence, for any category $\Catone$, if $F:\Catone \rightarrow \Cones$ is a functor such that \begin{enumerate}\item there exists $C_{\unit}$ in $\Catone$ such that $FC_\unit=\unit$  and \item the $F$-action on morphisms is bijective $\Catone(C_\unit,A) \cong \Cones(\unit,FA)$ for every object $A$,  \end{enumerate} then $F$ preserves all small limits.
          
\end{lemmarep}
\begin{proof}
  
Let $D:I \rightarrow \Cones$ a small diagram. We suppose that this diagram has a limit cone $(L,g_i:(L \rightarrow DX_i))_{X_i \in \ob I}$ in $\Cones$. Let $A$ an object in $\Cones$.
For every element $x \in A$, we can obtain from $C$ a cone $C_X(x):(\unit,f_X(x):\unit \rightarrow DX)$ in $\Icones$, defined as $f_X(x): \lambda \in \unit = \RR_{\geq 0} \mapsto \lambda \cdot f(x) \in X)$. We apply now our hypothesis on $L$, and we obtain the existence of a \emph{unique} element--that we'll note $f^\dagger(a) :\unit \rightarrow L$, such that $g_i(f^\dagger(a)) = f_i(a)$. All the problem now consists in showing that $a \in A \mapsto f^\dagger(a) \in L$ is indeed an $\Cones$-morphism, i.e. that $f^\dagger$ is linear, and Scott-continuous. 
    \begin{enumerate}
    \item Let $a,b \in A, \lambda \in \RR$. Let us first prove that $g_i(f^\dagger(a) + \lambda f^\dagger(b)) = f_i(a + \lambda b),$ for every $X_i \in \ob{I}$: since we know that both the $g_i$ and the $f_i$ are linear (since they are $\Cones$ morphism), this statement is equivalent to say that  $g_i(f^\dagger(a)) + \lambda g_i(f^\dagger(b)) = f_i(a) + \lambda f_i(b),$ and we can obtain the latter immediately from the definition of $f^\dagger(a)$ and $f^\dagger(b)$. From there, we can conclude by unicity of $f^\dagger(a + \lambda b)$ that $f^\dagger(a+ \lambda b) = f^\dagger(a) + \lambda .f^\dagger(b)$.
      \item The reasoning is similar for Scott-continuity: let us consider en increasing chain $x_n$ in the unit ball of $C$. We want to show that $f^\dagger(\sup\{x_n\}) = \sup \{f^\dagger(x_n)\}$. Again by unicity of the way $f^\dagger(\sup \{x_n\})$ is defined, it is enough to show that for every $i$, $g_i(\sup \{f^\dagger{x_n})) = f_i(\sup \{x_n\})$. But again we obtain that easily since the $f_i,g_i$ are $\Cones$ morphism thus Scott-continuous.
      \end{enumerate}

    \end{proof}

 As a corollary of Lemma~\ref{lemma:generic_on_icones_limits}, the functor $\SubStoch(\Arity) \rightarrow \Cones$ preserve small limits. We now look at how lifting this result to $\ICones$.
\begin{toappendix}
 \begin{propositionrep}\label{prop:preserv_limits_substoch_cones}
$\SubStoch(\Arity) \rightarrow \Cones$ preserve all small limits.         

\end{propositionrep}
\begin{proof}
  
  Let  $D:I \rightarrow \SubStoch(\Arity)$ be a small diagram. We suppose that this diagram has a limit cone $C:(L,g_i:(L \rightarrow DX_i))_{X_i \in \ob I}$ in $\SubStoch(\Arity)$.
  By Lemma~\ref{lemma:generic_on_icones_limits},  it is enough for $C^{\Cones}$ to be a limit cone on $D^{\Cones}$, that we are able to factorise any cone $E$ on $D^{\Cones}$ with source $\unit_{\Cones}$  through $C^{\Cones}$. The embedding $ {\SubStoch(\Arity) \hookrightarrow \Cones}$ is surjective on morphisms in $(\unit_\SubStoch,Y)$, for every $Y \in \Arity$, and $\unit_{\Cones} = (\unit_{\SubStoch})^{\Cones}$, thus there exists a cone $E'$ on $D$ such that $E = (E')^{\Cones}$. From there, we can conclude since by hypothesis $C$ is a limit cone in $\Stoch$. 
   
\end{proof}

  We can also do now --using Lemma~\ref{lemma:generic_on_icones_limits} -- the proof of the fact that $\SubStoch^{\leq \omega} \rightarrow \pcoh$ preserve small limits.
  \begin{lemma}\label{lemma:limits_preservation_substoch_pcoh}
$\SubStoch^{\leq \omega} \rightarrow \pcoh$ preserve small limits.
\end{lemma}
\begin{proof}
First, we can see that the requirements of Lemma~\ref{lemma:generic_on_icones_limits} hold for $F:\SubStoch^{\leq \omega} \rightarrow \SubStoch(\Arity) \rightarrow \Cones$, thus $F$ preserve limits. Moreover, recall that all our embeddings are consistent with each others, and in particular: $(\pcoh \rightarrow \Cones)\circ  (\SubStoch^{\leq \omega} \rightarrow \pcoh) = F$. To conclude that  $(\SubStoch^{\leq \omega} \rightarrow \pcoh)$, it is sufficient to observe that $(\pcoh \rightarrow \Cones) $ is fully faithful. 
  \end{proof}
  \end{toappendix}

For this, we first prove that the forgetful functor $\ICones \rightarrow \Cones$  \emph{lifts} all small limits, i.e. that any diagram $D$ that has a $\Cones$-limit $L$ when forgotten into $\Cones$ has also a limit $L'$ in $\ICones$, and that moreover we have a concrete description of its measurability structure.  Let us note  $\overline A^{\Cones},\overline{A}^{\MCones}$ for the image of $A$ in $\ICones$ under the relevant forgetful functors. 

     \begin{toappendix}
        \begin{lemma}\label{lemma:auxilliary_on_cones_limits}
          Let $D:I\rightarrow \Cones$ be a diagram, and $C =(L\xrightarrow{g_i} D_i)_{i \in I}$ a limit cone. Then:
          \begin{itemize}
            \item $\forall x,y \in L$ such that $\forall i$, $g_i(x) = g_i(y)$, then $x=y$;
            \item for every $x\in L$, $\norm x = \sup_{i \in I} g_i(x)$.
              \end{itemize}
\end{lemma}
\begin{proof}
 The first point can be seen by considering the cone $\unit \rightarrow D_i$ given by the $(g_i(x))_{i \in I}$ -- i.e. the one that to $\lambda \in \unit \mapsto \lambda\cdot g_i(x)$.
  Let $x\in X$.
Because all $\Cones$-morphisms are $1$-bounded, it is immediate that $\norm x \geq \sup_{i \in I} g_i(x)$. Now, suppose by contradiction that $\norm x \leq \alpha\cdot \sup_{i \in I} g_i(x)$, with $\alpha < 1$. We consider the cone $\alpha.\unit$, i.e. the cones that are the same elements as $\unit = \RR_{\geq 0}$, but $\norm{y}_{\alpha.\unit} = \alpha\cdot \norm y$. The key property of $\alpha.\unit$ is that it contains strictly \emph{more} elements in its unique ball, thus there is a morphism $\unit \rightarrow \alpha. \unit$, but not in the reverse direction. We can form a cone $E:(\lambda\in \alpha.\unit \mapsto \lambda\cdot g_i(x) \in D_i)$: it works because of our hypothesis $\norm x \leq \alpha\cdot \sup_{i \in I} g_i(x)$. So there exists a unique morphism $\eta:\alpha\cdot \unit \rightarrow L$ that factorises $E$, and since each $x\in L$ is uniquely characterised by its value under the $g_i$, we need to have $\eta(1) = x$. But $\eta$ can't be a $\Cones$ morphism, because it isn't $1$ bounded. 
  \end{proof}
        \end{toappendix}
  \begin{propositionrep}\label{prop:limits_lifting}
    The forgetful functors $\Icones \rightarrow  \Cones$ \emph{lifts} all small limits. 
      For $D$ is a small diagram in $\ICones$, and $C = (L, g_i:L\rightarrow \overline{D_i}^{\Cones})$ is a limit cone in $\Cones$ over $\overline{D}^{\Cones}$, we note $C^{\triangleleft}$ for the lifting of $C$ along $\ICones \rightarrow \Cones$. Moreover, it holds that:
      \begin{enumerate}
      \item\label{item:lifting_measurable_paths} a map $\gamma:R\rightarrow L$ is a measurable path on $L^{\triangleleft}$ if and only if every map $g_i\circ \gamma:R\rightarrow D_i$ is a measurable path;
        \item\label{item:lifting_integration} for such a measurable path $\gamma:R\rightarrow L^{\triangleleft}$, and $\mu \in \FMeas(R)$, $\cint_R^L\gamma\cdot d\mu$ can be characterised as the \emph{unique} element such that for every $i$, $g_i(\cint_R^L \gamma\cdot d\mu ) = \cint_R^{D_i} g_i\circ \gamma\cdot d\mu$.
        \end{enumerate}
  \end{propositionrep}
  \begin{proof}
    Let $D:I\rightarrow \ICones$ be a diagram. For any limit cone $C =(L\xrightarrow{g_i} \overline{ D_i}^{\Cones})_{i \in I}$ on $\overline{D}^{\Cones}$ in $\Cones$, we need to:
    \begin{enumerate}%[a)]
    \item find  a way to enrich $L$ into an  $L^{\triangleleft}$ with measurability tests such that moreover all the $L^{\triangleleft} \xrightarrow{g_i} {D_i}$ preserve measurable paths.
    \item prove that $L^{\triangleleft}$ is an integrable cone, and that moreover the $L^{\triangleleft} \xrightarrow{g_i} D_i$ preserve integrals;
    \item prove that the resulting cone $(L^{\triangleleft} \xrightarrow{g_i} D_i)$ is a limit on $D$ in $\ICones$.
      \end{enumerate}
  
      We do the three steps one by one:
      \begin{enumerate}%[a)]
   \item We take as measurability tests on $L$:
      $$\mtest R^{L} = \{((r,x) \in R\times L \mapsto t_i(r, g_i(x))):L\rightarrow \unit \mid t_i \in \mtest R^{D_i} \}.$$
      $\mtest R^{L}$ is indeed a valid measurability structure on $L$. It is immediate measurability and composability -- in the sense of Definition~\ref{def:measurability_structure} holds for $\mtest R^{}$. That also separation and respect of the norm hold is a consequence of Lemma~\ref{lemma:auxilliary_on_cones_limits} in the appendix. 
 
 Once we know that this defines a valid object in $\MCones$, it is  immediate that~\eqref{item:lifting_measurable_paths} holds.
    \item Let now be $R\in \Arity$, $\distrone \in \FMeas(R), \,\gamma:R\rightarrow L^{\triangleleft}$. 
      We form a cone on $D$ as: $E:=( \lambda\cdot \star\in \unit \xrightarrow{\lambda\cdot \int{g_i \circ \gamma}} D_i)$ -- to check that it is indeed a cone on $D$, we need to use the fact that the morphisms that forms $D$ preserve integrals. So we obtain a unique $\delta:\unit \rightarrow L$ that factorises $E^{\Cones}$ through $C$, and we defines $\cint_{R}^{L^{\triangleleft}} \gamma.d\mu =\delta(\star) $. We can check that it verifies the requirement for being the integral of $\gamma$ over $\mu$. We can check that moreover~\eqref{item:lifting_integration} holds by observing that when we fix $\mu$, the family $(\lambda\cdot \star \in \unit \mapsto \lambda\cdot g_i\circ\gamma(\mu) \in D_i)$ are a cone $(\unit \rightarrow D_i)$ on $\Cones$ . To conclude, we can also check immediately that the $g_i:L^{\triangleleft} \rightarrow D_i$ preserve integrals.
    \item To check that $(L^{\triangleleft} \xrightarrow{g_i} D_i$ is indeed a limit cone on $D$, we take another cone $(Z\xrightarrow {h_i} D_i$. Using the fact that $C$ is a limit cone on $D$, we obtain a unique $\Cones$-morphism $h:Z \rightarrow L^{\triangleleft}$. We can check easily that $h$ is actually an $\ICones$-morphism by using~\eqref{item:lifting_measurable_paths} and~\eqref{item:lifting_integration}.
      \end{enumerate}

    \end{proof}

    Proposition~\ref{prop:limits_lifting} tells us that starting from a limit cone $L$ of a diagram $D$ in $\SubStoch(\Arity)$, we can thus build \emph{two} cones on $D^{\ICones}$: the first one is the cone $C^{\ICones}$ given by the functor $\SubStoch(\Arity) \rightarrow \ICones$. The second one is the cone $(C^{\Cones})^{\triangleleft}$ obtained   using that   $C^{\Cones}$ is a limit on $D^{\Cones}$ -- since $\SubStoch(\Arity) \rightarrow \Cones$ preserve limits -- and thus by Proposition~\ref{prop:limits_lifting} we can lift it  into a limit cone  $(C^{\Cones})^{\triangleleft}$ on $D^{\ICones}$.  Using the concrete description of the measurability structure on $(C^{\Cones})^{\triangleleft}$ we proved as part of Proposition~\ref{prop:limits_lifting}, we prove -- in Proposition~\ref{prop:iso_to_lifting} in appendix -- that these two cones are actually isomorphic, and consequently that $C^{\ICones}$ is a limit cone on $D^{\ICones}$.
    \begin{toappendix} 
   \begin{propositionrep}\label{prop:iso_to_lifting}
Let $D$ a diagram in $\SubStoch(\Arity)$, and $C$ a limit cone for $D$ there. Then the image of $C$ by $\SubStoch(\Arity) \rightarrow \Icones$ is $\ICones$-isomorphic to the lifting $(C^\Cones)^{\triangleleft}$ from the limit cone $C^{Cones}$ on $D^{\Cones}$ along the forgetful functor $\Cones \rightarrow \ICones$.
\end{propositionrep}
\begin{proofsketch}
The main ingredient of the proof is the characterisation of measurability paths (and integrals) on  $(L^{\Cones})^{\triangleleft}$ that we proved as part of Proposition~\ref{prop:limits_lifting}. 
  \end{proofsketch}
\begin{proof}
  Let $C = (L,(L\xrightarrow{g_i} D_i)_{i\in I})$ a limit cone over $D$ in $\SubStoch(\Arity)$.
Recall from the definition of $(\_)^{\triangleleft}$ in the statement of Proposition~\ref{prop:limits_lifting} that $L^{\ICones}$ and $(L^{\Cones})^{\triangleleft}$ have the same underlying cone, which is $L^{\Cones}$.
\begin{itemize}
  \item First, since $(C^{\Cones})^{\triangleleft}$ is a limit cone on $D^{\ICones}$, there exists a unique cone morphism $C^{\ICones}\rightarrow (C^{\Cones})^{\triangleleft}$, given by a $f: L^{\ICones} \rightarrow (L^{\Cones})^{\triangleleft}$. Moreover, $f$ seen as a $\Cones$ map  is $\idm{\Cones}{L}$, because $C^{\Cones}$ is a limit cone on $D^{\Cones}$ by Proposition~\ref{prop:preserv_limits_substoch_cones}, and so there exists only one cone morphism $C^{\Cones}\rightarrow C^{\Cones}$, which is $\idm{\Cones}{L^{\Cones}}$ . 
    \item\label{item:new_proof_deux} We will now show that $\idm{\Cones}{L^{\Cones}}$ is in $\ICones((L^\Cones)^{\triangleleft} \rightarrow L^{\ICones}$, i.e. it sends  measurable paths on $(L^{\Cones})^{\triangleleft}$  to measurable paths on $L^{\ICones}$, and commutes with integrals in this direction. Observe that it is enough to prove this for \emph{$1$-bounded paths}, and the result for general paths follow by linearity.
     Let $R\in \Arity$, and $\gamma:R\rightarrow (L^{\Cones})^{\triangleleft}$ a measurable $1$-bounded path. We want to show that $\gamma:R \rightarrow L^{\Cones}$ is again a measurable path onto the cone $(L^{\ICones})$, i.e.-- by Lemma~\ref{lemma:meas_paths_are_kernels2}--that $\gamma \in \SubStoch(\Arity)(R,L)$. By Proposition~\ref{prop:limits_lifting} and Lemma~\ref{lemma:meas_paths_are_kernels2}, we know that for every $i$, $(g_i)_{\star}\circ_{\Set} \gamma \in \SubStoch(\Arity)(R,D_i)$. From there, we can see that $R,(g_i)_{\star}\circ_{\Set} \gamma : R\rightarrow D_i)$ is a cone on the diagram $D$ in $\SubStoch(\Arity)$. Here, we are going to use the fact that by hypothesis, $C$ is a \emph{limit cone} on $D$ in $\SubStoch(\Arity)$: as a consequence there exists a mediating $\delta \in \SubStoch(\Arity)(R,L)$ such that:
      \begin{equation}\label{eq:new_approach_1}
        \forall i, (g_i \circ_{\SubStoch(\Arity)} \delta) = (g_i)_\star\circ_{\Set} \gamma.
        \end{equation}
        From the equalities~\eqref{eq:new_approach_1}, we are going to show that $\delta = \gamma$, as $\Set$ maps, and from there we will immediately conclude that $\gamma \in \SubStoch(\Arity)(R,L)$. Let $r\in R$. We can build a cone on $D$ as:
        $F:=(\{\star\}, (g_i\circ_{\Set} \gamma)(\delta_\star):\star \rightarrow D_i)_{i \in I}$ -- recall that a morphism $\unit=\{\star\} \rightarrow X$ in $\SubStoch(\Arity)$ is simply a sub-probability measure. By the universal property of the cone $C$, there exists a unique mediating morphism $\{\star\}\rightarrow L$ that factorises this cone $F$ through $C$. But we can check that both $\gamma(r)$ and $\delta(r)$ are sub-probability measure that --seen as morphisms $\{\star\}\rightarrow L$ factorises $F$ through $C$, thus $\gamma(r) = \delta(r)$ .
      
      \end{itemize}

  \end{proof}

  \end{toappendix}
     \begin{theorem}\label{th:substoch_to_icones_preserve_limits}
$\SubStoch(\Arity) \rightarrow \ICones$ preserves all small limits.
\end{theorem}
By combining Theorem~\ref{th:substoch_to_icones_preserve_limits} with preservation of small connected limits by the embedding $\Stoch(\Arity) \rightarrow \SubStoch(\Arity)$, we can conclude that $\Stoch(\Arity) \rightarrow \ICones$ preserves all small connected limits.

\subsection{De Finetti Theorems}

Combining Theorem~\ref{th:substoch_to_icones_is_strong_monoidal} and Theorem~\ref{th:substoch_to_icones_preserve_limits}, we see that the delete chains, the draw-and-delete chains, and their respective limits are preserved from $\Stoch(\Arity)$ to $\Icones$. From there, we obtain a central technical contribution of our work, which is that the two categorical formulations of De Finetti's theorem hold in the category of integrable cones (when starting from  $X$ in the $\Arity$ category)\footnote{There is no hope to have this result for all objects in $\Icones$, since that would imply De Finetti theorems in $\Stoch$ for \emph{any} measurable space.}.

\begin{theorem}[De Finetti theorems in $\Icones$]\label{th:de_finetti_in_icones}
  Let $X$ be a standard Borel space in the category $\Arity$. In the category $\Icones$:
  \begin{enumerate}
  \item\label{item1:df_in_icones} $(X^\omega)^{\ICones}$ is the limit of the delete chain:
    $$1 \xleftarrow{{\dd_0}^{\ICones}}  X^{\ICones} \xleftarrow{{\dd_1}^{\ICones}}  X^{\ICones} \otimes  X^{\ICones} \ldots,$$ thus we'll note $({X^{\ICones}})^\omega$ for ${(X^{\omega})^{\ICones}}$.
  \item $(\MnCat{\Stoch} n X)^{\ICones}$ is the equaliser of  symmetries on $(X^{\ICones})^{\otimes n}$.
  \item $(\Giry X)^{\ICones}$ is the equaliser of all finitely supported symmetries on $(X^\omega)^{\ICones}$.
  \item ${\Giry X}^{\ICones}$ is the limit of the draw and delete chain built from $(X^\ICones,(\_)^{\ICones}:X^{\ICones} \rightarrow \unit)$.
    \end{enumerate}
  \end{theorem}

 \section{Characterising the total elements in $!X^{\pcoh}$}\label{sect:total_elements}

  In this last part, we combine our results from Sections~\ref{sect:connection_chains} and~\ref{sec:df_in_icones} to investigate what are the elements  of PCS $!X^{\pcoh}$, for any countable discrete measurable space $X$.  It is non trivial to describe those: indeed $!X^{\pcoh}$ can be defined either as a categorical limit with the layered construction of the free exponential, or with the original definition of $!$ in~\cite{danos2011probabilistic} as the bi-orthogonality closure of the set of all the \emph{promotions} of elements in $X^{\pcoh}$. None of these descriptions  gives us a direct and exhaustive description of $!X^{\pcoh}$ elements .
\begin{remark}There is also a characterisation due to Ehrhard and Geoffroy~\cite{ehrhard2025integration} of the bi-orthogonality closure: it is closure by finite convex sum, and by downward closure and $\omega$-continuity for the pointwise order. But it is nonetheless difficult to understand which kind of elements can be obtained by combining the downward closure and the Scott continuity.
\end{remark}
\subsection{Measures on $\Giry X$ live in  $!X^{\pcoh}$ as mixture of promotions}

In this section, we  use our De Finetti theorems in $\ICones$ -- Theorem~\ref{th:de_finetti_in_icones} -- to buld a monomorphism $\iota:\Giry X^{\ICones} \rightarrow (!X^\pcoh)^{\ICones}$, for every countable discrete measurable space $X$. Moreover, we will show that $\iota$ sends all probability measures in $\Giry(\Giry X)$ to \emph{continuous mixtures} of promotions in $\clique{!X^{\pcoh}}$ --  recall from Definition~\ref{prop:limit_morphisms} that for any PCS $A$, the promotions in $\clique {! A}$ are the $x^!$, for $x \in \clique A$, that generate $\clique{!A}$ by biorthogonality.
\subsubsection{The limit of the approximations $\DD^!$ chains are preserved by embedding into $ \ICones$.}
We have seen that for any countable discrete measurable space $X$, $!X^{\pcoh}$ is the limit in $\pcoh$ of the  draw-and-delete chain $DD^{!}_{X^{\pcoh}}$  generated by the free copointed object $(!X^{\pcoh})_\bullet = !X^{\pcoh} \with \unit$. Recall that $\pcoh \rightarrow \ICones$ is strong monoidal -- Theorem~\ref{thm:ic-full-faithful}. This embedding moreover preserves cartesian products, thus the image in $\ICones$ of $\DD^{!}_{X^\pcoh}$
coincides with the $\ICones$-DD-chain generated by $(!X^{\pcoh})^{\ICones} \with \unit$. The natural next step consists in looking at whether $(!X^{\pcoh})^{\ICones}$ is the limit of this chain, i.e. whether limits of approximations $\DD^!$-chains are preserved along $\pcoh \rightarrow \ICones$.

In this direction, we can first prove that all small limits are preserved along $\pcoh \rightarrow \Cones$.

  \begin{lemma}\label{lemma:limits_preserved_cones_pcoh}
 $c:\pcoh \hookrightarrow \Cones$ preserves all small limits.
\end{lemma}
\begin{proof}
  Since the embedding $\pcoh \rightarrow \Cones$ is fully faithful and sends the $\pcoh$ monoidal unit to the $\Icones$ monoidal unit, we can deduce from Lemma\ref{lemma:generic_on_icones_limits}  that all limits in $\pcoh$ are preserved in $\Cones$.
\end{proof}
It's however surprisingly difficult to extend this general limit preservation result to $\ICones$ -- it remains at this time an open question.
The approximations chain $\DD^{!}$ have however a sufficiently regular shape to avoid the difficulties that appear in the general case.

\begin{toappendix}
  \begin{lemmarep}\label{lemma:key_ic_bang}
    Let $X$ be a PCS.
The morphism $\idm \Cones{}: {!X^{\Cones}} \rightarrow {!X^{\Cones}}$ lifts to an $\ICones$ morphism $(!X^{\Cones})^{\triangleleft} \rightarrow !X^{\ICones}$.
\end{lemmarep}
\begin{proofsketch}
We need to show that $\idm{\Cones}{(!X)^{\Cones}}$  preserves measurable paths and integrals.
 For this, we are going to use the characterisation of measurable paths and integrals on $((!X)^{\Cones})^{\triangleleft}$ we established in Proposition~\ref{prop:limits_lifting}, as well as the concrete expression for $!X$ recalled in Proposition~\ref{prop:limit_morphisms}.
  \end{proofsketch}
\begin{proof}
  We need to show that $\idm{\Cones}{(!X)^{\Cones}}$  preserves measurable paths and integrals.
 For this, we are going to use the characterisation of measurable paths and integrals on $((!X)^{\Cones})^{\triangleleft}$ we established in Proposition~\ref{prop:limits_lifting}. Let $R\in \Arity$, $\gamma:R\rightarrow ((!X)^{\Cones})^{\triangleleft}$ a measurable path, $\mu \in \FMeas(R)$.
  \begin{itemize}
  \item  We want to show that $\gamma:R \rightarrow (!X)^{\ICones}$ is a measurable path, i.e. that for any $y \in (!X)^{\perp}$, $(r\in R \mapsto \scal {\gamma(r)} {y} \in \RR_{\geq 0})$ is measurable. Let $x \in (!X)^{\perp}$.
Using the concrete expressions of the $\rho_{\infty,n}$ given in Proposition~\ref{prop:limit_morphisms}, we see that:
  \begin{align*}
    \scal{\gamma(r)}{y} = \sum_{a\in \web{!\pcs A}} \gamma(r)_a\cdot y_a &= \sup_{n\in \NN}\sum_{a\in \web{\Mn n {\pcs A}}} \gamma(r)_a\cdot y_a \\ &=\sup_{n\in \NN}\sum_{a\in \web{\Mn n {\pcs A}}} y_{\hat a} \cdot \rho_{\infty,n}(\gamma(r))_a),
  \end{align*}
  where $y_{\hat a}$ is the element of $\MfSet {\web{A}}$ obtained from  $y \in \MnSet n{\web{A} \cup \{\star\}}$ by removing all the $\star$.  By applying the monotone convergence theorem, we see from there that it is enough to prove measurability of the function $f_{n,a}:=(r \in R \mapsto \rho_{\infty,n}(\gamma(r))_a \in \RR_{\geq 0})$, for every $n \in \NN$, and $a \in \web{\Mn n {\pcs A}}$. We fix such a pair $n,a$. Using the technical conditions in Definition~\ref{def:pcoh} -- combined with closure of the set of elements by downward closure for the pointwise order on coefficients --,
  we see that there exists $\lambda>0$, with $\lambda.e_a \in \clique{\Mn n {\pcs A}^{\perp}}$. Thus $f_{n,a} = t\circ (\rho_{\infty,n}\circ\gamma)$, where $t$ is the measurability test $t:= (z \in {!X}^{\ICones} \mapsto \scal{z}{\lambda.e_a})$ on  ${{\Mn n {\pcs A}}}^{\ICones}$. From this observation, we can conclude that $f_{n,a}$ is measurable by our hypothesis on $\gamma$, and Proposition~\ref{prop:limits_lifting}.
  
\item Let us now look at integrals preservation.
  We want to show that $\cint_{R}^{((!X)^{\Cones})^{\triangleleft}}\gamma \cdot d\mu = \cint_R^{!X^{\ICones}} \gamma\cdot d\mu$. By Proposition~\ref{prop:limits_lifting}, it is enough to check that for every $n$, $\rho_{\infty,n} ( \cint_R^{!X^{\ICones}} \gamma\cdot d\mu ) = \int_R^{\Mn n X} \rho_{n,\infty} \circ \gamma\cdot d\mu$. But that's immediate since the $\rho_{n,\infty}$ are integrable.
  \end{itemize}
\end{proof}
\end{toappendix}
\begin{propositionrep}
Let $X$ be any PCS. The limit of $DD_X^!$  is preserved by $\ic:\pcoh \hookrightarrow \ICones$.
\end{propositionrep}
\begin{proofsketch}
  Recall that the approximations chain generated by a PCS $X$ is the DD-chain $\DD_{X}^{!}$ generated by the copointed object $X\rightarrow X\with \unit$, and its limit cone in $\pcoh$ is  $ ( !X \xrightarrow{\rho_{\infty,n}} \Mn n X)$. We are going to make use again of the fact -- Proposition~\ref{prop:limits_lifting} -- that $\ICones \rightarrow \Cones$ lifts limits. Since the limits of the approximations chain is preserved in $\Cones$, we can lift it from $\Cones$ to $\ICones$, and we obtain a limit cone there on $\ic(\DD_{X}^{!}):  ( !X^{\Cones})^{\triangleleft} \rightarrow{} \ic{(\Mn n X)})_{n \in \NN}$, with a description of $(!X^\Cones)^{\triangleleft}$ measurability behaviour given in  Proposition~\ref{prop:limits_lifting}. The non-automatic part, done in Lemma~\ref{lemma:key_ic_bang} in appendix, consists in proving that the $\Cones$-identity on $!X^{\Cones}$ becomes a monomorphism from $(!X^{\Cones})^{\triangleleft} $ to $\ic{(!X)}$. For this, we need to use the actual shape of the limit cone $(!X,\rho_{\infty,n}:!X\rightarrow \Mn n X)$ on the chain $\DD^{!}_X$, that we recalled in Proposition~\ref{prop:limit_morphisms}.
\end{proofsketch}
\begin{proof}
  It is immediate, by combining Lemma~\ref{lemma:key_ic_bang}, with the fact that $\idm{\Cones}{(!X)^{\Cones}}$ is an $\ICones$ morphism, because of the limit cone property of $(((!X)^{\Cones})^{\triangleleft} \rightarrow \Mn n X)_{n \in \NN}$ over the chain $\ic(\DD_{X}^!)$ -- the reasonning is similar to the one in the proof of Proposition~\ref{prop:iso_to_lifting}.
  \end{proof}

\subsubsection{Sending measures on $\Giry X$ into $!X^{\pcoh}$.}\label{section:from_measures_to_bang_bool}
As we proved in Proposition~\ref{prop:high-level_chain_morph_in_pcoh_bis}, there is a chain morphism in $\pcoh$ -- thus also in $\Icones$ -- from the De Finetti chain to the DD-chain of the layered construction. Our De Finetti Theorem~\ref{th:de_finetti_in_icones} for integrable cones tells us that $\Giry X^{\ICones}$ is a limit of the De Finetti DD-chain  in $\Icones$, and we can thus exploit this chain morphism to obtain a canonical $\Icones$ morphism between the respective limits of these DD-chains.
\begin{propositionrep}\label{prop:existence_iota}
  We can define a $\Icones$ monomorphism $\iota:{\Giry X}^{\ICones} \rightarrow (!X^{\pcoh})^{\ICones}$ by
  $$\iota(\mu) = \cint_{r \in \Giry X}^{(!X^{\pcoh})^{\ICones}} {(r \in \Giry X \mapsto (r^{\pcoh}) ^!)}\cdot \mu(dr), $$
  -- where $r^{\pcoh}$ is the element in $\clique {X^{\pcoh}}$ defined as $(r^{\pcoh})_x:=r(x)$, and $(r^\pcoh)^{!}$ is its promotion in $\clique{!X^{\pcoh}} \subseteq (!X^{\pcoh})^{\ICones}$. 
  Moreover $\iota$ is the unique morphism such that:
{\small $$ \begin{tikzpicture}[scale=0.8]
\node (0d) at (0,0) {$\unit$};
\node (1d) at (3.5,0) {$\Mn 1 {X^{\ICones}}$};
\node (2d) at (8.5,0) {$\Mn 2 {X^{\ICones}}$};
\node (3d) at (11,0) {$\dots$};
\node (limitd) at (12.5,1) {$\Giry{X}^{\ICones}$};
\draw[->] (limitd) to node[left] {} (0d);
\draw[->] (limitd) to node[above] {} (1d);
\draw[->] (limitd) to node[right] {} (2d);
\node (0) at (0,-1.5) {$\unit$};
\node (1) at (3.5,-1.5) {$\Mn 1 {X^{\ICones}\with 1}$};
\node (2) at (8.5,-1.5) {$\Mn 2 {X^{\ICones} \with 1}$};
\node (3) at (11,-1.5) {$\dots$};
\node (limit) at (12.5,-2.2) {$(\oc{X^{\pcoh}})^{\ICones}$};
\draw[->] (3d) to node[below] {$ \DDn 2^{\text{DF}}$} (2d); 
\draw[->] (2d) to node[below] {$ \DDn 1^{\text{DF}}$} (1d);
\draw[->] (1d) to node[below] {$\DDn 0^{\text{DF}}$} (0d);
\draw[->] (limit) to node[right] {} (0);
\draw[->] (limit) to node[right] {} (1);
\draw[->] (limit) to node[right] {} (2);
\draw[->] (3) to node[above] {} (2);
\draw[->] (2) to node[above] {} (1);
\draw[->] (1) to node[above] {} (0);

\draw[->] (2d) to node[above] {} (2);
\draw[->] (1d) to node[above] {} (1); 
\draw[->] (0d) to node[above] {} (0);
\draw[->,dashed] (limitd) to node[left] {$\exists ! \iota$} (limit);
\end{tikzpicture}$$}
\end{propositionrep}

\begin{proofsketch}
  Theorem~\ref{th:de_finetti_in_icones} gives us the existence of $\iota$, and then we use the commuting diagram to find its expression. For this, we use the fact that we have access to concrete expressions for
  the limit morphisms $\multkern n: \Giry X \rightarrow  \Mn n X$ that have been made explicit in~\cite{jacobs2020finetti}; and the limit morphisms $\rho_{\infty,n}:!X^{\pcoh} \rightarrow \Mn n {X^{\pcoh} \with \unit}$ in Proposition~\ref{prop:limit_morphisms}. 
\end{proofsketch}

\begin{example}\label{ex:bernouilli_sampler}

  Let us first consider a Dirac probability measure $\dirac r$, with $r \in [0,1]$.  From the viewpoint of exchangeable sequence of random variables -- via De Finetti Theorem--it corresponds to an i.i.d. sequence of Bernouilli distributions of parameter $r$. Then $\iota({\dirac r})$ is the \emph{promotion} of the Boolean program that is $\texttt{true}$ with probability $r$ and $\texttt{false}$ with probability $(1-r)$.  
  
  Suppose now that we start from any sub-probability measure $\mu$ on $[0,1]$.  From the viewpoint of exchangeable sequence of random variables, it corresponds to a mixture along $\mu$ of i.i.d. sequences of Bernouilli distributions. We can understand $\iota(\mu)$ as the interpretation of a program that first samples a $r$ in $[0,1]$ according to $\mu$, and then each time a new boolean is asked of it, returns true with probability $r$ and $(1-r)$ otherwise. Such a program could be written for instance in a variant of the call-by-push-value language with $I = [0,1]$ as ground type\footnote{A discrete probabilistic variant of call-by-push-value has been interpreted in $\pcoh$ in~\cite{lmcs:1537}. A natural extension would be a model of probabilistic call-by-push-value with continuous data types (e.g.$[0,1]$) in $\Icones$. It has not been done formally yet, to the best of our knowledge.}:
  $M := \text{ let }x = \text{sample}(\mu) \text{ in } !x .$
\end{example}

Not \emph{all} cliques in e.g. $!\Bool$ can be seen as the image by $\iota$ of some probability measure on $[0,1]$: to build one counter-example, we can use the closure of the $!\Bool$ cliques by downward closure, as illustrated below.
\begin{example}\label{ex:non_total_elements}
  Let's look now at a variation $M'$ of the program $M$ of Example~\ref{ex:bernouilli_sampler}, where 
    at each call the program refuses to answer with probability $0< p < 1$:  this could be written as
    $M' := \text{ let }x = \text{sample}(\mu) \text{ in } !(x \oplus^{p} \Omega) ,$
    where $\oplus^p$ represent a $p$-biased probabilistic choice, and $\Omega$ a non-terminating program. We expect the interpretation of $M'$ to be $(\sem M')_m:= (p^{\card m}\cdot \sem{M})_{m}$ for any multiset $m$ over $\{\ttrue,\ffalse\}$, but this  $\sem{M'}$ cannot be written as a mixture of promotions. 
  \end{example}
  Example~\ref{ex:non_total_elements} illustrates the existence of computationally meaningful elements of $!X^{\pcoh}$ that cannot be built from a finite measure on $\Giry X$, because they have some kind of divergent behaviour. We could now ask whether all \emph{non-divergent} elements in $!X^{\pcoh}$ are the image of a probability measure on $\Giry X$. To do that, the first step is to precise what is a non-divergent element of $!X^{\pcoh}$: it will means a \emph{total} element, in the sense of the probabilistic coherence spaces with totality from~\cite{ehrhard2025variable}.
  \subsection{Totality in probabilistic coherence spaces}
  \begin{toappendix}
 \subsection{Totality in probabilistic coherence spaces}
    \end{toappendix}
    The model of probabilistic coherence spaces \emph{with totality} was introduced by Ehrhard et al~\cite{ehrhard_mpri_lecture_notes,ehrhard2025variable} as a refinement of $\pcoh$. The first step consists in introducing a \emph{strict} orthogonality relation, which strengthens the one from Notation~\ref{notation:weak_orthogonality}: for $X$  a PCS, and $U \subseteq \clique X$, we note 
    $U^{\Perp} = \{u \in (\clique X)^{\perp} \mid \scal x u = 1\},$ where $\scal {\,}{\,}$ is the scalar product introduced in Notation~\ref{notation:weak_orthogonality}, and $\clique{X}^{\perp}$ is the dual of $\clique X$ for the (lax) orthogonality  from Notation~\ref{notation:weak_orthogonality}.
  % \begin{notation}
  %   Let $X$ a PCS, and $U \subseteq PX$. We note:
  %   $U^{\Perp} = \{u \in P(X^{\perp}) \mid \scal x u = 1\}.$
  % \end{notation}
    \begin{definition}
      A probabilistic coherence space with totality (PCST) is a pair $(X, T (X))$ where $X$ is a PCS and $T (X) \subseteq \clique X$ satisfies $T (X)^{\Perp \Perp} = T (X)$.
  \end{definition}
      \begin{example}\label{ex:pcst}
We can equip the PCS $\Bool$ with the structure of a PCST by taking $T(\Bool) := \{x \mid x_\ttrue + x_\ffalse = 1\}$. The total elements of $\Bool$ are thus the proper probability distributions on $\deux$.
      \end{example}
  $T(X)$ is called the set of \emph{total} elements, and is designed to contain only those elements  that never induce additional non-termination, no matter in which context we put them. We note $\pCoht$ for the category where objects are PCST, and morphisms are those PCS morphisms that preserve the total elements.  In~\cite{ehrhard_mpri_lecture_notes}, a refinement of the linear logic structure of $\pcoh$ is built, in such a way that the forgetful functor  $\pCoht \rightarrow \pcoh$ preserve the structure of LL model (monoidal product, cartesian products, exponential \ldots).  We recall below the definition of $\oc$ in $\pCoht$: the set of total elements in $!\pcs A$ is generated -- by double orthogonality closure -- from the \emph{promotions} of total elements in $\pcs A$:
    \begin{definition}[Exponential with Totality]\label{def:exp_with_totality}
For $\pcs X = (\pcs A,T(\pcs A))$ a PCST, $!\pcs X = (!\pcs A,\{z^{!} \mid z \in T(\pcs A)\}^{\Perp\Perp})$.
\end{definition}
\begin{toappendix}
    \begin{lemma}[From~\cite{ehrhard_mpri_lecture_notes}]\label{lemma:totality_exponential_functions}
      Let $X,Y$ be a PCSTs. Then $f \in \pCoht(!X,Y)$ if and only if for every $x \in TX$, $f.x^{!} \in TY$.

    \end{lemma}
    \end{toappendix}
\subsection{Integrable Cone of Total Elements}
From any PCST $\pcs X^T = (\pcs X,T\pcs X)$, we define its \emph{cone of total elements} $\tc{\pcs X^T}$, that lives in the category $\Icones$. This cone consists of the smallest sub-cone of $\ic ({\pcs X})$ that contains all the (images by $\ic$) of the total elements in $\pcs X$ , and inherits its structure of integrable cone from $\ic{(\pcs X )}$. An important observation here is that for a PCST $\pcs X^T$, its cone of total elements \emph{can possibly be outside} the full subcategory $\ic(\pcoh)$ of $\ICones$, even though it is a sub-cone of a PCS. That is because cones in $\ic(\pcoh)$ have a very peculiar shape -- notably they have a \emph{countable basis} compatible with the order, the norm, and that also determines the measurability paths -- which is not necessarily inherited by all their sub-cones.

    \begin{propositionrep}[Totality Cone of a PCST]\label{prop:totality_cone}
      Let $\pcs{X}^{T} = (\pcs X, T {\pcs X})$ be a PCST. We build an object $\tc {\pcs{X}^T}$ in $\ICones$ as the sub-cone of ${(\pcs X)}^{\ICones}$ with underlying set:
      $\{\lambda \cdot u \mid u \in T\pcs{X},\lambda \in \RR_{\geq 0}\}. $ -- and that inherits all algebraic and  measurability structures of  ${(\pcs X)}^{\ICones}$.     We note $\iotatc: \tc{\pcs X}^T \rightarrow \pcs X^{\ICones}$ for the inclusion morphism in $\ICones$.
    \end{propositionrep}
    \begin{proofsketch}
We need to check that $\tc {\pcs{X}^T}$, equipped with the cone structure and measurability tests inherited from $\ic {\pcs X}$ is indeed an integrable cone. The proof can be found in the appendix. 
      \end{proofsketch}
    \begin{proof}
      We need to check that $\tc {\pcs{X}^T}$, equipped with the cone structure and measurability tests inherited from $\ic {\pcs X}$ is indeed a cone.  Observe first that
       the valid paths $R\rightarrow \tc{\pcs{X}^{T}}$ are exactly the valid paths $R \rightarrow \ic{\pcs X}$ whose image is contained in $\tc{\pcs{X}^T}$
       The only non-trivial things to check are:
      \begin{enumerate}
      \item that any non-decreasing sequence of elements $(x_n)_{n \in \NN}$ in $\ball{\tc{\pcs X^T}}$ has a supremum in $\ball{\tc{{\pcs X}^T}}$. Let us fix $(x_n)$ such a non-decreasing sequence: for every $n$, there exists $\lambda_n\geq 0, y_n\in T{\pcs X}$ such that $x_n = \lambda_n\cdot  y_n$. Suppose that all the $\lambda_n=0$, then it implies that all the $x_n=0$, and that ends the proof. Now, suppose otherwise: we  define $\lambda:=\sup{\lambda_n}$. Since since the order and the norm are inherited from $\ic{\pcs X}$, there exists $y \in \ball{\ic{\pcs X}} = \clique {\pcs{X}}$ such that $y = \sup(x_n)_{n\in \NN}$ in $\ic{\pcs X}$. We are going to check that $\frac 1 \lambda \cdot y$ is a total element in $\pcs X$. Since $TX$ is closed by bi-orthogonality for  $\Perp$, it is enough to check that for every $z\in T{\pcs X}^{\Perp}$, $\scal z {\frac 1 \lambda \cdot y} = 1$:
        \begin{align*}
          \scal z {\frac 1 {\lambda}\cdot y} & = \frac 1 {\lambda}\cdot \sup_{n \rightarrow +\infty } \scal z {x_n} 
                                              = \frac 1 \lambda \cdot \sup_{n \rightarrow +\infty } \lambda_n\cdot \scal z {y_n} 
          \\ &=  \frac 1 \lambda \cdot \sup_{n \rightarrow +\infty } \lambda_n \,(\text{ since }y_n\in T{\pcs X})\quad   = 1
          \end{align*}
        
        \item that for any measurable path $\gamma:R\rightarrow \ic{\pcs X}$ such that $\gamma(R) \subseteq \tc{\pcs{X}^T}$, and any measure $\mu \in \FMeas(R)$, $\cint^{\ic{\pcs X}}_R \gamma\cdot d\mu \in \tc{\pcs{X}^T} .$ Without loss of generality, we can suppose that $\gamma$ is $1$-bounded. By hypothesis on $\gamma$, there exists two maps $\lambda: R \rightarrow \RR_{\geq 0}$, and $\delta:R\rightarrow T\pcs{X}$ such that $\forall r, \gamma(r) = \lambda(r)\cdot \delta(r)$. From there, we can write, for every $z\in T{\pcs X}^{\Perp}$:
          \begin{align*}
            &\scal z {\cint^{\ic{\pcs X}}_R \gamma\cdot d\mu} = \cint ^{\ic{\pcs X}}_R \scal \gamma z \cdot d\mu \\ & \text{ since } (\_ \in \tc{\pcs X^T}\mapsto \scal z {\_} \in \RR_{\geq 0}) \text{ is a meas. test on }\ic{\pcs X} \\
                                                             &= \cint ^{\ic{\pcs X}}_R (r\mapsto \lambda(r)\cdot \scal {\delta(r)} z) \cdot d\mu \\
            &= \cint ^{\ic{\pcs X}}_R (r\mapsto \lambda(r)\cdot 1) \cdot d\mu \text{ since }z \in T\pcs{X}^{\Perp}
          \end{align*}
          Let us define $\alpha:=  \cint ^{\ic{\pcs X}}_R (r\mapsto \lambda(r)) \cdot d\mu$ -- observe that this quantity does not depend on $z$. If $\alpha>0$, we can conclude from the equalities above that $\frac 1 {\alpha} \cdot \cint^{\ic{\pcs X}}_R \gamma\cdot d\mu \in T{\pcs X}$, thus  $\cint^{\ic{\pcs X}}_R \gamma\cdot d\mu \in \tc{\pcs X^T}.$ Suppose now that $\alpha = 0$. Then:
          \begin{align*} \norm{\cint^{\ic{\pcs X}}_R \gamma\cdot d\mu}_{\ic{\pcs X}}  & = \sup_{t\text{ meas. test }} t({\cint^{\ic{\pcs X}}_R \gamma\cdot d\mu}) \\ &= \sup_{x \in \clique{\pcs X}^{\perp}} \scal x {\cint^{\ic{\pcs X}}_R \gamma\cdot d\mu} \\ 
                                                                                %      & = \sup_{x \in \clique{\pcs X}^{\perp}}  {\cint^{\ic{\pcs X}}_R \scal \gamma x\cdot d\mu} \\
                                                                                      & = \sup_{x \in \clique{\pcs X}^{\perp}}  {\cint^{\ic{\pcs X}}_R \lambda \cdot \scal \delta x\cdot d\mu}
          \end{align*}
          Since $\forall r, \delta(r)\in \clique{\pcs X}$, then $\forall x \in \clique{\pcs X}$, $ \scal{\delta(r)}{x} \leq 1$. Using this, we can deduce from the equalities above: $\norm{\cint^{\ic{\pcs X}}_R \gamma\cdot d\mu}_{\ic{\pcs X}} \leq  \alpha$. It means -- using the $\Cones$ axioms -- that whenever $\alpha = 0$, also $\cint^{\ic{\pcs X}}_R \gamma\cdot d\mu =0$, and that ends the proof.                               
          
        \end{enumerate}
      \end{proof}
     
      It can be shown that for a measurable path $\gamma:R\rightarrow \tc{\pcs{X}^{T}}$ such that $\gamma(R) \subseteq T{\pcs X}$, and for $\mu$ a probability measure,  $\cint_R^{\tc{\pcs X}^{T}}\gamma\cdot d\mu \in T{\pcs X}$. It means that in the particular case where we know that $\gamma: R\rightarrow \tc{\pcs X^T}$ reaches only \emph{proper} total elements in $T{\pcs X}$,  taking the integral of $\gamma$ over a measure $\mu$ into the cone of total elements should be understood as taking a \emph{continuous convex combination}, or \emph{mixture} of total elements.       \begin{example}
        The totality cone of the PCST $\Bool$ of  Example~\ref{ex:pcst}  is simply $(\Bool)^{\ICones}$. It is however not always the case that the totality cone of a PCST $\pcs X^{T} = (\pcs X,T\pcs{X})$ coincides with ${(\pcs X)}^{\ICones}$: that's because not all elements $x \in \pcs X$ can necessarily be written as a total element scaled by some factor.  In particular, we will see later that it is not the case for $!\Bool$. %Moreover, it is a priori not even the case that $\tc{\pcs{X^T}}$ is isomorphic in $\ICones$ to the image by $\ic$ of \emph{any} PCS.
\end{example}

  \begin{propositionrep}\label{lemma:iota_total_elements}
    Let $X$ a countable discrete space.

There exists a unique morphism $\iotares:{\Giry X}^{\ICones} \rightarrow \tc{(!X^{\pcoh})}$ such that $\iotatc \circ \iotares = \iota$, where $\iota:  \Giry X^{\pcoh} \rightarrow (!X^{\pcoh})^{\Icones}$ is the morphism we built in Proposition~\ref{prop:existence_iota}. 
\end{propositionrep}
\begin{proofsketch}
It is enough to prove that $\iota$  sends any probability measure $\mu \in (\Giry X)^{\ICones}$ to a total element in $\clique{!X^{\pcoh}}$. This can be deduced from  $\iota(\mu) = \cint_{r \in \Giry X}^{(!X^{\pcoh})^{\ICones}} (r \in \Giry X \mapsto (r^{\pcoh})^!)\cdot \mu(dr)$ -- see  Proposition~\ref{prop:existence_iota}.
  \end{proofsketch}
\begin{proof}
  It is enough to prove that $\iota$  sends any finite measure $\mu \in (\Giry X)^{\ICones}$ to an element in the totality sub-cone of ${!(X^{\pcoh})}$.
  Let $\distrone \in \FMeas{(\Giry X)}$. Suppose first that $\distrone$ is a  probability measure.
  Recall that by Proposition~\ref{prop:existence_iota}, $\iota(\mu) = \cint_{r \in [0,1]}^{\ic(!X^{\pcoh})} \ic(r^{\pcoh}  )^!)\cdot \mu(dr).$ Thus $\iota(\distrone)$ is a mixture of promotions of elements in $!X^{\pcoh}$. Moreover, those elements in $!X^{\pcoh}$ whose we take the promotions are always \emph{total elements}, which means that by definition of $!$ in $\pCoht$ those promotions themselves are total. Since, as highlighted before, a mixture along a probability distribution of total elements in a PCS is always total, $\iota(\distrone)$ is  a total element in $!\Bool$. Now, if $\distrone$ isn't a probability measure, then it is either $0$ or a $\lambda \cdot \distrtwo$, with $\distrtwo$ a probability measure, and $\iota(\distrone) = \lambda\cdot \iota(\distrtwo)$ is in the totality cone of $\tc{!X^{\pcoh}}$.  
  \end{proof}
    \subsection{Total elements in $!\Bool$ are mixtures of promotions}\label{sect:total_elem_mixtures_promotions}

We just proved in Proposition~\ref{lemma:iota_total_elements} that for $X$ a countable discrete space, probability measures in $\Giry X$ -- or equivalently exchangeable sequences of $X$-valued random variable -- are sent by $\iota$ to total elements of $!X^{\pcoh}$. As highlighted before, all total promotions are reached this way; more precisely they are the image of Dirac measures. More generally, any (continuous) convex combination of promotions --
i.e. $\cint_{\Giry X}^{\ic{!X}}( r\mapsto (r^{\pcoh})^!) \cdot d\mu$ --  is reached by $\iota$.
We are now going to conclude our work by proving that we caracterise this way \emph{all} total elements of $!X^{\pcoh}$.

Our first step in this direction is a combinatorial lemma on total elements $y$ of $!X^{\pcoh}$: the valuation $y_{\mu}$ for a multiset $\mu$ of size $n$ can be deduced of the valuations of $y$ for multisets of size $n-1$:
\begin{lemmarep}\label{lemma:key_technical_lemma_for_totality}
   Let $u \in T{!X^\pcoh}$. Then for every multiset over $X$, it holds that $u_{\mu} = \sum_{x\in X}u_{(\mu +[x])}$.
\end{lemmarep}
\begin{proof}
  We use the concrete presentation of $!X^\pcoh$ from Proposition~\ref{prop:limit_morphisms}, with its totality structure from Definition~\ref{def:exp_with_totality}.
 
  We fix $\mu_0 \in \web{!X^{\pcoh}}$
  \begin{enumerate}
     \item First, we show that there exists at least a total morphism $f:!X^\pcoh \rightarrow 1$ in $\pCoht$ such that $f_{(\mu_0,\star)} >0$, and $f_{(\nu,\star)} = 0$ when $\card \nu > \card {\mu_0}$. To see that, it is enough to take any enumeration $(a_1,...,a_n)$ of $\mu_0$, and to consider the program:
         $$ P:= \text{if}(x \neq a_1)(\star)( \text{ if }(x \neq a_2)(\star)(\text{if }....))).$$
         $P$ lives in an extension of discrete probabilistic $\PCF$ -- see~\cite{EhrhardPT18} -- where moreover we have a date type for $X$. It is not hard to see that the semantics~\cite{EhrhardPT18} of probabilistic $\PCF$ in $\pcoh$ can be extended accordingly (e.g.~si $X = \NN$, no adjustment is necessary).
         
    We can see easily that $\sem P\in \pCoh(!X^{\pcoh} \multimap 1)$ is a total function (for instance, by testing it against all the total promotions, and using Lemma~\ref{lemma:totality_exponential_functions}), and that $\sem P_{(\mu_0,\star)} >0$.
  \item Now, we build $f':!X^\pcoh \rightarrow 1$ in $\pCoh$ as:
    $$f'_{\nu,\star} = \begin{cases} f'_{\nu,\star} \text{ if }\card \nu \leq \card {\mu_0}, \mu_0 \neq \nu \\
                         f_{\mu_0,\star} \text{ if } \exists x \in X, \nu = \mu_0+[x]  \\
                         0 \text{ otherwise.}
                       \end{cases} $$
   Now, we can see that $f'$ also is a total function: again by Lemma~\ref{lemma:totality_exponential_functions}, it is enough to see that for every $x \in TX$, $f'(x^!) = 1 = f(x^{!})$. To see that, observe that  $f'(x^!) = f(x^!) + f_{\mu_0} \cdot ((\sum_{x \in X} x^!_{\mu+[x]}) - x^{!}_{\mu_0})$, and we can conclude since $(\sum_{x \in X} x^!_{\mu+[x]})  - x^{!}_\mu = (\sum_{x\in X}x_\mu - 1) \cdot x^{!}_{\mu} $ (by definition of the promotion operator $x^!$), and$ \sum_{x\in X}x_\mu= 1$ since $x \in T \Bool$.
    \item Since both $f$ and  $f'$ are total functions, and $u$ is a total element, we have $f'(u) = f(u) = 1$. We can compute (similarly as above) $f(u) - f'(u)$, and we obtain: $0 = f(u) - f'(u) = f_{\mu_0} \cdot ((\sum_{x \in X} u_{\mu_0+[x]})- u_{\mu_0})$, and since $f_{\mu_0} >0$, it implies that $\sum_{x \in X} u_{\mu_0 + [x]} = u_{\mu_0}$. 
    \end{enumerate}
 
  \end{proof}
  Let us now define the family of morphisms
  $$f_n:= \Mn n {\pi_1}: (\Mn n {X^{\ICones}} \rightarrow \Mn n {X^{\ICones} \with \unit})_{n \in \NN} $$
  where $\pi_1$ is the left projection morphism $X^{\ICones} \with \unit \rightarrow X^{\ICones}$.
As $\pi_1$ is not a copointed object morphism from $(X^{\ICones},(\_)^{\ICones}:X^{\ICones}\rightarrow \unit)$ to $(X^{\ICones})_\bullet = (X^{\ICones}\with \unit,\pi_2)$, Theorem~\ref{th:chain_morphism_from_copointed_morphisms} does not apply, and indeed the $f_n$ \emph{are not} a chain morphism from the De Finetti DD-chain to the approximants DD-chain. However, we can recover a cone from $\tc{!X^\pcoh}$ to the $\ICones$-De Finetti DD-chain, by precomposing the $(f_n)$ by the morphisms $ \rho_{\infty,n} \circ \iotatc : \tc{(!X^{\pcoh})} \rightarrow   \Mn n {X^{\ICones}}$, that sends a total element of $!X^\pcoh$ first to the larger cone of all $!X^\pcoh$ elements, and then to each of the $\Mn n {X^\ICones}$ by the limit morphisms $\rho_{\infty,n}! {(!X^{\pcoh})}^{\ICones} \rightarrow \Mn n {X^\ICones}$.

  \begin{lemmarep}\label{prop:applying_de_finetti_to_totality}
    The family $(m_n)_{n \in \NN}:= f_n \circ {\rho_{\infty,n}} \circ {\iotatc}: \tc{(!X^\pcoh)} \rightarrow   \Mn n {{X}^{\ICones}} $ forms a cone on the De Finetti draw-and-delete chain in $\Icones$.
  \end{lemmarep}
  \begin{proof}
    Since $\ic$ is (symmetric) strong monoidal, it is enough to show in $\pcoh$ that for any \emph{total element} in $\clique{! X^\pcoh}$,  $( \Mn n {\pi_1}\circ \rho_{n,\infty}(x)) = (\DDn n^{DF}\circ  \Mn n {\pi}\circ \rho_{n+1,\infty}(x))$, where we note $\DD^{DF}$ the draw and delete chain in $\pcoh$ generated by the copointed object ${(\_)}^{\pcoh}: X^\pcoh \rightarrow \unit$. Here, we use the concrete choice of PCS for the equalisers $\Mn n {\cdot}$ and the limit $!^\pcoh X$ presented in Propositions~\ref{def:mn_pcoh} and~\ref{prop:limit_morphisms}.
     For any element $x \in \clique{! X^\pcoh}$, $\Mn n {\pi}\circ \rho_{n,\infty}(x)$ is the element in $\clique{\Mn n {X^\pcoh}}$ obtained from $x$ by keeping exactly the multisets of size exactly $n$. As a consequence, for $\mu$ a multiset of size $n$ over $X$,  $(  \Mn n {\pi}\circ \rho_{n,\infty}(x))_\mu = x_\mu$, while  $(\DDn n^{DF}\circ  \Mn n {\pi}\circ \rho_{n+1,\infty}(x))_\mu = \sum_{a \in X} x_{\mu + [a]}$. We can conclude by our combinatorial Lemma~\ref{lemma:key_technical_lemma_for_totality}.
   \end{proof}

Since by our De Finetti theorem in $\ICones$ -- Theorem~\ref{th:de_finetti_in_icones}-- $\Giry X^{\ICones}$ is a limit of the $\ICones$ De Finetti DD-chain, we can deduce from Lemma~\ref{prop:applying_de_finetti_to_totality} a unique mediating morphism  $\alpha: \tc{ !X^{\pcoh}} \rightarrow \Giry X^\ICones$.
    
  \begin{theoremrep}\label{prop:iso_element_total_meas}
  
For any $X$ countable discrete space, $\tc{! X^\pcoh}$ and $\Giry X^{\ICones}$ are isomorphic, and thus the total elements of $!X^\pcoh$ are \emph{exactly} the continuous mixtures of promotions.
\end{theoremrep}
\begin{proofsketch}
The proof consists in showing that $\alpha:\tc{!X^{\pcoh}} \rightarrow {\Giry X^\pcoh}$ and the morphism  $\iotares:  {\Giry X}^{\Icones} \rightarrow \tc{!X^{\pcoh}}$ from Lemma~\ref{lemma:iota_total_elements} are inverse of one another. 
  \end{proofsketch}
\begin{proof}
  We are going to show that $\alpha:\tc{!X^{\pcoh}} \rightarrow {\Giry X^\pcoh}$ given by the cone in Lemma~\ref{prop:applying_de_finetti_to_totality} and $ \iotares:\Giry X^\pcoh  \xrightarrow{\iotatc}\tc{!X^\pcoh}$ are inverse one another.
  \begin{itemize}
  \item We first show that $\tc{!X^{\pcoh}} \xrightarrow{\alpha} \Giry X^{\ICones} \xrightarrow{\iota} !X^\pcoh$ is the mediating morphism that factorises the cone $ \tc{!X^\pcoh} \xrightarrow{\rho_{\infty,n} \circ \iotatc} \Mn n {( X^{\ICones}\with \unit)}$ on the approximants DD-chain through  $!X^\pcoh$. To check this, it is enough to see that $\forall n \in \NN$, $\rho_{n,\infty}\circ \iota \circ \alpha  = \rho_{\infty,n} \circ \iotatc$. Recall that by definition of $\iota$ -- in Proposition~\ref{prop:existence_iota} --as a mediating morphism itself, $\rho_{\infty,n} \circ \iota = \Mn n {(\langle X^\ICones,\_\rangle)}\circ \multkern n$. Thus $\rho_{n,\infty}\circ  \iota \circ \alpha =  \Mn n {(\langle X^\ICones,\_\rangle)}\circ \multkern n \circ \alpha$. We are now ready to use the definition of $\alpha$ as a mediating morphism: it tells us that  $\multkern n \circ \alpha = {\Mn n \pi_1} \circ \rho_{n,\infty}\circ \iotatc$. So summing up, we proved: $\rho_{n,\infty}\circ  \iota \circ \alpha =  \Mn n {(\langle X^\ICones,\_\rangle)} \circ {\Mn n \pi_1} \circ \rho_{n,\infty}\circ \iotatc$. Observe that $\langle X^\ICones ,\_ \rangle \circ {\pi_1} = \id_{X^{\ICones}}$, and from there we can deduce $\Mn n {(\langle X^{\ICones},\_\rangle)} \circ {\Mn n \pi_1} = \id_{\Mn n {X^\ICones \with \unit}}$. We can conclude from here. 
    
  \item Another morphism that factorises the cone $ \tc{!X^\pcoh} \xrightarrow{\rho_{\infty,n} \circ \iotatc} \Mn n {( X^{\ICones}\with \unit)}$ on the approximants DD-chain through  $!X^\pcoh$ is simply $\tc{!X^\pcoh} \xrightarrow{\iotatc} !X^{\pcoh}$. By the universal property of the limit, it means that  $\iotatc = {\iota}\circ \alpha $. Since $\iota =\iotatc  \circ \iotares$ (by Lemma~\ref{lemma:iota_total_elements})  and $\iotatc$ is obviously a monomorphism, we can deduce that $\iotares \circ \alpha = \id_{\tc{!X^\pcoh}}$.
    \item To conclude from the above that $\iotares \circ \alpha$, it is enough to prove that $\iotares$ is a monomorphism :indeed, $(\iotares \text{ mono }\wedge\iotares \circ \alpha = \id_{\tc{!X^\pcoh}}) \Rightarrow (\alpha \text{ mono }\wedge  \iotares \circ \alpha \circ \iotares= \id \circ \iotares) \Rightarrow \alpha \circ \iotares = \id.$ It is ieasy to check, because since $(\Giry X \xrightarrow{\multkern n} \Mn n{X})_{n \in \NN}$ is a limit cone, the $(\multkern n)_{n \in \NN}$ are jointly monic, and moreover the morphisms $\Mn n X \xrightarrow{\langle X,\_ \rangle} \Mn n {X\with \unit}$ are monos. 
   \end{itemize}

    \end{proof}

  \section{Future Work}
  The natural next step would be to try to prove that the exponential comonad $!$ in $\Icones$ can indeed be obtained by Melliès et al.'s layered construction of the free exponential. That would allow to connect $\Giry X$ and $!X$ for any $X$ in the category $\Arity$, instead of only for countable discrete spaces. Another direction would be to extend this connection to other flavours of De Finetti theorems, e.g.~to establish bridges between the quantum categorical De Finetti theorem from~\cite{staton2023quantum,Fritz:2023jou} and LL models for quantum computations~\cite{pagani_selinger}. Again another direction would be to explore what flavours of \emph{continuous probabilistic calculus} can be interpreted into $\pcoh$, starting from the fact that we now know that the datatype $[0,1] = \Giry \deux$ can be interpreted in $!\Bool$.

\bibliographystyle{plainurl}
\bibliography{biblio}

\end{document}